\documentclass[11pt]{article}
\pdfoutput=1

\usepackage[shortlabels]{enumitem}
\usepackage{titling}

\usepackage[T1]{fontenc} 
\usepackage{palatino} 
\usepackage{mathpazo}
\usepackage{braket}
\usepackage{amsfonts}
\usepackage{amssymb}
\usepackage{amsmath}
\usepackage{mathrsfs} 
\usepackage{bbm}
\usepackage{latexsym}
\usepackage{amsthm}
\usepackage[usenames]{color}
\usepackage{etoolbox}
\usepackage{diagbox}
\usepackage{subcaption}
\usepackage{verbatim}
\usepackage{ragged2e, graphicx, pifont}
\usepackage{indentfirst}
\usepackage{algorithm}
\usepackage[noend]{algpseudocode}
\usepackage{tikz-cd}
\usetikzlibrary{decorations.pathmorphing}
\usetikzlibrary{arrows.meta}
\usepackage{mdframed}
\usepackage{mathtools} 

\usepackage{qcircuit} 

\usepackage{accents}

\usepackage{hyperref} 
\usepackage{xurl} 
\hypersetup{
breaklinks=true,
colorlinks = true,
citecolor= blue,
linkcolor= blue
}

\floatname{algorithm}{Scheme}

\newcommand\Algphase[1]{
\vspace*{-.5\baselineskip}\Statex\hspace*{\dimexpr-\algorithmicindent-2pt\relax}
\Statex\hspace*{-\algorithmicindent}\textbf{#1}
\vspace*{-.8\baselineskip}\Statex\hspace*{\dimexpr-\algorithmicindent-2pt\relax}
}

\newcommand{\alglinelabel}{%
  \addtocounter{ALG@line}{-1}
  \refstepcounter{ALG@line}
  \label
}

\makeatletter
\renewcommand\paragraph{\@startsection{paragraph}{4}{\z@}%
                                      {\parskip}
                                      {-1em}%
                                      {\normalfont\normalsize\bfseries}}
\makeatother

\usepackage{tabularx}

\hypersetup{pdfpagemode=UseNone}

\newtheorem{theorem}{Theorem}
\newtheorem{lemma}[theorem]{Lemma}

\newtheorem{cor}[theorem]{Corollary}
\newtheorem{definition}{Definition}

\newtheorem{remark}{Remark}

\newtheorem{fact}[theorem]{Fact}

\newcommand{\norm}[1]{\ensuremath{\left\lVert #1 \right\rVert}}

\newcommand{\eps}{\varepsilon}
\newcommand{\inner}[2]{\langle #1, #2 \rangle}
\newcommand{\matel}[3]{\langle #1 | #2 | #3\rangle}

\newcommand{\clA}{\mathcal{A}}
\newcommand{\clB}{\mathcal{B}}
\newcommand{\clC}{\mathcal{C}}

\newcommand{\clE}{\mathcal{E}}
\newcommand{\clF}{\mathcal{F}}

\newcommand{\clH}{\mathcal{H}}

\newcommand{\clM}{\mathcal{M}}
\newcommand{\clO}{\mathcal{O}}
\newcommand{\clP}{\mathcal{P}}

\newcommand{\clS}{\mathcal{S}}
\newcommand{\clT}{\mathcal{T}}

\newcommand{\clX}{\mathcal{X}}
\newcommand{\clY}{\mathcal{Y}}
\newcommand{\clZ}{\mathcal{Z}}

\newcommand{\sfB}{\mathsf{B}}
\newcommand{\sfD}{\mathsf{D}}
\newcommand{\sfF}{\mathsf{F}}
\newcommand{\sfH}{\mathsf{H}}
\newcommand{\sfI}{\mathsf{I}}
\newcommand{\sfP}{\mathsf{P}}

\newcommand{\sfV}{\mathsf{V}}

\newcommand{\Id}{\mathbbm{1}}
\newcommand{\bbF}{\mathbb{F}}
\newcommand{\bbN}{\mathbb{N}}

\newcommand{\bbC}{\mathbb{C}}

\newcommand{\bbE}{\mathop{\mathbb{E}}}

\newcommand{\A}{\mathrm{A}}
\newcommand{\B}{\mathrm{B}}
\newcommand{\C}{\mathrm{C}}

\newcommand{\tA}{\widetilde{A}}

\newcommand{\tS}{\widetilde{S}}
\newcommand{\tC}{\widetilde{C}}
\newcommand{\tE}{\widetilde{E}}

\newcommand{\tM}{\widetilde{M}}
\newcommand{\tX}{\widetilde{X}}
\newcommand{\tU}{\widetilde{U}}
\newcommand{\tY}{\widetilde{Y}}
\newcommand{\tZ}{\widetilde{Z}}

\newcommand{\hA}{\widehat{A}}
\newcommand{\hB}{\widehat{B}}
\newcommand{\hC}{\widehat{C}}
\newcommand{\hS}{\widehat{S}}
\newcommand{\oC}{\overline{C}}

\newcommand{\oT}{\overline{T}}

\newcommand{\vph}{\varphi}

\newcommand{\ketbra}[2]{\ket{#1}\!\!\bra{#2}}
\newcommand{\state}[1]{\ketbra{#1}{#1}}

\newcommand{\cmark}{\ding{51}\,}
\newcommand{\xmark}{\ding{55}\,}

\newcommand{\syn}{\mathtt{syn}}
\newcommand{\lsyn}{\ell_{\syn}}

\newcommand{\supp}{\operatorname{supp}}
\newcommand{\Tr}{\operatorname{Tr}}

\newcommand{\lEC}{p_\mathrm{err}}
\newcommand{\cl}{\mathrm{CLONE}}
\newcommand{\negl}{\mathrm{negl}}
\newcommand{\CHSH}{\mathrm{CHSH}}
\newcommand{\junk}{\mathrm{junk}}
\newcommand{\triv}{\mathrm{triv}}

\newcommand{\IP}{\mathrm{IP}}

\newcommand{\dga}{\delta_{\gamma,\alpha}}

\newcommand{\qhon}{q}

\newcommand{\Enc}{\mathsf{Enc}}
\newcommand{\KeyR}{\mathsf{KeyRel}}
\newcommand{\Dec}{\mathsf{Dec}}
\newcommand{\KeyG}{\mathsf{KeyGen}}
\newcommand{\kpriv}{k_\mathrm{priv}}
\newcommand{\Kpriv}{K_\mathrm{priv}}
\newcommand{\kdec}{k_\mathrm{dec}}
\newcommand{\Kdec}{K_\mathrm{dec}}
\newcommand{\kenc}{k_\mathrm{enc}}
\newcommand{\Kenc}{K_\mathrm{enc}}
\newcommand{\Prot}{\mathsf{Protect}}
\newcommand{\Eval}{\mathsf{Eval}}

\newcommand{\otp}{R}
\newcommand{\eA}{\widetilde{A}} 
\newcommand{\eX}{\widetilde{X}} 
\newcommand{\rS}{S'} 
\newcommand{\eS}{\widetilde{S}} 
\newcommand{\gA}{G} 
\newcommand{\rVec}{V} 
\newcommand{\rvec}{v} 

\definecolor{dgreen}{rgb}{0.0, 0.6, 0}

\allowdisplaybreaks

\title{Device-independent uncloneable encryption}
\author{Srijita Kundu \thanks{Institute for Quantum Computing and Department
of Combinatorics and Optimization, University of Waterloo, Waterloo, Ontario N2L 3G1, Canada. \tt{srijita.kundu@uwaterloo.ca}}
\and
Ernest Y.-Z. Tan \thanks{Institute for Quantum Computing and Department
of Physics and Astronomy, University of Waterloo, Waterloo, Ontario N2L 3G1, Canada. \tt{yzetan@uwaterloo.ca}}}

\begin{document}
\maketitle

\begin{abstract}
Uncloneable encryption, first introduced by Broadbent and Lord (TQC 2020) is a quantum encryption scheme in which a quantum ciphertext cannot be distributed between two non-communicating parties such that, given access to the decryption key, both parties cannot learn the underlying plaintext. In this work, we introduce a variant of uncloneable encryption in which several possible decryption keys can decrypt a particular encryption, and the security requirement is that two parties who receive independently generated decryption keys cannot both learn the underlying ciphertext. We show that this variant of uncloneable encryption can be achieved device-independently, i.e., without trusting the quantum states and measurements used in the scheme, and that this variant works just as well as the original definition in constructing quantum money. Moreover, we show that a simple modification of our scheme yields a single-decryptor encryption scheme, which was a related notion introduced by Georgiou and Zhandry. In particular, the resulting single-decryptor encryption scheme achieves device-independent security with respect to a standard definition of security against random plaintexts. Finally, we derive an ``extractor'' result for a two-adversary scenario, which in particular yields a single-decryptor encryption scheme for single bit-messages that achieves perfect anti-piracy security without needing the quantum random oracle model. 
\end{abstract}

\section{Introduction}

A fundamental difference between classical and quantum information lies in the fact that quantum information cannot be perfectly copied. This property can be used to do cryptography, as was noted by Wiesner \cite{Wie83}, who gave the first scheme for \emph{quantum money} which cannot be forged. Later \cite{Got03} considered the question of whether in the context of encryption schemes, one could construct a form of \emph{uncloneable encryption}; i.e.~quantum ciphertexts that in some sense cannot be copied. Pursuing this line of reasoning,~\cite{Got03} developed an encryption scheme in which an adversary attempting to copy a quantum ciphertext would be caught with high probability by the intended (honest) recipient. Following up on a question posed in that work,~\cite{BL20} subsequently constructed an encryption scheme achieving a slightly different notion of uncloneable encryption\footnote{\cite{BL20} refer to the scheme in~\cite{Got03} as one that achieves \emph{tamper-detection} rather than uncloneability; in this work we follow their terminology rather than the original terminology in~\cite{Got03}.}, namely a quantum ciphertext that cannot be distributed amongst two parties in such a way that they can both decrypt the message with high probability (after obtaining the decryption key).

An uncloneable encryption scheme needs to satisfy a standard notion of indistinguishability (or semantic security) that any encryption scheme needs to satisfy. Aside from this, \cite{BL20} introduced two notions of security that an uncloneable encryption scheme should satisfy: uncloneability and uncloneable-indistinguishability. The uncloneable encryption scheme given \cite{BL20} was a simple construction based on Wiesner states and monogamy of entanglement games. While this scheme achieved uncloneability in the plain model, it did not achieve uncloneable-indistinguishability, even in the quantum random oracle model (QROM). 

Following \cite{BL20}, there have been several subsequent works on uncloneable cryptography. \cite{MST21} showed that uncloneable-indistinguishability cannot be achieved by schemes using certain kinds of states, and other limitations of the proof techniques employed in \cite{BL20}. \cite{AK21} considered public-key uncloneable encryption; \cite{GMP22} gave a protocol for uncloneable encryption based on the post-quantum hardness of the learning with errors (LWE) problem. Recently, \cite{AKL+22} gave a more complicated uncloneable encryption protocol based on subset coset states that achieves uncloneable-indistinguishability in the QROM. Moreover, \cite{AKL+22} also gave some impossibility results showing that certain kinds of schemes cannot achieve uncloneable-indistinguishability in the plain model. The concept of uncloneable decryption keys, or single-decryptor encryption, was introduced by \cite{GZ20}, which takes a reversed perspective compared to uncloneable encryption: here a quantum decryption key is made uncloneable rather than a ciphertext, with the security requirement being only a single party should be able to use the decryption key to decrypt a classical ciphertext. \cite{GZ20} showed that single-decryptor encryption is equivalent to uncloneable encryption under a certain security definition, but other definitions of single-decryptor encryption have also been considered, see e.g. \cite{CLLZ21}.

As noted in \cite{BJL+21}, uncloneable encryption can be considered the second level in the hierarchy of uncloneable objects, since it makes \emph{information} uncloneable. The first level of the hierarchy only lets us verify the authenticity of objects: this is where private-key quantum money lies. At the top level of the hierarchy, \emph{functionalities} are made uncloneable: this includes quantum copy protection and secure software leasing. It would be natural to ask if higher levels of the hierarchy can be used to achieve lower levels. Indeed, \cite{BL20} showed that uncloneable encryption can be used to construct private-key quantum money. More surprisingly, it has been shown \cite{CMP20, AK21} that uncloneable encryption can be used to construct quantum copy protection of a certain class of functions, although these constructions require either the QROM \cite{CMP20} or additional computational assumptions \cite{AK21}.

\paragraph{Device-independence.} The results in~\cite{Got03,BL20} were derived in a \emph{device-dependent} setting, in which it is assumed that any honest parties can generate trusted states and/or perform trusted measurements. However, it was observed in e.g.~\cite{BHK05,AGM06,PAB+09} that in some situations, one can construct protocols that are secure under much weaker assumptions: no assumptions on the states and measurements are made except that the measurements of spatially separated parties are in tensor product (or commute). This strong form of security is referred to as the \emph{device-independent} (DI) paradigm, in that security can be achieved (almost) independently of the underlying operations being performed by the devices used in the protocol. Device-independent protocols that have information theoretic security are often based on the property of \emph{self-testing} or \emph{rigidity} displayed by some non-local games. Suppose a non-local game is played with some unknown state and measurements. If these measurements and state achieve a winning probability close to the optimal winning probability of the game, then self-testing tells us the state and measurements are close to the ideal state and measurements needed to achieve the optimal winning probability for that game.\footnote{In this work, for simplicity of presentation we use the self-testing results described in~\cite{MYS12} which have closed-form expressions; however, there would be no obstacles in principle if one were to instead use the more robust bounds computed in~\cite{BNS+15} using semidefinite programming. The latter approach can give bounds that are quite robust to non-maximal winning probabilities on the non-local game.} This means that we can do cryptography with this state and measurements as though they were the ideal state and measurements.

We note that in the specific context of uncloneable encryption, the security proofs in~\cite{Got03,BL20} do already have a form of ``one-sided device-independent'' property, in the sense that for the uncloneable encryption scenario the receiver may be dishonest, and hence the security proof must cover the possibility of the receiver not performing the intended operations. However, our goal in this work is to extend the device-independence to cover the client's devices as well (we briefly elaborate on how the~\cite{BL20} scheme is insecure in the fully DI setting in Section~\ref{subsec:results} below).
This is somewhat similar to the scenario considered by~\cite{GMP22} (for which they achieve polynomial rather than exponential security), except that in their scenario, while the states and measurements are indeed not trusted, it is still assumed that the devices are computationally bounded. Due to this, the security achieved in their scenario is not information theoretic, but under the assumption that the LWE problem cannot be solved by polynomial-time quantum computers. Hence thus far, there has not been an uncloneable encryption scheme in the ``standard'' fully DI scenario, without computational assumptions. 

\subsection{Our results}
\label{subsec:results}
In this work, we develop protocols for a modified version of uncloneable encryption (of classical messages), which we term \emph{uncloneable encryption with variable keys} (VKECM), as well as single-decryptor encryption --- see Remark~\ref{remark:defnchoice} below for comments regarding the security definitions used. Our main contribution is to show that these protocols fulfill a standard notion of DI security, i.e.~without computational assumptions, which has not been achieved in any previous work. (As mentioned above,~\cite{GMP22} is the only work achieving anything similar to DI security, but their result required computational assumptions; furthermore, our result is also stronger in the sense that it achieves exponential security rather than polynomial security as in that work.)
While VKECM differs from the usual notion of uncloneable encryption, we note in particular that it is at least still sufficient to yield a quantum money scheme that is secure in the standard fully DI setting, as we discuss later.

\begin{remark}\label{remark:defnchoice}
Currently, there are multiple competing security definitions in the literature for uncloneable encryption and single-decryptor encryption. We do not aim to provide a detailed comparison of the different definitions within the scope of this work, as they are fairly technical, though we discuss some of them briefly in Section~\ref{sec:secdefn} when specifying the definitions we chose to use. We highlight however that for uncloneable encryption, we are not studying the standard notion but are instead focusing on the modified version VKECM described below, and hence we have to introduce appropriate new security definitions for VKECM rather than the standard form of uncloneable encryption. Still, we note that they are similar in spirit to those used in the work~\cite{GMP22} on uncloneable encryption, and still suffice to yield a DI quantum money scheme as we discuss later. As for single-decryptor encryption, we choose to follow essentially the security definitions used in~\cite{CLLZ21, LLQZ22}.
\end{remark}

\paragraph{Uncloneable encryption with variable keys.} In our modified version of uncloneable encryption, the idea is that a particular ciphertext can be decrypted with several possible decryption keys, and each adversary in a cloning attack gets an independently generated decryption key. To further illustrate what we mean, we shall discuss this in the context of the uncloneable encryption scheme based on Wiesner states given by \cite{BL20}. (We stress that this is merely an example to demonstrate the key ideas, and VKECM is not restricted to such a specific implementation --- for a general formulation of VKECM, refer to Section~\ref{subsec:VKECM}.) For $a, x\in \{0,1\}$, we shall use $\ket{a^x}$ to denote the state $H^x\ket{a}$, where $H$ is the Hadamard matrix that takes the computational basis to the $\ket{+}, \ket{-}$ basis. For $a, x \in \{0,1\}^n$, we shall use $\ket{a^x}$ to denote $\bigotimes_{i=1}^n\ket{(a^i)^{x_i}}$. These $\ket{a^x}$ states are called Wiesner states. The basic encryption scheme (without using the QROM) in \cite{BL20} is as follows: the ciphertext corresponding to a message $m$ of $n$ bits is $(m\oplus a, \ket{a^x})$, for uniformly random $x$ and $a$, and the decryption key is $x$. On getting $x$, a single party can measure the quantum part of the ciphertext in the bases indicated by $x$ to recover $a$, and hence $m$. However, because the Wiesner states satisfy a monogamy of entanglement property \cite{TFK+13}, two parties cannot simultaneously do this. 

Note that in the scheme described above, the string $a$ which is generated really is a ``private key''\footnote{To avoid confusion: note that here we do not use the term ``private key" in the same sense as in a public key encryption procedure.} that is required to do the encryption procedure, but which cannot be revealed to any party if any kind of security is desired. Fortunately, after the encryption procedure is completed, $a$ does not need to be stored; only the string $x$, which is completely independent of $a$, needs to be stored and possibly released later as a decryption key.
 
Now consider the following modification: we still use Wiesner states $\ket{a^x}$, but we cannot use all the bits of $a$ as a one-time-pad for the message --- in fact we require that each party that wants to decrypt the message has to learn a different (independently generated) subset of the bits of $a$ in order to do so. The reasons we need to do this are technical and have to do with the proof style based on parallel repetition we use (we shall expand more on this in Section \ref{sec:overview}), but it can be achieved by modifying the protocol in the following way. If the message length is $n$ bits, the Wiesner states will now be $l$ bits, for $l > n$. The ciphertext will be $(m\oplus r, \ket{a^x})$, where $r$ is a uniformly random string of $n$ bits. $(r, a, x)$ will all need to be stored as private key now, and there will be a ``key release" procedure that takes the private key and generates a decryption key with a random subset $T$ of $[l]$ of size $n$. An instance of the decryption key is $(r\oplus a_T, T, x_T)$, and each time a decryption key is released from the decryption procedure, $T$ is generated independently. Obviously this means there are many possible decryption keys, corresponding to different values of $T$. A single decryptor given the decryption key and using $\ket{a^x}$ can learn $r$, and thus can learn $m$ using the classical part of the ciphertext. This also satisfies the property we required, that if two parties both want to learn the message, they have to learn independent subsets of the bits of $a$.

Our full DI scheme for achieving uncloneable encryption with variable keys is very similar to the modified version of the \cite{BL20} scheme described above, except with a ``testing'' step to obtain DI security\footnote{To see that the~\cite{BL20} scheme does not work if the state preparation is untrusted, observe that if the state prepared is simply a classical record of the values $(a,x)$ rather than the Wiesner states $\ket{a^x}$, then it is trivially insecure. If converted to an entanglement-based protocol in which the client performs some choice of measurement $x$ and obtains an output $a$, observe that if the client's measurements are untrusted, then the devices could just be implementing a completely classical strategy in which for each round the output $a$ is perfectly deterministic for each $x$, in which case all dishonest parties will know the value of $a$ once given $x$. (If desired, this deterministic behaviour could be made undetectable by any statistical checks involving only the frequency distribution of $a$ and/or $x$, by instead making the value of $a$ for each $x$ a function of some classical ``hidden variable'' $\Lambda$, a copy of which is held by all dishonest parties.)} --- basically, we play a nonlocal game on a random subset of the rounds, with the informal goal of ensuring that the devices cannot both win the nonlocal game on those rounds with high probability and still be able to usefully clone the resulting states. (Note, however, that our results are \emph{not} based on parallel self-testing theorems; rather, the only place we invoke self-testing is to study a single protocol round, after which we separately derive a parallel-repetition theorem to analyze the entire protocol. In particular, this means that in principle one could substitute the self-testing argument with other methods for analyzing single protocol rounds.) 

We provide a formal definition of uncloneable encryption with variable keys and its related security criteria in Section \ref{subsec:VKECM}. As indicated in the illustrative example above, the main difference between our definition and that introduced by \cite{BL20} is that the whole private key that was used in the encryption needs to be stored, and there is a key release procedure that takes the private key as input, uses additional private randomness, and outputs an independent decryption key each time one is requested (here by independent we mean the additional private randomness is independent for each decryption key). Additionally, since we work in the DI setting, our encryption procedure involves a small amount of interaction\footnote{However, this interaction can potentially be removed if the client can impose some additional constraints on their devices; we elaborate on this in Remark~\ref{remark:interactive}.} to implement the ``testing'' step, and we include an option to abort the procedure if this test fails. Such features are typically required in DI cryptography protocols, which rely on rigidity properties of various interactive procedures (such as non-local games, taking into account the need to check whether the game is won) to ``test'' if the device behaviour is close to the ideal case, as mentioned above. For instance, such interaction was also present in the \cite{GMP22} scheme, which was DI under computational assumptions.\footnote{We stress however that despite this interactive aspect, we do \emph{not} assume the receiver has to be honest in our setup. While a dishonest receiver could of course lie about the outputs of their devices, this poses no problems for a DI security proof, because such behaviour can always be absorbed into the operations/measurements performed by the dishonest party --- this line of reasoning has been used in many previous works on cryptographic scenarios with some potentially dishonest receiver (including uncloneable encryption) such as~\cite{FM18,GMP22,KT23}. (See also Remark~\ref{remark:receiver} later below.)}

Under those definitions, our main result regarding the achievability of DI uncloneable encryption with variable keys is stated in Theorem \ref{thm:main}, and the scheme achieving this is described in Scheme~\ref{prot:DI-VKECM}.
Some additional notable features of our scheme are as follows:
\begin{itemize}
\item The uncloneable encryption scheme of \cite{GMP22}, which is device-independent with computational assumptions, allows for some noise in the devices, but their approach requires the noise parameter to vanish in the limit of large message length $n$. In contrast, our protocol tolerates a \emph{constant} level of noise in the honest devices. (For the device-dependent uncloneable encryption schemes, to our knowledge none of them have explicitly analyzed noise in the devices, though it should be possible to modify some of the schemes to account for this.)
\item Most DI cryptographic protocols that guarantee information theoretic security require that there is no communication between the devices of different parties involved in the protocol. Our security proof is based on the parallel repetition of a form of a non-local game; it was shown in \cite{JK21} that proofs based on parallel repetition can tolerate a small (but linear in $n$) amount of communication, or leakage, between the devices of all parties involved. We give a simpler proof, inspired by the lower bound of quantum communication complexity in terms of the quantum partition bound in \cite{LLR12}, to show that our scheme tolerates leakage between the client and the receiver during the encryption procedure, and between two parties who are both trying to decrypt the message in a cloning attack. This makes our security criterion qualitatively different from that considered in \cite{BL20} (even after accounting for the fully DI setting we consider), since the two parties in a cloning attack no longer have to be non-communicating --- any communication between the parties can obviously also be considered leakage between their devices. We only require that the total number of bits thus leaked be bounded. Our argument could also be adapted to obtain device-dependent schemes for uncloneable encryption that can tolerate some communication between the two parties in a cloning attack.
\end{itemize}

\paragraph{Application to quantum money.} Although our aim in constructing this modification of uncloneable encryption was to be able to prove DI security via proof techniques for parallel repetition, we believe that the modified definition can still be useful. For instance, it can be used just as well as the original notion of uncloneable encryption to get private key quantum money. The approach here is the same as that sketched in \cite{BL20}. Basically, a bank could produce a banknote by encrypting a random string $M$ using our procedure, then storing the private key as well as $M$ in its internal records, and placing the quantum ciphertext in the banknote. To verify a banknote, the bank would use the private key to run the key release procedure and generate a decryption key, then use this to decrypt the state in the banknote and check whether the output matches $M$. Our security definition for uncloneability immediately implies that if an adversary attempted to clone the banknote and submit it to two separate bank locations for verification (each of which independently runs the key release procedure), the probability of both being accepted\footnote{More precisely, the probability of both banknotes being accepted \emph{and} the original encryption procedure accepting as well.} is exponentially small. 

Since our scheme is secure in the DI setting, this means the above approach yields a method for obtaining DI quantum money.\footnote{This approach avoids an impossibility result for DI quantum money derived in~\cite{HS20}, because we are considering a somewhat different setup from that work. 
}  
We note that it does seem possible that another approach for DI quantum money would be to slightly modify existing protocols for DI quantum key distribution; however, there may be some benefits to using the approach in this work. In particular, while our analysis above does not say anything about what happens if the bank returns successfully verified banknotes (rather than destroying them and generating fresh ones), the proof techniques we use here should be modifiable to allow a security proof for a scheme that returns successfully verified banknotes a fixed number of times, chosen at the point of generation of the banknote. 

\paragraph{Single-decryptor encryption.} In single-decryptor encryption of classical messages (SDECM), a key generation process first produces an encryption key and a quantum decryption key; the decryption key is then given to a dishonest party, while the encryption key is stored with the honest party. The requirement is that two parties between whom the dishonest party has distributed the decryption key cannot both decrypt a message (or two independent messages) that have been encrypted with the encryption key. A simple change of perspective in the Wiesner state protocol for uncloneable encryption gives us a single-decryptor encryption protocol: here $\ket{a^x}$ will be the decryption key, $(a,x)$ the encryption key (note here the difference with uncloneable encryption, where only $x$ needed to be stored), and messages will be encrypted with $(a,x)$ as $(m\oplus a,x)$.

In general the encryption procedure for single-decryptor encryption is allowed to use randomness outside of the encryption key in order to encrypt a message (indeed, this was the case in some of the schemes described in \cite{GZ20}). 
This has given rise to different possible security definitions, depending on whether the two parties who are challenged to decrypt 
should receive the same message encrypted using the same additional randomness, or whether the messages and randomness should be independent. 
Again, we defer the full details to Section~\ref{sec:defSDECM} rather than this introductory section, but to give a brief overview,
\cite{CLLZ21, LLQZ22} considered a definition where the two parties receive encryptions of independent messages using independent randomness, i.e.~they receive different ciphertexts. On the other hand,
\cite{GZ20} considered a definition where the two parties receive encryptions of the same message using the same additional randomness, i.e., when they both receive the same ciphertext, and furthermore showed that this is equivalent to uncloneable encryption under a particular security definition.\footnote{It is actually possible to consider variations of the security definition other than those in \cite{GZ20} and \cite{CLLZ21}: one can make either or both of the message and the additional randomness used in encryption for the two parties independent or identical. We discuss these issues in more detail in Section~\ref{sec:defSDECM}.}

We believe that single-decryptor encryption under the security definition of \cite{CLLZ21} should be equivalent to our definition of uncloneable encryption with variable keys, using essentially the same reduction as \cite{GZ20}, though we do not attempt to formally prove this claim. Instead, we directly prove that a very straightforward modification of our DI protocol for uncloneable encryption with variable keys, in essentially the same manner as the Wiesner state protocol, achieves single-decryptor encryption under a definition similar to \cite{CLLZ21} (see Section~\ref{sec:defSDECM} for detailed definitions).

\paragraph{Single-decryptor encryption of bits and trits.} For single-bit or single-trit messages specifically, we are also able to modify the above single-decryptor encryption scheme in such a way that (under the independent-message security definition) the adversaries' probability of guessing the messages is at most only negligibly larger than the trivial value. This is a stronger security property than that of the above scheme (where the adversary's guessing probability is only exponentially small in the message length --- see Section~\ref{sec:defSDECM} for more precise details), and furthermore we obtain it without using the QROM, in contrast to previous work.\footnote{After the publication of the preprint of this work (in which our focus was on applying this idea for VKECM), a similar result was independently obtained in~\cite{AKL23}. We have since identified a flaw in our application of this idea in VKECM, which we elaborate on in this updated manuscript (see Remark~\ref{remark:OTPreuse}). However, we found that it could still be used in single-decryptor encryption to obtain a similar result to~\cite{AKL23}, and present this finding in this updated version.} The idea here is as follows: suppose the probability for two adversarial parties to simultaneously guess some ``raw'' $l$-bit strings is exponentially small in $l$, but larger than $2^{-l}$. Now, we would ideally like to ``extract" a single bit such that the probability of the two parties simultaneously guessing it is $\frac{1}{2} + \negl(l)$.\footnote{This part of our analysis does not depend significantly on the exact scaling of the function $\negl(l)$; when applying it in our protocol, it will be an exponentially small function (since we achieve exponential rather than polynomial security), though larger than $2^{-l}$.} If there had just been one party instead of two, then this would be achievable by using randomness extractors \cite{TSS+11,DPV+12}; however, the proof techniques in those works do not seem to straightforwardly generalize to two parties. The difficulties with using extractors against two parties were noted in \cite{AKL+22}, although using randomness extractors was not ruled out by the impossibility results in the same work.

In the setting of single-decryptor encryption (under the independent-message security definition), our method to overcome this difficulty is to design our scheme such that it involves implementing randomness extractors with \emph{different} random seeds for the two parties. Specifically, we shall use the inner product function, which is known to be an extractor. If the two parties cannot guess $x$ with high probability from a shared state $\ket{\rho}_x$, we prove that the probability that one party guesses $x\cdot \rvec^1$ and the other guesses $x\cdot \rvec^2$, where $\rvec^1$ and $\rvec^2$ are independently generated from the uniform distribution over $\{0,1\}^l$, is at most $\frac{1}{2} + \negl(l)$.\footnote{The argument also works if there are two (possibly correlated) strings $x^1$ and $x^2$, and we have a bound on one party learning $x^1$ and the other party learning $x^2$ from $\ket{\rho}_{x^1x^2}$, which will be the case in our actual setting. Here we can upper bound the probability of the parties learning $x^1\cdot \rvec^1$ and $x^2\cdot \rvec^2$ respectively.} Instead of going via the standard arguments for extractors, our argument for this is done similar to the quantum version of the Goldreich-Levin theorem for hardcore bits \cite{AC02}. The idea is that if the probability of the parties guessing $x\cdot \rvec^1$ and $x\cdot \rvec^2$ averaged over $\rvec^1, \rvec^2$ is more than $\frac{1}{2} + \negl(l)$, then they can independently run the Bernstein-Vazirani algorithm on their halves of the shared state, to both guess the entire string $x$ with too high a probability. Note that since the two parties need to run the Bernstein-Vazirani algorithm independently, this argument does not work if we consider the probability of both parties learning $x\cdot v$ for the same $v$.

This argument generalizes to trits as well, but it does not seem to extend to longer bit strings. 
The implication of this result is that for single-decryptor encryption, we have ``perfect'' anti-piracy security (see Section~\ref{sec:defSDECM} for formal definitions) for single-bit or single-trit messages, which in turn also implies CPA-style security for such messages.
We remark that in principle, one could make a similar modification to our scheme for uncloneable encryption with variable keys, to attempt to obtain uncloneable bits or trits. However, there is a subtle issue that arises in the attempt, and hence we were unable to prove security of the resulting scheme --- we discuss this difficulty further in Remark~\ref{remark:OTPreuse} later, along with the relevant properties that would need to be proven in order to obtain such a result.

\subsection{Technical overview}\label{sec:overview}
In this section, we give an overview of how we construct our DI uncloneable encryption with variable keys (DI-VKECM) scheme, and how we prove its security. All the arguments in this section are for proving uncloneability. In the context of VKECM, uncloneability is the requirement that if the encryption of a uniformly random message gets distributed between two parties, say Bob and Charlie, then the average probability that both of them guess the message given access to independently generated decryption keys, is exponentially small in the number of bits in the message.

\paragraph{Proving security by analyzing a non-local game.} The first thing to notice is that the Wiesner states correspond to the states after the measurement of one of the parties (let's say Alice) in the CHSH non-local game. Therefore, we can consider doing something very similar to many protocols for DI quantum key distribution (QKD) \cite{PAB+09}. During the encryption process, the client (who is honest) and the receiver (who may be dishonest) will share states compatible with some $l$ copies of the CHSH game. On some of these copies, the CHSH game will be self-tested, and the rest will be used for encryption. The idea is that, if a random subset is used for self-testing, the dishonest receiver, who had to have prepared the devices beforehand, would have had to actually prepare i.i.d. copies of the ideal CHSH state and measurements (up to some amount of noise per copy) in order to pass the self-testing check.\footnote{In principle, one might be able to use some other non-local game with self-testing properties; however, the CHSH game is particularly convenient since (after Alice's measurement) the ideal strategy produces exactly the Wiesner states. We also note that the monogamy of entanglement game from~\cite{TFK+13} (which was used in the~\cite{BL20} security proof) was shown to have some self-testing properties as well~\cite{BC21}. However, to our knowledge all self-testing properties of that game are derived for the ``one-sided DI'' scenario where Alice's operations are trusted, and hence it is unclear whether they could be applied to obtain a fully DI protocol.}

However, there is one key aspect in which the situation here differs from that in QKD. In a DIQKD protocol, measurements corresponding to all copies of the game in question can be done at once, and then the Serfling bound can be applied to the resulting input-output distribution, in order to say something about the entire distribution from what is observed in the self-testing subset of inputs and outputs. However, for uncloneable encryption, measurements corresponding to the copies of the game which will be used for encryption cannot be done at the same time as the testing subset; in an honest implementation, these states need to remain unmeasured because the ciphertext needs to be quantum. Thus there is no fixed distribution on which the Serfling bound can be applied.

Because of the above problem, we shall instead use a two-round non-local game\footnote{While we refer to this as a non-local game for brevity, we note that strictly speaking it does not fall within the ``standard'' framework of non-local games (in which there is no communication between the parties), because in our scenario the player Barlie sends or distributes states to the players Bob and Charlie \emph{after} receiving some inputs. It would be convenient if we could have analyzed some three- or four-player non-local game in the ``standard'' sense; however, some sort of communication between parties seems unavoidable in order for us to be able to relate this game to the task of uncloneable encryption (similar to the situation for monogamy of entanglement games in~\cite{TFK+13,BL20}). We specify the exact order in which the players can share or distribute states in our full game description in Section~\ref{sec:cl-game}.} whose first round is played between two players Alice and Barlie, and the second round is played between two players Bob and Charlie, between whom Barlie's portion of the state at the end of the first round gets distributed (Alice's portion of the state at the end of the first round remains untouched). We call this game the cloning game $\cl_\gamma$. We describe a single instance of the game first; we shall actually need the parallel-repeated version of the game for the security proof. In the first round of a single instance of the game, Alice will receive a uniformly random single-bit input, and Barlie will receive a trit in $\{0,1, \text{keep}\}$, with ``keep" occuring with probability $(1-\gamma)$. The input ``keep" to Barlie indicates that the CHSH game will not be tested, and in this case the first round is automatically won; otherwise, Alice and Barlie's inputs are the same as the inputs of the CHSH game, and they need to produce outputs that satisfy the CHSH winning condition. If Barlie did not get ``keep" in the first round, the second round is automatically won; otherwise Bob and Charlie get the same input as Alice did in the first round of the game, and they both have to guess Alice's output bit from the first round. Note that Alice gets the same input regardless of Bob's input, and in the honest case, if she gets input $x$ and produces output $a$ in the first round, the state on Barlie's side is $\ket{a^x}$. If the first round is won with high enough probability, then the state after the first round on Barlie's side is close to $\ket{a^x}$, which means by the monogamy of entanglement property of Wiesner states, we can upper bound the probability that Bob and Charlie can guess $a$ given $x$. This means that the overall winning probability of $\cl_\gamma$ is bounded away from 1. The style of argument we employ here is similar to what was considered in \cite{ACK+14, KST22} for two-party cryptography, although parallel repetition was not considered in those works.

When we consider the parallel-repeated $\cl_\gamma$, i.e., $l$ i.i.d. copies of $\cl_\gamma$, a constant fraction of the instances (in expectation) will be tested in the first round, and the rest can be used for encryption. The ciphertext in this case will be the state after the first round on Barlie's side (which in the honest case is a Wiesner state of the form $\ket{a_S^{x_S}}$, where $x$ and $a$ and Alice's first round inputs and outputs, and $S$ is the subset on which Barlie's first round input was ``keep"), along with the classical string $m \oplus a_S$. If we can prove that the winning probability of the parallel-repeated $\cl_\gamma$ decreases exponentially in $n$, then we could prove the uncloneability of this scheme. Note that the scheme as described here is actually an uncloneable encryption scheme in the original sense of \cite{BL20}, since $x_S$ is the only decryption key.

\paragraph{Parallel repetition of the cloning game.} Unfortunately, we cannot prove a parallel repetition of the game $\cl_\gamma$ as described. What we can prove a parallel repetition theorem for is a modified version of $\cl_\gamma$ which is ``anchored". The anchoring transformation we use is similar to those used in \cite{Vid17, KT23}. Essentially, the anchoring property requires that Bob and Charlie's second round inputs cannot be perfectly correlated with Alice's first round input, or each other --- with some small probability, these distributions need to be product instead. We do this by giving Bob and Charlie independently ``blank" inputs rather than $x_i$ on some instances, and then not using those instances for encryption for them (this corresponds to the second round being won for free on these instances of the game) --- this forces us to use the additional random string $r$ in the ciphertext as described earlier, with $(r\oplus a_{T\cap S}, x_{T\cap S}, T)$ being a decryption key for random $T$. This now places us in the setting of VKECM.

We can prove a parallel repetition theorem for the anchored version of $\cl_\gamma$, which we denote by $\cl_{\gamma, \alpha}$ for anchoring parameter $\alpha$, similar to how a parallel repetition theorem for a different two-round game was proved in \cite{KT23}. The game considered in \cite{KT23} had only two players in both rounds, so in our proof we require some additional steps to take care of the two players Bob and Charlie, between whom Barlie's state is divided after the first round. We use the information theoretic framework for parallel repetition that was introduced by \cite{Raz95, Hol09}. In this framework, we consider a strategy for $l$ copies $\cl_{\gamma, \alpha}$ and condition on the event $\clE$ of the winning condition being satisfied on some $C\subseteq[l]$ instances. We show that if $\Pr[\clE]$ is not already small, then we can find another coordinate in $i\in\oC=[l]\setminus C$ where the winning probability conditioned on $\clE$ is bounded away from 1. The proof is by contradiction: we show that if the probability of $\clE$ is large and the probability of winning in $i$ conditioned on $\clE$ is not bounded away from 1, then there is a stategy for a single copy of $\cl_{\gamma, \alpha}$, whose winning probability is higher than the maximum winning probability of $\cl_{\gamma, \alpha}$. This is done by defining a state representing the inputs, outputs and shared entanglement in the strategy for $\cl_{\gamma, \alpha}^l$, conditioned on $\clE$. When Alice, Barlie, Bob and Charlie's inputs for the $i$-th instance of $\cl_{\gamma, \alpha}$ are $x_i, u_i, y_i, z_i$ respectively, we denote this state by $\ket{\vph}_{x_iu_iy_iz_i}$. The state $\ket{\vph}_{x_iu_iy_iz_i}$ is such that the distribution obtained by measuring the corresponding $i$-th output registers on it is the distribution of outputs in the original strategy for the inputs $x_iu_iy_iz_i$, conditioned on the event $\clE$. Therefore, if the probability of winning in the original strategy conditioned on $\clE$ is too high, and the players are able to output from the state $\ket{\vph}_{x_iu_iy_iz_i}$ (or close to this distribution) in a single instance of $\cl_{\gamma, \alpha}$, they can win the single instance with too high probability.

The rest of the proof will thus involve showing how the players can provide outputs that are close to the output distribution in $\ket{\vph}_{x_iu_iy_iz_i}$ when playing a single instance of $\cl_{\gamma, \alpha}$. Let us denote the blank inputs to Bob and Charlie in the second round by $\bot$, so that $\ket{\vph}_{x_iu_i\bot\bot}$ is $\ket{\vph}_{x_iu_iy_iz_i}$ when they both get this blank input. The first thing to note is that there exists some $i$ such that the distribution of Alice and Barlie's first round outputs is almost the same in $\ket{\vph}_{x_iu_i\bot\bot}$ and $\ket{\vph}_{x_iu_iy_iz_i}$ for any $y_iz_i$. This is because the distributions were exactly the same originally (because the first round outputs were produced without access to $y_iz_i$), and conditioning on $\clE$ does not change the distribution too much (because we have assumed the probability of $\clE$ happening is not too small). So in the first round, Alice and Barlie could jointly produce the state $\ket{\vph}_{x_iu_i\bot\bot}$ on getting inputs $x_i, u_i$, and measure its output registers to give their first round outputs. The distribution of $X_iU_i$ conditioned on $Y_i=\bot, Z_i=\bot$ is product, so we can use an argument very similar to that of \cite{JPY14} to prove a parallel repetition theorem for one-round games with product distributions, to argue that there exist unitaries $V^\A_{x_i}$ and $V^{\B\C}_{u_i}$ that Alice and Barlie can appear on their registers of a shared state $\ket{\vph}_{\bot\bot}$ (which is the superposition of $\ket{\vph}_{x_iu_i\bot\bot}$ over all $x_iu_i$) to get close to $\ket{\vph}_{x_iu_i\bot\bot}$.

In the second round, Barlie then distributes his registers of the shared state (which is close to $\ket{\vph}_{x_iu_i\bot\bot}$, except with the $i$-th first-round output registers being measured) between Bob and Charlie, and also gives them $u_i$. Because Bob and Charlie's inputs are anchored w.r.t. each other, which means that they are product (with each other and with Alice) with some constant probability, we can show by a similar argument that there exist unitaries $V^\B_{u_iy_i}$ and $V^\C_{u_iz_i}$ that Bob and Charlie can apply on their registers of $\ket{\vph}_{x_iu_i\bot\bot}$ to get it close to $\ket{\vph}_{x_iu_iy_iz_i}$. These unitaries do not act on the output registers from the first round, so we can argue that on average, $V^\B_{u_iy_i}\otimes V^\C_{u_iz_i}$ also takes $\ket{\vph}_{x_iu_i\bot\bot}$ conditioned on particular first round output values to $\ket{\vph}_{x_iu_iy_iz_i}$ conditioned on the same values. Therefore, Bob and Charlie can apply these unitaries on the state they get from Barlie, and provide second round outputs by measuring the $i$-th second round output registers of the resulting state. This completes the stategy for Alice, Barlie, Bob and Charlie for the single instance of $\cl_{\gamma, \alpha}$.

\subsection{Discussion and future work}
The first question left open by our work is whether it is possible to achieve DI security for the original notion of ECM rather than VKECM. We highlight that it seems necessary to do a parallel rather than sequential (which is the setting where parties enter their inputs and get outputs from their devices one by one, instead of all at once) style of proof here. This is because, while we could ensure that the receiver during the encryption procedure enters the inputs for self-testing into their device one by one, simply by sending the inputs one by one and requiring a reply before supplying the next input (thereby giving the protocol many rounds of interaction), it seems fairly unnatural to enforce this constraint on Bob and Charlie. The only proof technique we have for parallel device-independent settings is parallel repetition, so proving a parallel repetition theorem for a game like the $\cl_\gamma$ game we described (the version without anchoring) seems necessary. However, it is not known how to prove an exponential parallel repetition theorem for even one-round two-player non-local games where the inputs of the two players are arbitrarily correlated. The most general exponential parallel repetition theorem known here is also for anchored games \cite{BVY17}. Of course, we do not need to prove a parallel repetition theorem for all possible games, only the specific game $\cl_\gamma$. In our parallel repetition theorem $\cl_{\gamma, \alpha}$, we did not make use of any structure in the game except for the input distribution. So it may be possible to prove parallel repetition for $\cl_\gamma$ by making use of its specific structure.

The second question we leave open is whether it is possible to extend the ``randomness extraction" result we have to more than one bit or trit. As we noted before, our style of proof does not work for bit strings, but could some sort of parallel repetition or composability result be used to prove security for a case where instead of one inner product we have several inner products? Also, there remains the question of whether this ``extraction'' result can be applied in protocols outside of single-decryptor encryption; in particular, whether it could be used in VKECM (see Remark~\ref{remark:OTPreuse} for further discussion) or even the original ECM.

Finally, there is the question of finding more applications for VKECM. For instance, although we showed that VKECM can be used just as well as ECM for private key quantum money, we have not studied whether the application to quantum copy protection also works with VKECM. Both constructions of copy protection from uncloneable encryption \cite{CMP20, AK21} are schemes for copy-protecting a class of functions known as \emph{multi-bit point functions}. A point function $f_{a,b}$  evaluates to $0$ on all inputs except a special input $a$, on which it evaluates to a string $b$. When constructing copy protection from uncloneable encryption, the string $a$ is taken as the decryption key of the uncloneable encryption procedure, and $b$ is the encrypted message. The idea is that two parties among whom the copy-protected program has been distributed, should not both be able to evaluate $b$ when they have $a$ as their input --- this is guaranteed by the security of the uncloneable encryption scheme. The question then is: what happens if we try to do these constructions with VKECM instead of ECM? Because there are many possible valid decryption keys in VKECM, does this mean we could copy protect a class of functions different from point functions? We leave all of these interesting questions for future work.

\subsection{Organization of the paper}
In Section \ref{sect:prelim}, we introduce some notation we shall be using throughout the paper, and describe some preliminaries on probability theory and quantum information. In Section \ref{sec:secdefn}, we formally define VKECM, SDECM, and related security criteria; we also describe the device-independent setting and assumptions therein in this section. In Section \ref{sec:cl-game}, we formally define the two-round cloning games $\cl_\gamma$ and $\cl_{\gamma, \alpha}$ and prove that the probability of winning a single instance of it is bounded away from 1. In Section \ref{sec:DIVKECM}, we describe our DI-VKECM scheme and prove its security using the parallel repetition theorem for $\cl_{\gamma, \alpha}$. In Section~\ref{sec:DISDECM}, we describe our DI-SDECM scheme and prove it satisfies the security definition used in~\cite{CLLZ21, LLQZ22}. In Section \ref{sec:bittritPA}, we describe how the scheme in Section \ref{sec:DISDECM} can be modified to ``extract'' the randomness in the raw keys, producing a stronger security result for single-bit or single-trit messages. Finally, in Section \ref{sec:parrep}, we prove the parallel repetition theorem for $\cl_{\gamma, \alpha}$.

\section{Preliminaries}\label{sect:prelim}

\subsection{Probability theory}
We shall denote the probability distribution of a random variable $X$ on some set $\clX$ by $\sfP_X$. For any event $\clE$ on $\clX$, the distribution of $X$ conditioned on $\clE$ will be denoted by $\sfP_{X|\clE}$. For joint random variables $XY$ with distribution $\sfP_{XY}$, $\sfP_X$ is the marginal distribution of $X$ and $\sfP_{X|Y=y}$ is the conditional distribution of $X$ given $Y=y$; when it is clear from context which variable's value is being conditioned on, we shall often shorten the latter to $\sfP_{X|y}$. We shall use $\sfP_{XY}\sfP_{Z|X}$ to refer to the distribution
\[ (\sfP_{XY}\sfP_{Z|X})(x,y,z) = \sfP_{XY}(x,y)\cdot\sfP_{Z|X=x}(z).\]
Occasionally we shall use notation of the form $\sfP_{XY}\sfP_{Z|x^*}$. This denotes the distribution
\[ (\sfP_{XY}\sfP_{Z|x^*})(x,y,z) = \sfP_{XY}(x,y)\cdot\sfP_{Z|X=x^*}(z),\]
which potentially takes non-zero value when $x\neq x^*$. For two distributions $\sfP_X$ and $\sfP_{X'}$ on the same set $\clX$, the $\ell_1$ distance between them is defined as
\[ \Vert\sfP_X - \sfP_{X'}\Vert_1 = \sum_{x\in\clX}|\sfP_X(x) - \sfP_{X'}(x)|.\]

\begin{fact}\label{fc:l1-dec}
For joint distributions $\sfP_{XY}$ and $\sfP_{X'Y'}$ on the same sets,
\[ \Vert\sfP_X -  \sfP_{X'}\Vert_1 \leq \Vert\sfP_{XY} - \sfP_{X'Y'}\Vert_1.\]
\end{fact}
\begin{fact}\label{fc:l1-dist}
For two distributions $\sfP_X$ and $\sfP_{X'}$ on the same set and an event $\clE$ on the set,
\[ |\sfP_X(\clE) - \sfP_{X'}(\clE)| \leq \frac{1}{2}\Vert\sfP_X - \sfP_{X'}\Vert_1.\]
\end{fact}
\begin{fact}\label{fc:cond-prob}
Suppose probability distributions $\sfP_X, \sfP_{X'}$ satisfy $\Vert \sfP_X - \sfP_{X'}\Vert_1 \leq \eps$, and an event $\clE$ satisfies $\sfP_X(\clE) \geq \alpha$, where $\alpha > \eps$. Then,
\[ \Vert\sfP_{X|\clE} - \sfP_{X'|\clE}\Vert_1 \leq \frac{2\eps}{\alpha}.\]
\end{fact}

\subsection{Quantum information}
The $\ell_1$ distance between two quantum states $\rho$ and $\sigma$ is given by
\[ \Vert\rho-\sigma\Vert_1 = \Tr\sqrt{(\rho-\sigma)^\dagger(\rho-\sigma)} = \Tr|\rho-\sigma|.\]
The fidelity between two quantum states is given by
\[ \sfF(\rho,\sigma) = \Vert\sqrt{\rho}\sqrt{\sigma}\Vert_1.\]
The Bures distance based on fidelity is given by
\[ \sfB(\rho,\sigma) = \sqrt{1-\sfF(\rho,\sigma)}.\]

$\ell_1$ distance, fidelity and Bures distance are related in the following way.
\begin{fact}[Fuchs-van de Graaf inequality]\label{fc:fvdg}
For any pair of quantum states $\rho$ and $\sigma$,
\[ 2(1-\sfF(\rho,\sigma)) \leq \Vert\rho-\sigma\Vert_1\leq 2\sqrt{1-\sfF(\rho,\sigma)^2}.\]
Consequently,
\[ 2\sfB(\rho,\sigma)^2 \leq \norm{\rho-\sigma}_1 \leq 2\sqrt{2}\cdot\sfB(\rho,\sigma).\]
For two pure states $\ket{\psi}$ and $\ket{\phi}$, we have
\[ \norm{\,\state{\psi} - \state{\phi}\,}_1 = \sqrt{1 - \sfF\left(\state{\psi},\state{\phi}\right)^2} = \sqrt{1-|\inner{\psi}{\phi}|^2}.\]
\end{fact}
\begin{fact}[Uhlmann's theorem]\label{fc:uhlmann}
Suppose $\rho$ and $\sigma$ are mixed states on register $X$ which are purified to $\ket{\rho}$ and $\ket{\sigma}$ on registers $XY$, then it holds that
\[ \sfF(\rho, \sigma) = \max_U|\matel{\rho}{\Id_X\otimes U}{\sigma}|\]
where the maximization is over unitaries acting only on register $Y$. Due to the Fuchs-van de Graaf inequality, this implies that there exists a unitary $U$ such that
\[ \norm{(\Id_X\otimes U)\state{\rho}(\Id_X\otimes U^\dagger) - \state{\sigma}}_1 \leq 2\sqrt{\norm{\rho-\sigma}_1}.\]
\end{fact}
\begin{fact}\label{fc:chan-l1}
For a quantum channel $\clE$ and states $\rho$ and $\sigma$,
\[ \Vert\clE(\rho) - \clE(\sigma)\Vert_1 \leq \Vert\rho-\sigma\Vert_1 \quad \quad \sfF(\clE(\rho),\clE(\sigma)) \geq \sfF(\rho,\sigma).\]
\end{fact}

The entropy of a quantum state $\rho$ on a register $Z$ is given by
\[ \sfH(\rho) = -\Tr(\rho\log \rho).\]
We shall also denote this by $\sfH(Z)_\rho$. For a state $\rho_{YZ}$ on registers $YZ$, the entropy of $Y$ conditioned on $Z$ is given by
\[ \sfH(Y|Z)_\rho = \sfH(YZ)_\rho - \sfH(Z)_\rho\]
where $\sfH(Z)_\rho$ is calculated w.r.t. the reduced state $\rho_Z$.

The relative entropy between two states $\rho$ and $\sigma$ of the same dimensions is given by
\[ \sfD(\rho\Vert \sigma) = \Tr(\rho\log\rho) - \Tr(\rho\log\sigma).\]
\begin{fact}[Pinsker's Inequality]\label{pinsker}
For any two states $\rho$ and $\sigma$,
\[ \Vert\rho-\sigma\Vert_1^2 \leq 2\ln 2\cdot\sfD(\rho\Vert\sigma) \quad \text{ and } \quad \sfB(\rho,\sigma)^2 \leq \ln 2\cdot\sfD(\rho\Vert\sigma).\]
\end{fact}
The mutual information between $Y$ and $Z$ with respect to a state $\rho$ on $YZ$ can be defined in the following equivalent ways:
\[  \sfI(Y:Z)_\rho = \sfD(\rho_{YZ}\Vert\rho_Y\otimes\rho_Z) = \sfH(Y)_\rho - \sfH(Y|Z)_\rho = \sfH(Z)_\rho - \sfH(Z|Y)_\rho.\]
The conditional mutual information between $Y$ and $Z$ conditioned on $X$ is defined as
\[ \sfI(Y:Z|X)_\rho = \sfH(Y|X)_\rho - \sfH(Y|XZ)_\rho = \sfH(Z|X)_\rho - \sfH(Z|XY)_\rho.\]
Mutual information can be seen to satisfy the chain rule
\[  \sfI(XY:Z)_\rho =  \sfI(X:Z)_\rho +  \sfI(Y:Z|X)_\rho.\]

A state of the form
\[ \rho_{XY} = \sum_x \sfP_X(x)\state{x}_X\otimes\rho_{Y|x}\]
is called a CQ (classical-quantum) state, with $X$ being the classical register and $Y$ being quantum. We shall use $X$ to refer to both the classical register and the classical random variable with the associated distribution. As in the classical case, here we are using $\rho_{Y|x}$ to denote the state of the register $Y$ conditioned on $X=x$, or in other words the state of the register $Y$ when a measurement is done on the $X$ register and the outcome is $x$. Hence $\rho_{XY|x} = \state{x}_X\otimes \rho_{Y|x}$. When the registers are clear from context we shall often write simply $\rho_x$.
For CQ states, the expressions for conditional entropy, relative entropy and mutual information for $\rho_{XY}$ and $\sigma_{XY}$ given by
\[ \rho_{XY} = \sum_x\sfP_{X}(x)\state{x}_X\otimes\rho_{Y|x} \quad \quad \sigma_{XY} = \sum_x\sfP_{X'}(x)\state{x}_X\otimes\sigma_{Y|x},\]
reduce to
\begin{gather*}
\sfH(Y|X)_\rho = \bbE_{\sfP_X}\sfH(Y)_{\rho_x}, \\
\sfD(\rho_{XY}\Vert \sigma_{XY}) = \sfD(\sfP_X\Vert \sfP_{X'}) + \bbE_{\sfP_X} \sfD(\rho_{Y|x}\Vert \sigma_{Y|x}) \\
\sfI(Y:Z|X) = \bbE_{\sfP_X}\sfI(Y:Z)_{\rho_x}.
\end{gather*}
\noindent When talking about entropies of only the classical variables of a CQ state, we shall sometimes omit the state in the subscript. Additionally, for an event $\clE$ defined on the classical variable $X$ of a CQ state, we shall use notation like use $\sfH(X|\clE)$ and $\sfH(X_1|X_2;\clE)$ to talk about entropies of the event when the classical distribution is conditioned on $\clE$.

\section{Security definitions}
\label{sec:secdefn}

In this section, we formally define VKECM and SDECM in the DI setting. 
We begin by laying out the form of the devices we consider, and making a formal statement of what it means to have DI security. For this statement, we are following fairly standard conventions in the field of DI cryptography. With this perspective, DI security is a notion analogous to information-theoretic security or computational security, in that it describes a class of possible behaviours for dishonest parties.

\begin{definition}\label{def:DI-sec} (Device-independent security) Consider any cryptographic scheme involving multiple parties. We say that \emph{devices} for such a scheme consist of objects with the following functionality:
\begin{enumerate}
\item Initially, each party's device holds a share of some quantum state.
\item\label{pt:DIinputs} Each party's device can (possibly more than once) accept a classical input string, which it uses to perform some measurement on its share of the state and produce a classical output string.
\item A dishonest party can at any time perform arbitrary quantum operations on its share of the state, including joint operations that involve other registers held by that party.
\end{enumerate}

We say that a scheme using such devices achieves \emph{device-independent (DI) security} (under any security definition required to hold against some class of possible dishonest behaviours) if the class of possible dishonest behaviours includes arbitrary choices of the initial state distributed (and the Hilbert spaces the state is defined on), as well as the measurements performed in point~\ref{pt:DIinputs}, \emph{including} those in honest parties' devices. 

We furthermore say that it achieves \emph{DI security against $k$ bits of leakage} if we extend the class of possible dishonest behaviours to allow $k$ occasions at which a classical\footnote{There is essentially no loss of generality in restricting to classical communication, because if quantum communication is desired, the devices can include pre-shared entanglement in the initial quantum state and then use it for quantum teleportation with only classical communication, though this does require $2$ classical bits to send $1$ qubit of information.}  bit is communicated between the devices.
\end{definition}

For brevity, we may refer to a VKECM scheme achieving DI security as a DI-VKECM scheme, and analogously for SDECM, following the same convention as e.g.~DIQKD.


\begin{remark}\label{remark:receiver}
From the perspective of modelling dishonest behaviour, it would seem we should also allow a dishonest party to directly operate on their share of the quantum state in point~\ref{pt:DIinputs} as well, rather than being strictly constrained to supplying inputs to the device in the prescribed fashion. However, the critical observation here is that since we have not placed any constraints on the measurements in the dishonest case, any operations that a dishonest party could perform on the state can equivalently be carried out by the device itself in performing those measurements. Hence there is no loss of generality by focusing on the model presented above.
\end{remark}

\subsection{Uncloneable encryption with variable keys}
\label{subsec:VKECM}
We now describe the security definitions we use for uncloneable encryption with variable keys.
These definitions are essentially similar to those used in \cite{GMP22}, apart from the distinction between the private key and decryption key in our scheme. The definitions in \cite{GMP22} specify a classical client (the honest party). 
In our case the client will not be classical, but instead needs to have enough quantum capabilities to implement devices as described above (note however that in the honest implementation, the client will only need to perform single-qubit measurements). 
On the other hand, we shall not be assuming any computational limitations when considering dishonest behaviour (so we can achieve DI security in the sense described above), in contrast to that work.

We remark that while the definitions in~\cite{GMP22} differ slightly from those in~\cite{BL20}, the only differences are basically that the encryption procedure for a message is allowed to be interactive, and are allowed to output an abort symbol, with the protocol not continuing if there is an abort (the probability of not aborting also appears in the security condition).

\begin{definition}[Encryption of classical messages with variable keys]\label{def:vkecm}
Let $\lambda$ be a security parameter.
An \emph{encryption of classical messages with variable keys} (VKECM) scheme consists of a tuple $(\Enc, \KeyR, \Dec)$ such that
\begin{itemize}
\item $\Enc(1^\lambda, m)$ is a (potentially interactive) protocol between an honest client, who takes as input the security parameter $\lambda$ and a message $m$ in some message space $\clM$, and a potentially dishonest receiver, who takes as input the security parameter $\lambda$. 
The output of the protocol is $(F,\Kpriv,\rho)$, where $F$ is a flag held by the client which takes values \cmark (accept) or \xmark (reject), $\Kpriv$ is a private key held by the client, and $\rho$ is a quantum ciphertext state held by the receiver on a register we denote as $Q$.
\item $\KeyR(\kpriv)$ takes as input a private key $\kpriv$, and uses some internal randomness to generate and output a decryption key $\Kdec$.
\item $\Dec(\kdec,\rho)$ takes as input a decryption key $\kdec$ and a ciphertext state $\rho$  
on register $Q$,
and outputs a message value $\widetilde{M} \in \clM$. 
\end{itemize}
We say that a VKECM scheme is efficient if the computations that the client and an honest receiver perform in $\Enc,\KeyR,\Dec$ are polynomial time in $\lambda$ and the length (in bits) of $m$.

We require that the VKECM scheme satisfies \emph{completeness}: if all steps are carried out honestly, then for any $m\in\clM$, we have (in the following statements, terms such as $\Dec\circ\KeyR\circ\Enc$ should be understood as having each procedure acting only on the relevant registers, e.g.~$\KeyR$ only acts on the private-key output of $\Enc$):
\begin{equation}
\label{eq:complete}
\Pr\left[\left(F = \text{\cmark}\right)\land\left(\Dec\circ\KeyR\circ\Enc(1^\lambda, m) = m\right)\right] \geq 1 - \negl(\lambda),
\end{equation}
i.e.~the probability of accepting and correctly decrypting the message is high.

Additionally, we impose the condition that for any distribution of a message $M$, if we let $\sigma_{MQ|F=\text{\xmark}}$ be the state produced by $\Enc(1^\lambda, M)$ on registers $MQ$ conditioned on $F=\text{\xmark}$, then we have\footnote{Here we have differed very slightly from~\cite{GMP22}, in that for their protocol, when $F=\text{\xmark}$ the protocol basically stops and no state is produced with the receiver (although in their actual protocol the receiver of course ends up with some quantum state in either case --- the state just does not depend on $m$ when $F=\text{\xmark}$, though it does depend on the key value). The definition we have used is essentially saying the same thing, i.e.~$\rho$ does not depend on $m$ if $F=\text{\xmark}$. In the protocol we design below, when $F=\text{\xmark}$ the receiver gets the state $\rho$ they would have gotten if $F=\text{\cmark}$ and the message were some uniformly random ``dummy value'' $m^\mathrm{fake}$ which is independent of $m$. 
We do this instead of aborting the protocol
to simplify the security proof, because this way we do not have to analyze the receiver behaving differently conditioned on the value of $F$.
} 
\begin{equation}\label{eq:prodstate}
\sigma_{MQ|F=\text{\xmark}} = \sigma_{M|F=\text{\xmark}} \otimes \sigma_{Q|F=\text{\xmark}},
\end{equation}
i.e.~when $\Enc(1^\lambda, M)$ aborts, the receiver's state is independent of $M$.
\end{definition}

\begin{remark}
In various somewhat similar tasks such as QKD~\cite{PR14} 
or certified deletion with a third-party eavesdropper~\cite{KT23}, when considering dishonest behaviour we typically also require a \emph{correctness} condition along the following lines: even when the devices are dishonest, the probability that the message is incorrectly decrypted {and} the protocol accepts is low. In our context however, the only potentially dishonest party is the recipient, and hence it does not make sense to bound this probability for dishonest behaviour: the set of dishonest-receiver behaviours always includes trivial processes where they simply set some random value as the ``decrypted message'', in which case it is clearly impossible to give any nontrivial bounds on the probability of incorrectly decrypting.
(On the other hand, when focusing on the case of honest behaviour, note that the completeness condition~\eqref{eq:complete} indeed incorporates the requirement that the probability of incorrectly decrypting the message is low.)
\end{remark}

We now state the security definitions we use, which are the same as in~\cite{GMP22} except with minor modifications to account for the difference in our decryption procedures.
First, we aim to capture the notion that an adversary without the decryption key cannot distinguish the message from a ``dummy'' value:

\begin{definition}[Distinguishing attack and indistinguishable security]\label{def:indist-sec}
A \emph{distinguishing attack} on a VKECM scheme is a 
process of the following form (here for ease of explanation we take the message space $\clM$ to contain a particular value labelled as $\mathbf{0}$ without loss of generality): 
\begin{enumerate}
\item An adversary generates a state on registers $ME$, where $M$ is a classical register on the message space and $E$ is some (possibly quantum) side-information.
\item A uniformly random bit $B$ is independently generated, and used as follows (in which the receiver can behave dishonestly when implementing $\Enc$): 
if $B=0$ then $\Enc(1^\lambda, \mathbf{0})$ is performed; if $B=1$ then $\Enc(1^\lambda, M)$ is performed on the $M$ register produced by the adversary in the previous step. In either case, a flag $F$, a private key $\Kpriv$, and a ciphertext state $\rho$ (on a register $Q$) are produced.
\item A measurement is performed on the state on $QE$ to produce a single bit $\hat{B}$.
\end{enumerate}
A protocol is said to be \emph{indistinguishable-secure} if for all distinguishing attacks, we have 
\begin{align*}
\Pr[(F=\text{\cmark}) \land (B=\hat{B})] \leq \frac{1}{2} + \negl(\lambda),
\end{align*}
where the probability is taken over all randomness in the described procedures.
\end{definition}

Technically, our DI-VKECM protocol in fact satisfies a stronger form of indistinguishability, namely that the ciphertext state is completely independent of the message (because our DI-VCECM protocol basically involves applying a one-time-pad to the message, which serves to perfectly encrypt it). However, we present the definition in the above form for consistency with past work, and also because it is less clear how to formulate an analogous property for cloning and cloning-distinguishing attacks, which we now turn to.

Next, we aim to capture the idea that an adversarial receiver cannot clone the ciphertext such that two parties can later decrypt the message (after receiving decryption keys) without communication:
\begin{definition}[Cloning attack and uncloneable security]\label{def:cl-sec}
A \emph{cloning attack} on a VKECM scheme is a process of the following form:
\begin{enumerate}
\item A uniformly random message $M\in\clM$ is prepared and $\Enc$ is applied to it (with a potentially dishonest receiver), producing a flag $F$, a private key $\Kpriv$, and a ciphertext state $\rho$ (on a register $Q$).
\item An arbitrary channel is applied to the ciphertext state $\rho$ to distribute it between two parties Bob and Charlie. 
\item Without any further communication between Bob and Charlie except as mediated by leakage via their devices, they receive independently generated decryption keys from $\KeyR(\Kpriv)$, and use them together with their shares of the state to produce guesses $M^\B$ and $M^\C$ respectively for the original message $M$.
\end{enumerate}
A protocol is said to be \emph{$(t(\lambda),g(\lambda))$-uncloneable-secure} if for all cloning attacks, we have
\begin{align*}
\Pr[(F=\text{\cmark}) \land (M=M^\B=M^\C)] \leq \frac{2^{t(\lambda)}}{|\clM|} + g(\lambda),
\end{align*}
where the probability is taken over the distributions of $M$ and all randomness in the described procedures.
\end{definition}
Qualitatively, the smaller the functions $t(\lambda)$ and $g(\lambda)$ are, the ``more secure'' the protocol is (since the probability of Bob and Charlie guessing the message is smaller). Note that if the protocol also requires that the message is a bitstring of length $\lambda$ (for instance in the first protocol in~\cite{BL20}, or the version of our protocol that we describe in Section~\ref{sec:DIVKECM} here), then any $t(\lambda)$ such that $t(\lambda) \geq \lambda$ would be a trivial statement (assuming $g(\lambda)$ is non-negative), since in that case the bound on the guessing probability in the above definition would reach the trivial value of $1$. In other words, for such protocols we only have nontrivial results when $t(\lambda) < \lambda$.

\begin{remark}\label{remark:termsplit}
In previous work such as~\cite{BL20}, the uncloneability definition only required specifying a single function $t(\lambda)$ such that an upper bound of the form $\frac{2^{t(\lambda)}}{|\clM|} + \negl(\lambda)$ holds. However, a reviewer has pointed out that such a formulation is not well-posed for the protocols mentioned above where the message is a bitstring of length $\lambda$: given any such protocol that achieves $t(\lambda) = \tau\lambda$ for some constant $\tau\in(0,1)$, the $\frac{2^{t(\lambda)}}{|\clM|}$ term would itself be a negligible function of $\lambda$, making the split into the $\frac{2^{t(\lambda)}}{|\clM|}$ and $\negl(\lambda)$ contributions somewhat ill-defined. We thank the reviewer for pointing out this problem with the previous formulation. For this work, we have chosen to resolve it in our definitions by requiring a specification of \emph{both} the functions $t(\lambda)$ and $g(\lambda)$ in the bound $\frac{2^{t(\lambda)}}{|\clM|} + g(\lambda)$. It might perhaps be possible to find a more concise way to address this issue, but we leave this for future work, as it does not affect the qualitative implications of any of our results.
\end{remark}

While the definition of uncloneable security refers to uniformly distributed messages, satisfying the above definition implies analogous properties for messages that are not uniformly distributed; see~\cite{BL20}.

Finally, for completeness we state a security definition from~\cite{BL20} that captures some aspects of both of the preceding properties (though in this work, we will not discuss this security definition in much detail; see Remark~\ref{remark:defnrelations} below):
\begin{definition}[Cloning-distinguishing attack and uncloneable-indistinguishable security]\label{def:cl-indist-sec}
A \emph{cloning-distinguishing attack} on a VKECM scheme is a process of the following form (here for ease of explanation we take the message space $\clM$ to contain a particular value labelled as $\mathbf{0}$ without loss of generality):
\begin{enumerate}
\item An adversary generates a state on registers $ME$, where $M$ is a classical register on the message space and $E$ is some (possibly quantum) side-information.
\item A uniformly random bit $B$ is independently generated, and used as follows (in which the receiver can behave dishonestly when implementing $\Enc$): 
if $B=0$ then $\Enc(1^\lambda, \mathbf{0})$ is performed; if $B=1$ then $\Enc(1^\lambda, M)$ is performed on the $M$ register produced by the adversary in the previous step. In either case, a flag $F$, a private key $\Kpriv$, and a ciphertext state $\rho$ (on a register $Q$) are produced.
\item An arbitrary channel is applied to the state on $QE$ to distribute it between two parties Bob and Charlie. 
\item Without any further communication between Bob and Charlie except as mediated by leakage via their devices, they receive independently generated decryption keys from $\KeyR(\Kpriv)$, and use them together with their shares of the state to produce bits $B'$ and $B''$ respectively.
\end{enumerate}
A protocol is said to be \emph{uncloneable-indistinguishable-secure} if for all cloning attacks, we have
\begin{align*}
\Pr[(F=\text{\cmark}) \land (B=B'=B'')] \leq \frac{1}{2} + \negl(\lambda),
\end{align*}
where the probability is taken over all randomness in the described procedures.
\end{definition}

Following~\cite{BL20,GMP22}, we have stated the above definitions of indistinguishable security and uncloneable-indistinguishable security in terms of the adversary initially generating an arbitrary classical-quantum state on $ME$, which we shall denote here as $\rho^\mathrm{ini}_{ME}$. However, one can argue that without loss of generality, we can restrict to attacks where the value on the $M$ register is deterministic and the state on the $E$ register is trivial (as long as the final measurement in each of the attacks is allowed to be a general POVM). This follows from the following observation: consider any attack that achieves the optimal ``success probability'' (in the sense of the probability of the event stated at the end of the respective definition). This attack would be based on some initial state $\rho^\mathrm{ini}_{ME}$, which is a classical-quantum state and hence equivalent to a classical mixture of states of the form $\state{m}_M \otimes \rho_{E|m}$. Because of this, the success probability of this attack is a convex combination of the success probabilities that would be obtained by preparing the initial state in the form $\state{m}_M \otimes \rho_{E|m}$ for various values of $m$. By linearity, at least one particular value $m^\star$ must attain the optimal success probability\footnote{Furthermore, by observing that if $\rho^\mathrm{ini}_{ME}=\state{\mathbf{0}}_M \otimes \rho_{E|\mathbf{0}}$ then the success probability can only be exactly $1/2$ (since in that case the states produced for either value of $B$ are identical), we see that we can take $m^\star \neq \mathbf{0}$ without loss of generality.}, or in other words the same success probability could have been attained by the adversary simply preparing the initial state in the form $\state{m^\star}_M \otimes \rho_{E|m^\star}$. With this attack the state on $M$ is deterministic as claimed; furthermore, since the state on $E$ is now in product with $M$, we can simply suppose that the state $\rho_{E|m^\star}$ is generated after $\Enc$ is applied (and hence absorbed into the subsequent measurements/channels) rather than at the beginning of the attack. 

It was shown in~\cite{BL20} that the three security properties listed above are somewhat related to each other, as follows:
\begin{lemma}\label{fc:uncind_ind}
Uncloneable-indistinguishability implies indistinguishability.
\end{lemma}
\begin{lemma}\label{fc:0unc_uncind}
If the message space size $|\clM|$ is independent of the security parameter $\lambda$, then $(0,\negl(\lambda))$-uncloneability implies uncloneable-indistinguishability.
\end{lemma}
While the definitions they use differ slightly from ours (because they consider protocols which do not have an abort outcome and do not use variable decryption keys), their arguments carry over straightforwardly to our scenario; we outline the main ideas here.
\begin{proof}[Proof sketch]
For Lemma~\ref{fc:uncind_ind}, \cite{BL20} observe that given any distinguishing attack, one can immediately construct a cloning-distinguishing attack that succeeds with the same probability: simply perform the distinguishing attack to produce the classical bit $\hat{B}$, and distribute copies of this value to Bob and Charlie, who output it as their values $B',B''$ respectively in the cloning-distinguishing attack. Hence the optimal success probability of a distinguishing attack cannot be higher than that for cloning-distinguishing attacks, and referring back to the definitions we see that this means uncloneable-indistinguishability implies indistinguishability.

For Lemma~\ref{fc:0unc_uncind}, the idea is again similar: given a cloning-distinguishing attack that succeeds with some probability $p\in[0,1]$, \cite{BL20} prove that one can construct a cloning attack that succeeds with probability at least $\frac{2}{|\clM|}p$. (We remark that our above observation regarding the structure of optimal cloning-distinguishing attacks can somewhat simplify this proof in~\cite{BL20}, since without loss of generality we can suppose that the cloning-distinguishing attack starts by simply setting $M$ to a deterministic value.) This implies that the optimal success probability of a cloning-distinguishing attack is at most $\frac{|\clM|}{2}$ times greater than that of a cloning attack, and since $(0,\negl(\lambda))$-uncloneability is the statement that the latter is upper bounded by $\frac{1}{|\clM|} + \negl(\lambda)$, this gives the desired result as long as $|\clM|$ is independent of $\lambda$. 
To see that this approach also applies for our definitions, one simply has to instead consider the states conditioned on $F=\text{\cmark}$, in which case the \cite{BL20} construction gives
\[
\Pr[ M=M^\B=M^\C | F=\text{\cmark}] \geq \frac{2}{|\clM|} \Pr[ B=B'=B'' | F=\text{\cmark}],
\]
where the 
right-hand-side refers to the probabilities in the cloning attack, and the 
left-hand-side refers to those in the cloning-distinguishing attack constructed from it.
Multiplying both sides of this inequality by $\Pr[F=\text{\cmark}]$ and then following the same argument as above gives the desired result for our context.
\end{proof}

\begin{remark}\label{remark:defnrelations}
Due to the above reductions, we see that for protocols with $\clM$ independent of $\lambda$, proving $(0,\negl(\lambda))$-uncloneability would be sufficient to imply the other two properties as well. 
However, in this work, for messages of arbitrary size (but independent of $\lambda$) we only prove $(O(\lambda),\negl(\lambda))$-uncloneability and hence we cannot use this simplification. As for single-bit or single-trit messages, if our result in Section~\ref{sec:bittritPA} on ``extractors'' against two adversaries in the single-decryptor encryption setting could be extended to the VKECM setting, it might allow a proof of $(0,\negl(\lambda))$-uncloneability for such messages; however, we were unable to resolve this question within this work (see Remark~\ref{remark:OTPreuse}) and hence leave it for future investigation.
\end{remark}

\subsection{Single-decryptor encryption}
\label{sec:defSDECM}

We now describe the security definition we use for device-independent single-decryptor encryption.\footnote{Similar to the situation for variable-key encryption above, while we use the terms ``encryption key'' and ``decryption key'' in the following descriptions, we stress that they do not correspond to the public key and private key respectively in a public-key encryption scheme. In particular, the encryption key in the scheme we design happens to also allow decrypting the ciphertext, so it is certainly not usable as a public-key encryption scheme.} Like in Definition \ref{def:vkecm}, the syntax will be similar to those in \cite{GZ20, CLLZ21} except the key generation procedure will need to interactive, with the possibility of aborting, in order to achieve device-independence. The notion of security we shall use will 
be essentially the same as the definition of \emph{random challenge anti-piracy} in~\cite{CLLZ21} (Definition~6.5), apart from the modification to account for the abort outcome.

\begin{definition}[Device-independent single-decryptor encryption of classical messages]\label{def:sdecm}
Let $\lambda$ be a security parameter. A scheme for \emph{single-decryptor encryption of classical messages} (SDECM) consists of a tuple $(\KeyG, \Enc, \Dec)$ such that
\begin{itemize}
\item $\KeyG(1^\lambda)$ is a (potentially interactive) protocol between an honest client and a potentially dishonest receiver, both of whom take as input the security parameter $\lambda$. The output of the protocol is $(F, \Kenc, \rho)$, where $F$ is a flag held by the client that takes values \cmark (accept) and \xmark (reject), $\Kenc$ is a classical encryption key held by the client, and $\rho$ is a quantum decryption key state in a register $Q$ held by the receiver.
\item $\Enc(m, \kenc, f)$ takes as input a message $m$ in the message space $\clM$, an encryption key $\kenc$, a flag value $f \in \{\text{\cmark}, \text{\xmark}\}$, and uses additional internal randomness to produce a classical ciphertext $C$.
\item $\Dec(c, \rho)$ takes as input a ciphertext $c$, a decryption key $\rho$, and outputs a message $\tM \in \clM$.
\end{itemize}
We say an SDECM scheme is efficient if the computations that the client and the honest receiver perform in $\KeyG, \Enc, \Dec$ are polynomial in $\lambda$ and the length of $m$.

We require that the SDECM scheme satisfies completeness, i.e., if all steps are carried out honestly, then
\begin{equation}\label{eq:complete-2}
\Pr\left[(F=\text{\cmark})\land\left(\Dec\circ\Enc(m, \KeyG(1^\lambda))=m\right)\right] \geq 1 - \negl(\lambda),
\end{equation}
where it is understood that $\Enc$ uses the encryption key and flag produced by $\KeyG$, and $\Dec$ uses the decryption key state.

Additionally, 
we impose the condition that for any value of the encryption key $\kenc$, the distribution of ciphertexts produced by $\Enc(m, \kenc, \text{\xmark})$ is independent of $m$; i.e.~the ciphertext is independent of the message when $F=\text{\xmark}$.
\end{definition}
SDECM schemes will be required to satisfy indistinguishable security in basically the same way as VKECM, and we shall not define it separately. We now focus on presenting definitions of anti-piracy security (against random challenge plaintexts), based on notions presented in~\cite{CLLZ21} (further discussion below).

\begin{definition}[Pirating attack and anti-piracy security, for identical random challenge plaintexts]\label{def:pir-sec}
A \emph{pirating attack} on a SDECM scheme is a process of the following form:
\begin{enumerate}
\item $\KeyG(1^\lambda)$ is carried out between the client and (dishonest) receiver, producing a flag $F$, an encryption key $\Kenc$, and a quantum decryption key $\rho$.
\item An arbitrary channel is applied to the decryption key state $\rho$ to distribute it between two parties Bob and Charlie. 
\item A uniformly random message $M\in\clM$ is prepared. Without any further communication between Bob and Charlie except as mediated by leakage via their devices, they receive independently generated instances of $\Enc(M,\Kenc,F)$, and use them together with their shares of the state to produce guesses $M^\B$ and $M^\C$ respectively for the message $M$.
\end{enumerate}
A protocol is said to be \emph{$(t(\lambda),g(\lambda))$-anti-piracy-secure against identical random challenge plaintexts} if for all such pirating attacks, we have\footnote{Technically, in~\cite{CLLZ21} the level of security was instead quantified by writing the upper bound in the form $\frac{1}{|\clM|} + \gamma(\lambda) + \negl(\lambda)$ and specifying the function $\gamma(\lambda)$. Here, for a closer analogy to Definition~\ref{def:cl-sec} (which was based on~\cite{BL20}), we have written it in a form similar to the one used in that definition; the definitions are interconvertible by taking $\gamma(\lambda) = \frac{2^{t(\lambda)}-1}{|\clM|}$ or inversely $t(\lambda) = \log(|\clM|\gamma(\lambda) + 1)$. Furthermore, we have again required an explicit specification of the function $g(\lambda)$ added to $\frac{2^{t(\lambda)}}{|\clM|}$ rather than just requiring it to be negligible, due to the issues pointed out in Remark~\ref{remark:termsplit}.}
\begin{align*}
\Pr[(F=\text{\cmark}) \land (M=M^\B=M^\C)] \leq \frac{2^{t(\lambda)}}{|\clM|} + g(\lambda),
\end{align*}
where the probability is taken over the distributions of $M$ and all randomness in the described procedures.
\end{definition}
\begin{definition}[Pirating attack and anti-piracy security, for independent random challenge plaintexts]\label{def:pir-sec-ind}
The pirating attack as the same as in Definition \ref{def:pir-sec} except in step 3, two independent and uniform messages $M_1$ and $M_2$ are sampled from $\clM$. Bob and Charlie receive independently generated instances of $\Enc(M_1, \Kenc, F)$ and $\Enc(M_2, \Kenc, F)$ respectively, and produce guesses $M^\B$ and $M^\C$ for $M_1$ and $M_2$. The protocol is said to be \emph{$(t(\lambda),g(\lambda))$-anti-piracy-secure against independent random challenge plaintexts} if for all such pirating attacks, we have
\begin{align*}
\Pr[(F=\text{\cmark}) \land (M_1=M^\B)\land(M_2=M^\C)] \leq \frac{2^{t(\lambda)}}{|\clM|} + g(\lambda),
\end{align*}
\end{definition}
If an SDECM scheme with $\clM$ independent of $\lambda$ achieves $t(\lambda)=0$ and $g(\lambda)=\negl(\lambda)$ under one of the above definitions, we may sometimes qualitatively refer to it as having ``perfect'' anti-piracy security (with respect to the corresponding definition).

Regarding the definitions above (which we based on~\cite{CLLZ21}), we highlight that there were some technical differences between the piracy setups studied in~\cite{GZ20} compared to~\cite{CLLZ21}. The first difference is in whether Bob and Charlie receive identical copies of a single output of $\Enc(M,\Kenc,F)$, versus independently generated outputs of that procedure (put another way: whether the same randomness or independent randomness is used to generate the ciphertexts they receive); here we follow~\cite{CLLZ21} and use the latter, as we currently require that property in our security proof. Definition \ref{def:pir-sec-ind} follows \cite{CLLZ21} and has Bob and Charlie also receive encryptions of two independent uniform messages $M_1, M_2$ (which they have to respectively guess). Definition \ref{def:pir-sec} is similar to \cite{GZ20} in that it gives Bob and Charlie encryptions of the same message (although independent randomness is still used in the encryption), and this is the definition that has a closer analogy to the uncloneability definition for DI-VKECM.
For the protocols we present in this work, we will specify in each case which of the above definition(s) they satisfy.

\cite{CLLZ21} also considered another different notion of anti-piracy security that they called CPA-style security. Here instead of the message to be encrypted being randomly sampled, the adversary chooses two distinct messages $(m_0, m_1)$ which could potentially be encrypted. A uniformly random choice is made between these two messages, and then Bob and Charlie have to guess which of the two messages have been encrypted.\footnote{Here again one can have variations depending on whether the message to be encrypted is chosen independently or identically for Bob and Charlie.} The security requirement is that their joint guessing probability should be close to $\frac{1}{2}$. We shall not be considering the CPA-style security definition in this work, but we point out that the CPA-style definition and the random ciphertext definition with $t(\lambda)=0$ are equivalent for single-bit messages (since a uniform choice is made between the 0 message or the 1 message in either case). Since in this work we present (in Section~\ref{sec:bittritPA}) a protocol that indeed achieves $(0,\negl(\lambda))$-anti-piracy security for single-bit messages, it follows that it also achieves CPA-style security for single-bit messages.

\section{Cloning game}\label{sec:cl-game}
For any non-local game $G$ (which includes games with multiple rounds and with a variable number of players, like we shall soon describe), we shall use $\omega^*(G)$ to denote its quantum value, i.e., its maximum winning probability with a quantum strategy. We shall use $\omega^*(G^l)$ to denote the winning probability of $l$ parallel copies of $G$, and $\omega^*(G^{t/l})$ to denote the winning probability of $t$ copies out of $l$ parallel copies of $G$.


\subsection{The CHSH game}
The CHSH game is one of the simplest one-round non-local games between two players Alice and Bob, which is as follows:
\begin{itemize}
\item Alice and Bob get inputs $x, y \in \{0,1\}$ uniformly at random.
\item Alice and Bob output $a, b \in \{0,1\}$.
\item The game is won if $a_i\oplus b_i = x_i\cdot y_i$.
\end{itemize}
The optimal classical winning probability $\omega(\CHSH)$ of the game is $\frac{3}{4}$, while the optimal quantum winning probability is $\omega^*(\CHSH)=\frac{1}{2}\left(1 + \frac{1}{\sqrt{2}}\right)$. In the optimal quantum strategy for the CHSH game, Alice and Bob share the maximally entangled state $ \frac{1}{\sqrt{2}}(\ket{00} + \ket{11})$. Alice's measurements corresponding to inputs 0 and 1 are $\{\state{0},\state{1}\}$ and $\{\state{\frac{\pi}{4}},\state{-\frac{\pi}{4}}\}$ respectively, and Bob's measurements corresponding to 0 and 1 are $\{\state{\frac{\pi}{8}},\state{-\frac{3\pi}{8}}\}$ and $\{\state{\frac{3\pi}{8}}, \state{-\frac{\pi}{8}}\}$ respectively, where we are using $\ket{\alpha}$ to denote the state $\cos\alpha\ket{0}+\sin\alpha\ket{1}$. We shall often refer to these state and measurements as the ideal CHSH state and measurements.

The CHSH game satisfies the following \emph{rigidity} property (given by a slight extension of Theorem~2 in~\cite{MYS12},
by writing the projectors $\Pi, \widetilde{\Pi}$ in the following statement as linear combinations of the hermitian observables described in that work).
\begin{fact}\label{fc:CHSH-rig}
Let $\ket{\phi}, \Pi^\A_{a|x}, \Pi^\B_{b|y}$ denote the state and measurements used in the ideal CHSH strategy. Suppose a quantum strategy for the CHSH game with shared state $\ket{\widetilde{\phi}}$ and projective measurements $\widetilde{\Pi}^\A_{a|x}$ and $\widetilde{\Pi}^\B_{b|y}$ of Alice and Bob achieves winning probability $\omega^*(\CHSH) - \mu$. Then there exist isometries $V^\A, V^\B$ acting only on Alice and Bob's registers in $\ket{\rho}$, and a state $\ket{\junk}$ such that for all $x, y, a, b$,
\begin{gather*}
\norm{V^\A\otimes V^\B\ket{\widetilde{\phi}} - \ket{\phi}\otimes\ket{\junk}}_2 \leq O(\mu^{1/4});\\
\norm{(V^\A\otimes V^\B)(\widetilde{\Pi}^\A_{a|x}\otimes\Id)\ket{\widetilde{\phi}} - (\Pi^\A_{a|x}\otimes\Id)\ket{\phi}\otimes\ket{\junk}}_2 \leq O(\mu^{1/4}); \\
\norm{(V^\A\otimes V^\B)(\Id\otimes\widetilde{\Pi}^\B_{b|y})\ket{\widetilde{\phi}} - (\Id\otimes\Pi^\B_{b|y})\ket{\phi}\otimes\ket{\junk}}_2 \leq O(\mu^{1/4}).
\end{gather*}
\end{fact}

Alice's ideal measurements in the CHSH game are in the computational and Hadamard basis. Since the ideal shared state between Alice and Bob is the maximally entangled state, Bob's state when Alice does the measurement corresponding to $x=0$ or 1 and obtains outcome $a=0$ or 1 is $H^x\ket{a}$. These states are called Wiesner states, and we shall denote the Wiesner state produced when $x$ is the input and $a$ is the output by $\ket{a^x}$. The Wiesner states satisfy the following monogamy of entanglement property, which we shall use in our security proof.
\begin{fact}[\cite{TFK+13,BL20}]\label{fc:mon-game}
If $\ket{a^x}$ denotes a Wiesner state, then for any Hilbert spaces $\clH^\B$ and $\clH^\C$, any two collection of measurements $\{\Pi^\B_{a|x}\}_a$ and $\{\Pi^\C_{a|x}\}_a$ (for each $x$) on $\clH^\B$ and $\clH^\C$ respectively, and any CPTP map $\Lambda: \bbC^2 \to \clH^\B\otimes\clH^\C$, we have
\[ \bbE_{a,x}\Tr\left[\left(\Pi^\B_{a|x}\otimes\Pi^\C_{a|x}\right)\Lambda(\state{a^x})\right] = \frac{1}{2}\left(1 + \frac{1}{\sqrt{2}}\right).\]
\end{fact}
Note that we can also define $n$-qubit versions of the Wiesner states, and the above result holds with the right-hand side being $\left(\frac{1}{2} + \frac{1}{2\sqrt{2}}\right)^n$ in that case. However, we shall only need the result for single-qubit Wiesner states for our purposes.

\subsection{The 2-round cloning game}
\label{subsec:simplegame-def}
Using the rigidity or self-testing property of the CHSH game and the monogamy of entanglement property of the Wiesner states, we shall formulate a game that we shall call the 2-round cloning game, which we first qualitatively describe as follows. In a single instance of the game, one of two things will happen probabilistically: either the CHSH game will be played between two players whom we call Alice and Barlie, or Alice will get her input for the CHSH game and produce her output as usual (without knowing what is happening on Barlie's side), but Barlie will split into two players Bob and Charlie, who will both be given the same input as Alice and have to guess her output bit. For technical reasons, we shall actually split the game into two rounds, with the CHSH component happening in the first round and Bob and Charlie guessing Alice's output in the second.

The measurements used by Barlie for the CHSH component may be different from those used by Bob and Charlie for the guessing component. However, Alice's device does not know which component is taking place, and the shared state between Alice and Barlie or Alice, Bob and Charlie is distributed beforehand, and therefore if the CHSH component is won with probability close to $\omega^*(\CHSH)$, then by the self-testing property of CHSH, the shared state and Alice's measurements must be close to the ideal state and measurements. Thus, the state on Bob and Charlie's side post-Alice's measurement must be close to a Wiesner state, which will allow us to use the monogamy of entanglement property to upper bound the probability of Bob and Charlie both guessing Alice's output. We shall define the 2-round cloning game formally below, and formalize the above argument about its winning probability in Section \ref{subsec:cl-win}.

We now give the detailed description.
The 2-round cloning game $\cl_\gamma$ with parameter $\gamma$ involves four players Alice, Barlie, Bob and Charlie, although not all of them have to perform actions in each round. 
At the beginning of the game, there can be some arbitrary entangled state shared between Alice and Barlie.\footnote{For the security analysis, it might seem more general to allow an initial entangled state across all four parties. However, as will become clear from the later description, this does not make any difference since any registers Bob and Charlie might have started with could also be analyzed by initially giving them to Barlie instead.}
The first round only involves the two players Alice and Barlie, who receive some inputs and produce some outputs (without communication). Specifically, the first round inputs, outputs and the winning condition are as follows:
\begin{itemize}
\item Alice receives $x\in\{0,1\}$ uniformly at random. Barlie receives $u\in \{0,1,\mathrm{keep}\}$, such that 
$u=\mathrm{keep}$ with probability $1-\gamma$, and $u=0$ or $1$ with probability $\frac{\gamma}{2}$ each.
\item Alice outputs $a$ and Barlie outputs $s$.
\item The first round win condition is
\begin{equation*}
\sfV_1(x,u,a,s) = \begin{cases} 1 & \text{ if } u = \mathrm{keep} \\ 1 & \text{ if } (u \neq \mathrm{keep})\land(a\oplus s = x\cdot u) \\
0 & \text{ otherwise.}\end{cases} 
\end{equation*}
\end{itemize} 
After Barlie produces his first round output, he can do some further operation on his part of the shared state (though this can also be absorbed into the operation he does to produce $s$). Then Barlie distributes his state in some arbitrary fashion between two other players, Bob and Charlie. Bob and Charlie do not communicate after this (and Barlie no longer plays any further role in the game, apart from having his first-round values $(u,s)$ being involved in the win condition). In the second round, Alice and Barlie do not receive any inputs and are not required to produce any output; Bob and Charlie receive inputs and have to produce outputs separately. The inputs, outputs and win condition in the second round are as follows:
\begin{itemize}
\item Bob and Charlie both receive $x$ as their input and produce $b$ and $c$ as their outputs.
\item The second round win condition is
\[ \sfV_2(x,u,a,s,b,c) = \begin{cases} 1 & \text{ if } u \neq \mathrm{keep} \\ 1 & \text{ if } (u = \mathrm{keep})\land(b=c=a) \\ 0 & \text{ otherwise.} \end{cases} \]
\end{itemize}

For a schematic depiction of the structure of this game, refer to Fig.~\ref{fig:clone} below (which technically depicts the modified game we describe in the next section, but the overall structure is fairly similar).

\subsection{Modifying the 2-round cloning game}\label{subsec:game-def}

Although it would be nice to be able to use $\cl_\gamma$ directly in our encryption scheme and security proof, in order to be able to prove a parallel repetition theorem, we need to modify the game somewhat, by an ``anchoring" transformation similar to \cite{Vid17, KT23}. In our case, this means that Alice's input $x$ will not be revealed with some probability $\alpha$ independently to Bob and Charlie; we can let the second round win condition be automatically satisfied for Bob and Charlie if 
either of them does not
receive $x$ (this will be denoted by the input being $\bot$, and in the protocol will correspond to this instance not being used for key generation). This is needed so that Bob and Charlie's second round inputs and Alice's registers are in product with some probability.

The modified 2-round cloning game $\cl_{\gamma,\alpha}$ is the same as $\cl_\gamma$ in the first round, and after the first round, Barlie again distributes his state between Bob and Charlie in the same fashion as $\cl_\gamma$. (Refer to Fig.~\ref{fig:clone} for a schematic diagram of this structure.) However, the second round inputs, outputs and win condition for $\cl_{\gamma,\alpha}$ are instead as follows:
\begin{itemize}
\item Bob and Charlie receive $y, z \in \{0,1, \bot\}$ respectively, such that $y=\bot$ and $z=\bot$ independently with probability $\alpha$, and when $y$ and $z$ are not $\bot$, they are equal to $x$.
\item Bob and Charlie output $b$ and $c$ respectively.
\item The second round win condition is
\begin{equation*}
\sfV_2(x,u,y,z,a,s,b,c) = \begin{cases} 1 & \text{ if } u \neq \mathrm{keep} \\ 1 & \text{ if } (u=\mathrm{keep})\land(y=z=\bot) \\ 1 & \text{ if } (u=\mathrm{keep})\land(y=\bot\neq z) \\ 1 & \text{ if } (u=\mathrm{keep})\land(z=\bot\neq y) \\
1 & \text{ if } (u = \mathrm{keep})\land(y, z \neq \bot)\land(b=c=a) \\ 0 & \text{ otherwise.} \end{cases}
\end{equation*}
\end{itemize}

For later use, we summarize the overall win condition $\sfV_1(x,u,a,s)\cdot\sfV_2(x,u,y,z,a,s,b,c)=1$ as follows:
\begin{align}
\sfV_1(x,u,a,s)\cdot\sfV_2(x,u,y,z,a,s,b,c) = 
\begin{cases} 
1 & \text{ if } (u \neq \mathrm{keep})\land(a\oplus s = x\cdot u) \\ 
1 & \text{ if } (u=\mathrm{keep})\land(y=z=\bot) \\ 
1 & \text{ if } (u=\mathrm{keep})\land(y=\bot\neq z)\\ 
1 & \text{ if } (u=\mathrm{keep})\land(z=\bot\neq y)\\
1 & \text{ if } (u = \mathrm{keep})\land(y, z \neq \bot)\land(b=c=a) \\ 0 & \text{ otherwise.} 
\end{cases} 
\label{eq:wincondfull}
\end{align}
As qualitatively described above, the various possible win conditions when $u=\mathrm{keep}$ (which arise mainly from the second-round conditions, since the first round is trivially won for that case) can be equivalently rewritten into a single condition
\begin{align*}
(y=\bot)\lor(z=\bot)\lor(b=c=a),
\end{align*}
i.e.~(at least) one of Bob and Charlie got $\bot$ as input, or they both guessed $a$ correctly.

\begin{figure}
\centering
\leavevmode
\large 
\Qcircuit @C=1.5em @R=1.2em {
& \barrier[-3.1em]{4} & \text{First round} & \barrier[-3.5em]{4} & & \hspace{-3.5em}\text{Second round} &\\
\lstick{\text{Alice}}
&\qw & \gate{M_{a|x}} & & & &\\
\lstick{\text{Barlie}}
&\qw & \gate{M'_{s|u}} & \qw & \multigate{2}{\clE_{\mathrm{distribute}}} & & \\
\lstick{\text{Bob}} & & & & \nghost{\clE_{\mathrm{distribute}}} & \qw & \gate{M''_{b|y}} \\
\lstick{\text{Charlie}} & & & & \nghost{\clE_{\mathrm{distribute}}} & \qw & \gate{M'''_{c|z}} \\
}
\caption{Circuit diagram depicting the game $\cl_{\gamma,\alpha}$. Initially, some entangled state is shared between the parties Alice and Barlie. In the first round, Alice measures her register according to her input $x$ to produce an output $a$, which we depict in the above diagram via its POVM elements $M_{a|x}$; analogously, Barlie measures his register according to some other POVM $M'_{s|u}$. While Alice performs no further operations (so she can discard her state immediately after measuring), Barlie retains his post-measurement register for the second round. In the second round, he redistributes his register to two other parties Bob and Charlie in some arbitrary fashion, which we denote with the channel $\clE_{\mathrm{distribute}}$, and Bob and Charlie then respectively perform measurements specified by inputs $y$ and $z$, producing outputs $b$ and $c$, which we again represent via POVM elements $M''_{b|y}$ and $M'''_{c|z}$. Note that the simpler game $\cl_\gamma$ described previously in Section~\ref{subsec:simplegame-def} also has an essentially similar structure, except that it has a different input distribution (for instance Bob and Charlie's inputs are basically set to $y=z=x$), and a different win condition.}
\label{fig:clone}
\end{figure}

In our actual parallel repetition proof, we shall not need to use any specific structure of $\cl_{\gamma,\alpha}$ aside from its input distribution, and thus the parallel repetition result holds for a more general class of games. The specific properties of the input distribution that we shall need to use are:
\begin{enumerate}[(i)]
\item The distribution $\sfP_{XU}$ of the first round inputs is a product distribution; $U$ is also independent of the second round inputs. \label{prop:dist-1}
\item The second round inputs $Y, Z$ are correlated with $X$ in the following way:
\begin{gather*}
\sfP_{XYZ}(x,x,x) = (1-\alpha)^2\cdot \sfP_X(x) \\
\sfP_{XYZ}(x,\bot,x) = \sfP_{XYZ}(x,x,\bot) = \alpha(1-\alpha)\cdot\sfP_X(x) \\
\sfP_{XYZ}(x,\bot,\bot) = \alpha^2\cdot\sfP_X(x).
\end{gather*}
 \label{prop:dist-2}
\end{enumerate}
Note that the above implies that conditioned on $Y=Z=\bot$, the conditional distribution $\sfP_{XU||\bot,\bot}$ of $x,u$ is exactly the same as the marginal distribution of $x, u$. Similarly, conditioned on $Y=\bot$, the conditional distribution $\sfP_{XUZ|Y=\bot}$ is the same as the marginal distribution of $x, u, z$, and conditioned on $Z=\bot$, the conditional distribution $\sfP_{XUY|Z=\bot}$ is the same as the marginal distribution of $x, u, y$. Moreover, $Y$ and $Z$ are independent conditioned on $X$, and since $U$ is independent of everything else, $YZ$ are independent conditioned on $XU$ as well. In fact, if we define random variables $DF$ as follows: $D$ is a uniformly random bit, and $F=XU$ or $F=YZ$ depending on whether $D=0$ or $D=1$, then $XUYZ$ are independent conditioned on $DF$, i.e., $\sfP_{XUYZ|DF}$ is a product distribution.

\subsection{The winning probability of $\cl_{\gamma,\alpha}$}\label{subsec:cl-win}

We shall now show that $\omega^*(\cl_{\gamma,\alpha})$  is strictly less than the trivial upper bound of $(1-\gamma) + \gamma\omega^*(\CHSH)$. 
At first sight, it might appear that this claim holds simply by the following argument sketch: if Bob and Charlie could win perfectly in the second round, then Alice's first-round output must be deterministic conditioned on their side-information, which implies that the devices cannot achieve the maximum quantum CHSH winning probability in the first round. However, this idea has two technical flaws (which we implicitly addressed in previous works~\cite{KST22,KT23}, though without detailed elaboration). Firstly, for some non-local games it is known that the classical winning probability can be exceeded while still ensuring some inputs give deterministic outputs (see e.g.~the \emph{partially deterministic polytope} in~\cite{woodheadthesisnodoi}); also, in our scenario Bob and Charlie's guesses could potentially depend on Alice's round-1 input, in which case the above argument sketch does not straightforwardly work (see e.g.~\cite{Tan21} Appendix~B)\footnote{Alternatively, observe that the argument clearly fails if we were to consider only a single player Bob instead of both Bob and Charlie, since even if the devices shared the ideal CHSH state in the first round, when the second round occurs Bob could simply (focusing on the $u=\text{keep}$ and $y\neq\bot$ case, since otherwise the second round is automatically won) use his knowledge of Alice's first-round input to measure in an appropriate basis and learn her output perfectly.}. Secondly, even if that obstacle were overcome, this argument would only imply that any particular strategy for $\cl_{\gamma,\alpha}$ has winning probability less than $(1-\gamma) + \gamma\omega^*(\CHSH)$; it would not rule out the possibility of a sequence of strategies achieving winning probabilities \emph{arbitrarily close} to $(1-\gamma) + \gamma\omega^*(\CHSH)$. For our later security proof, we really do need the property that the value $\omega^*(\cl_{\gamma,\alpha})$ is less than $(1-\gamma) + \gamma\omega^*(\CHSH)$, which is a stronger condition (since it means that all strategies have winning probability \emph{bounded away from} $(1-\gamma) + \gamma\omega^*(\CHSH)$). Hence the above argument sketch does not work by itself; we now present the actual proof.\footnote{We also remark that in order to obtain the desired bound, a self-testing result with ``nonzero robustness'' seems to be necessary, i.e.~if Fact~\ref{fc:CHSH-rig} had only been proven for $\mu=0$, it would not have been enough to obtain our desired result (since it leaves open the possibility of devices with CHSH winning probability arbitrarily close to $\omega^*(\CHSH)$ while still allowing Bob and Charlie to win the second round with arbitrarily high probability).}

\begin{remark}
In the proof below, for ease of description we have used the rigidity bound of~\cite{MYS12} presented as Fact~\ref{fc:CHSH-rig} above, as it has a simple closed-form expression. However, if better bounds on $\omega^*(\cl_{\gamma,\alpha})$ are desired, one can instead use the bounds computed in e.g.~\cite{BNS+15} using semidefinite programming techniques, which are typically tighter and more noise-robust (but do not have closed-form expressions).
\end{remark}

\begin{theorem}\label{thm:cl-win}
For any $\gamma, \alpha \in (0,1)$, there exists a value $\dga > 0$ such that $\omega^*(\cl_{\gamma,\alpha})$ is at most $(1-\gamma) + \gamma\omega^*(\CHSH) - \dga$.
\end{theorem}
\begin{proof}
Take an arbitrary strategy for playing $\cl_{\gamma,\alpha}$. 
Note that the events listed in~\eqref{eq:wincondfull} that give $\sfV_1(X,U,A,S)\cdot\sfV_2(X,U,Y,Z,A,S,B,C)=1$ are mutually exclusive. Therefore, the winning probability of this strategy on $\cl_{\gamma,\alpha}$ is given by summing the probabilities it gives for these events.

We begin by considering the first event. We first note that $\Pr[A\oplus S = X\cdot U | U \neq \mathrm{keep}]$ is just the probability for this strategy to win the CHSH game if it is played in the first round. Therefore, we can denote this value as $\omega^*(\CHSH) - \mu$ for some $\mu \in [0,\omega^*(\CHSH)]$ without loss of generality. Putting this together with $\Pr[U \neq \mathrm{keep}] = \gamma$ from the game definition, we have
\begin{align*}
\Pr\left[(U \neq \mathrm{keep})\land(A\oplus S = X\cdot U)\right] = \gamma(\omega^*(\CHSH) - \mu).
\end{align*}

For the second, third, and fourth events, we apply the game definition to get:
\begin{gather*}
\Pr\left[(U=\mathrm{keep})\land(Y=Z=\bot)\right] = (1-\gamma) \alpha^2, \\
\Pr\left[(U=\mathrm{keep})\land(Y=\bot\neq Z)\right] = (1-\gamma) \alpha(1-\alpha), \\
\Pr\left[(U=\mathrm{keep})\land(Z=\bot\neq Y)\right] = (1-\gamma) \alpha(1-\alpha),
\end{gather*}
so the sum of their probabilities is 
\[
(1-\gamma)(\alpha^2+2\alpha(1-\alpha)) = (1-\gamma)(1-\tilde{\alpha}), \quad \text{ where } \tilde{\alpha}=(1-\alpha)^2.
\]

For the last event, first recall that we have parametrized the probability of this strategy winning the CHSH game as $\omega^*(\CHSH) - \mu$. Hence by the CHSH rigidity property (Fact~\ref{fc:CHSH-rig}), after Alice performs her measurement (chosen uniformly at random), the resulting joint state across her input/output registers and Barlie's register is $O(\mu^{1/4})$-close in trace distance to the ``ideal'' mixture $\bbE_{a,x}\left[\state{x} \otimes \state{a} \otimes \state{a^x}\right]$ (where $\state{a^x}$ again denotes Wiesner states), up to an isometry on Barlie's register.\footnote{While that rigidity property is only stated for projective measurements, in this case we are interested only in the classical-quantum state produced after Alice's measurements, which allows us to focus on a suitably chosen Stinespring dilation to projective measurements without loss of generality.} Putting this together with the monogamy of entanglement property (Fact~\ref{fc:mon-game}), we can conclude that $\Pr\left[B=C=A | (U = \mathrm{keep})\land(Y, Z \neq \bot)\right] \leq \frac{1}{2}\left(1 + \frac{1}{\sqrt{2}}\right) + O(\mu^{1/4})$ (here we implicitly used the fact that the event $(U = \mathrm{keep})\land(Y, Z \neq \bot)$ is independent of Alice's measurement distribution). Combining this with the trivial bound $\Pr\left[B=C=A | (U = \mathrm{keep})\land(Y, Z \neq \bot)\right]$ $\leq 1$, we have overall
\begin{align*}
\Pr\left[(U = \mathrm{keep})\land(Y, Z \neq \bot)\land(B=C=A)\right] &\leq (1-\gamma)(1-\alpha)^2 \min\left\{ 1 , \frac{1}{2}\left(1 + \frac{1}{\sqrt{2}}\right) + O(\mu^{1/4})\right\} \\
&=(1-\gamma)\tilde{\alpha} \min\left\{ 1 , \frac{1}{2}\left(1 + \frac{1}{\sqrt{2}}\right) + O(\mu^{1/4})\right\}.
\end{align*}

Summing the terms, we get the following upper bound on the winning probability:
\begin{align*}
&(1-\gamma) \left(1-\tilde{\alpha} + 
\tilde{\alpha} \min\left\{ 1 , \frac{1}{2}\left(1 + \frac{1}{\sqrt{2}}\right) + O(\mu^{1/4})\right\} \right) + \gamma(\omega^*(\CHSH) - \mu) \\
=&(1-\gamma) \min\left\{ 1 , 1-\tilde{\alpha} + 
\tilde{\alpha}
\left(\frac{1}{2}\left(1 + \frac{1}{\sqrt{2}}\right) + O(\mu^{1/4})\right)
\right\} + \gamma(\omega^*(\CHSH) - \mu) \\
=&(1-\gamma)  \min\left\{1,\beta + O(\mu^{1/4})\right\} + \gamma(\omega^*(\CHSH) - \mu) \quad\text{for } \beta = 1-\tilde{\alpha} + 
\tilde{\alpha}\left(\frac{1}{2}\left(1 + \frac{1}{\sqrt{2}}\right)\right) < 1,
\end{align*}
absorbing a factor of $\tilde{\alpha}$ into the $O(\mu^{1/4})$ term in the last line (recall $\alpha$ is a constant for the purposes of this proof). 

Finally, observe that we have $(1-\gamma) \min\left\{1,\beta + O({\mu}^{1/2})\right\} \leq 1-\gamma$ and $\gamma(\omega^*(\CHSH) - \mu) \leq \gamma\omega^*(\CHSH)$ for any $\mu \in [0,\omega^*(\CHSH)]$; however, there is no $\mu$ in that interval that simultaneously saturates both inequalities (the only value that saturates the second inequality is $\mu=0$, in which case the first inequality is not saturated since $\beta<1$). Therefore we have a strict inequality 
\begin{align*}
(1-\gamma) \min\left\{1,\beta + O(\mu^{1/4})\right\} + \gamma(\omega^*(\CHSH) - \mu) < (1-\gamma) + \gamma\omega^*(\CHSH),
\end{align*}
for all $\mu \in [0,\omega^*(\CHSH)]$. 
To ensure the winning probability is bounded away from $(1-\gamma) + \gamma\omega^*(\CHSH)$ by a constant $\dga>0$, we note that $(1-\gamma) \min\left\{1,\beta + O(\mu^{1/4})\right\} + \gamma(\omega^*(\CHSH) - \mu)$ is a continuous function of  $\mu$ on the closed interval $[0,\omega^*(\CHSH)]$.
Hence it attains its maximum at some value in that interval, and this maximal value will be an upper bound strictly smaller than $(1-\gamma) + \gamma\omega^*(\CHSH)$, as desired.
\end{proof}

With this, we give the following parallel repetition theorem upper bounding the winning probability of $\cl_{\gamma,\alpha}^{t/l}$:

\begin{theorem}\label{thm:cl-parrep}
There exists $\kappa >0$ such that for $
\dga
$ from Theorem \ref{thm:cl-win}, and $t=((1-\gamma)+\gamma\omega^*(\CHSH) - {\dga}/{2})l$,
\[ \omega^*\left(\cl_{\gamma,\alpha}^{t/l}\right) \leq 2^{-\kappa\dga^3\alpha^4l}.\]
\end{theorem}
We prove this theorem (in fact, a more general version of it) in Section \ref{sec:parrep}. 

\section{Uncloneable encryption scheme with variable keys}
\label{sec:DIVKECM}
In this section, we prove our main theorem:
\begin{theorem}\label{thm:main}
There is a scheme for DI-VKECM with message space $\clM=\{0,1\}^\lambda$, such that:
\begin{enumerate}
\item It satisfies the completeness property~\eqref{eq:complete} given i.i.d.~honest devices with a constant level of noise; 
\item It achieves indistinguishable security;
\item There exists some $\nu>0$ such that the scheme achieves $(t(\lambda),0)$-uncloneable security against dishonest devices with $\nu\lambda$ bits of leakage, where $t(\lambda)$ is a function satisfying $t(\lambda)<\lambda$ for all sufficiently large $\lambda$.
\end{enumerate}
\end{theorem}
The protocol achieving Theorem~\ref{thm:main} is described in Scheme~\ref{prot:DI-VKECM}. We prove that Scheme~\ref{prot:DI-VKECM} satisfies completeness in Section~\ref{subsec:corr-proof}, and prove that it satisfies indistinguishability and $(t(\lambda),0)$-uncloneability in Section~\ref{subsec:secur-proof}, which together constitute the proof of Theorem~\ref{thm:main}.

\paragraph{Protocol description.} Our DI-VKECM scheme is described below --- for now, we leave some of the parameters (such as $\gamma$ and $\alpha$) to be arbitrary values, with the exact choices that yield Theorem~\ref{thm:main} to be specified in the explanations after the protocol. We highlight that when a receiver is dishonest, they do not need to be following their parts as described in the scheme, although we still require them to supply some values to the client
at all steps where they are supposed to in the scheme (but these values do not have to be ``honestly'' produced, of course --- recall also that the possibility of a receiver lying about their outputs can be absorbed into the device behaviour without loss of generality). 

\paragraph{Informal description of main parameters and notation.}
\begin{description}
\item $\lambda$: Security parameter (for security definitions from Section~\ref{sec:secdefn})
\item $\gamma$: Probability of a ``test'' instance in the encryption procedure
\item $\alpha$: Probability of an instance being eventually ``erased'' for anchoring purposes
\item $X$, $A$: Input and output strings of client's device in encryption procedure
\item $\eX$, $\eA$: Modified versions of $X$, $A$ in which some instances have been ``erased'' to account for anchoring
\item $U$, $S$: Input and output strings of receiver's device in encryption procedure
\item $\rS$: Output string of receiver's device in decryption procedure
\item $\eS$: Modified version of $\rS$ in which some instances have been ``erased'' to account for anchoring
\item $\otp$: ``One-time-pad'' used to encrypt the message
\item $\gA$: Receiver's guess for $\eA$
\item $F$: Flag for whether the encryption procedure accepted
\item $C$: Classical part of ciphertext
\item $\Kdec$: Decryption key
\end{description}

For this scheme, the client and receiver's devices are to have the functionality that the client's device can take an input string $X \in \{0,1\}^l$ and output a string $A \in \{0,1\}^l$, while the receiver's device can take inputs twice: first it takes an input string $U \in \{0,1, \perp\}^l$ and outputs a string $S \in \{0,1\}^l$, then it later takes an input string $\eX \in \{2,3, \perp\}^l$ and outputs a string $\rS \in \{0,1\}^l$. With this in mind, if we interpret the client and receiver as Alice and Barlie and ignore the second input-output round for the latter, this is exactly the same structure as $l$ parallel instances of the first round of the game $\cl_{\gamma,\alpha}$ described in the preceding section. The reason the second round has a different structure is that we only need to consider the second round of the game $\cl_{\gamma,\alpha}$ when considering the security definition and constructing the security proof, \emph{not} for defining the protocol itself. (The second-round input and output in the protocol itself is designed to allow an honest receiver to obtain a string that is ``highly correlated'' to the client's string; this will become clearer when reading the description of the honest device behaviour immediately below the protocol description.)
Finally, we remark that for this scheme we take both $l$ and the message length to be equal to the security parameter $\lambda$.

\begin{algorithm}[!h]
\caption{DI-VKECM with security parameter $\lambda$ and messages in $\clM = \{0,1\}^\lambda$}
\label{prot:DI-VKECM}
\vspace{0.3cm}
\begin{algorithmic}[1]
\Algphase{$\Enc(1^\lambda, m)$:}
\State Devices of the form described above are distributed between the client and receiver, with $l=\lambda$\alglinelabel{alg:initial}\;
\State The client samples strings $X,U$ as follows: for each $i\in[l]$, set $X_i \in \{0,1\}$ uniformly at random, and independently set $U_i=\mathrm{keep},0,1$ with probabilities $1-\gamma,\gamma/2,\gamma/2$ respectively \;
\State The client inputs $X$ into their device and receives an output string $A\in\{0,1\}^l$ \alglinelabel{alg:Aoutput} \;
\State The client sends $U$ to the receiver \;
\State The receiver inputs $U$ into their device, interpreting $U_i=\mathrm{keep}$ as $\perp$ for each $i$, and receives an output string $S\in \{0,1\}^l$ \;
\State The receiver sends $S$ to the client \alglinelabel{alg:recmessage} \;
\State The client tests if the number of $i\in [l]$ such that $U_i \neq \mathrm{keep}$ and $A_i\oplus S_i \neq X_i\cdot U_i$ is at most $(\gamma (1-\omega^*(\CHSH)) + {\dga }/{2})l$ \alglinelabel{alg:test} \;
\If{the test passes}
\State The client sets the flag to $F=\text{\cmark}$ \;
\State The client samples $\otp\in \clM$ uniformly at random and sends $C = m \oplus \otp$ to the receiver \alglinelabel{alg:otp} \;
\Else
\State The client sets the flag to $F=\text{\xmark}$ \;
\State The client samples $\otp\in \clM$ and $M^\mathrm{fake}\in \clM$ (independently) uniformly at random and sends $C = M^\mathrm{fake} \oplus \otp$ to the receiver \alglinelabel{alg:dummy} \;
\EndIf
\State The client stores $(X,U,A,\otp)$ as the private key; the receiver stores $\rho = \rho'\otimes\state{C}$ as the ciphertext, where $\rho'$ is the quantum state in the receiver's share of the devices\;

\Algphase{$\KeyR(\kpriv)$:}
\State Interpret the private key as $\kpriv = (X,U,A,\otp)$ \;
\State Sample a string $\eX \in \{2, 3, \perp\}^l$ as follows: for each $i\in[l]$, set $\eX_i=\bot$ with probability $\alpha$, and otherwise set $\eX_i=X_i+2$ \alglinelabel{alg:decinputs} \;
\State Set a string $\eA \in \{0,1\}^l$ as follows: for each $i\in[l]$, set $\eA_i = 0$ if $U_i \neq \mathrm{keep}$ or $\eX_i=\bot$, and otherwise set $\eA_i = A_i$ \alglinelabel{alg:erase} \;
\State Compute the syndrome $\syn(\eA)$ following the error-correction procedure described below (see~\eqref{eq:hABhon}) \alglinelabel{alg:syn} \;
\State Release the decryption key $\Kdec = (\otp \oplus \eA, \syn(\eA), \eX)$ \alglinelabel{alg:dkey} \;

\Algphase{$\Dec(\rho, k)$:}
\State Interpret $\rho$ as $\rho'\otimes\state{C}$ where $\rho'$ has $l$ qubit registers and $C \in \clM$; interpret the decryption key $\kdec$ as $(D,\syn(\eA),\eX)$ \;
\State Input $\eX$ into the receiver's device and obtain the output string $\rS \in \{0,1\}^l $ \alglinelabel{alg:recmeasures1} \;
\State Set a string $\eS \in \{0,1\}^l$ as follows: for each $i\in[l]$, set $\eS_i = 0$ if $U_i \neq \mathrm{keep}$ or $\eX_i=\bot$, and otherwise set $\eS_i = \rS_i$ \alglinelabel{alg:recmeasures2} \;
\State Use $\eS,U,\eX$ and $\syn(\eA)$ to compute a guess $\gA$ for $\eA$ \alglinelabel{alg:guess} \;
\State Output $\widetilde{M} = C \oplus D \oplus \gA$ \alglinelabel{alg:decrypt} \;
\vspace{0.3cm}
\end{algorithmic}
\end{algorithm}
\newpage 

\begin{remark}\label{remark:interactive}
Qualitatively, the interaction in our encryption procedure arises solely from the fact that we want the receiver to tell the client the output string $S$, so the client can check whether enough instances win the CHSH game, which later allows us to prove security by invoking the parallel-repetition and rigidity results from Section~\ref{sec:cl-game}. 
We note however that there is an alternative setup that does not require interaction, if we suppose the client is able to locally impose a particular constraint on the devices. Specifically, we could instead consider a procedure in which the client begins by holding \emph{both} the devices and performs steps~\ref{alg:initial}--\ref{alg:recmessage} on them, before sending the ``receiver's device'' over to the receiver (and then proceeding with the rest of the scheme according to the original description), in which case the encryption procedure is non-interactive since it only involves the client sending messages (and devices) to the receiver. For this version, our security proof is still valid as long as the client can enforce that during steps~\ref{alg:initial}--\ref{alg:recmessage}, the devices are still subject to the bounded-leakage constraint (physically, this might be possible by imposing some ``shielding'' measures on the devices, keeping them well isolated from each other). Qualitatively, this version can be viewed as having the client locally ``test'' the devices (while constraining the communication/leakage between them) before using them in the rest of the scheme.
\end{remark}

\paragraph{Honest-device description.} We now specify the honest device behaviour. We allow a small amount of noise in the honest behaviour, under a depolarizing-noise model. Specifically, we suppose that the honest devices share $l$ i.i.d. copies of a \emph{Werner state}
\begin{equation}\label{eq:werner}
\rho_q = (1-2q) \state{\Phi^+} + 2q \frac{\Id}{4},
\end{equation}
where $\ket{\Phi^+} = \frac{1}{\sqrt{2}}(\ket{00} + \ket{11})$, and $q\in[0,1/2]$ will be referred to as the depolarizing-noise parameter.
As for the measurements for each input string, we suppose that (for each instance $i$) the client's device performs Alice's ideal CHSH measurement corresponding to the $i$-th input bit, on the $i$-th copy of the shared state. The receiver's measurements are slightly more complicated, though the same in both the first or second input rounds: for each instance $i$ if it obtains an input value in $\{0,1\}$ then it performs Bob's ideal CHSH measurement corresponding to that value on the $i$-th copy, whereas if it gets an input value in $\{2,3\}$ it performs Alice's ideal CHSH measurement corresponding to that value on the $i$-th copy, and finally if it gets input value $\perp$ then it does nothing to the state and just trivially outputs a fixed value (say, $0$). Note that the honest decryption procedure only requires the use of outputs in the second round from those instances $i$ which had input $\perp$ during encryption, which means that in the honest behaviour, all physically relevant outputs are produced by only measuring each qubit once.\footnote{In step~\ref{alg:recmeasures1} of our protocol description, the decryption procedure as presented does technically provide non-trivial inputs to some locations $i$ where we had $U_i\neq\mathrm{keep}$ in the first round, i.e., the first round input was not $\perp$. However, this is only for simplicity of description --- it can be seen in step~\ref{alg:recmeasures2} that the second-round outputs corresponding to these locations are ignored, 
so it does not matter what the (honest) device does on these locations in the second round, e.g.~even if it is attempting to measure the corresponding qubit a second time.
}

This model for the honest devices has the following important properties (for each instance). First, it achieves a CHSH winning probability of $(1-2q) \omega^*(\CHSH) + q$. Second, when the client and receiver's input pairs are $(0,2)$ or $(1,3)$, the outputs are equal with probability $1-q$ (and the marginal output distributions are uniform). In particular, this means that in the decryption step (where the receiver essentially inputs $\eX_i=X_i+2$ into their device for many instances), the honest devices produce outputs that are reasonably ``well correlated'' between the client and receiver, up to the depolarizing-noise parameter $q$.

While this noise model is quite simple, we remark that given i.i.d.~honest devices with an arbitrary single-instance input-output distribution, it is possible in principle to perform a depolarization procedure (see~\cite{MAG06}) to transform their input-output distribution into a form matching the above description (except on the input pairs $(0,3)$ and $(1,2)$, which do not occur in the protocol), although this is not necessarily the optimal way to use the devices~\cite{MPW21}. Alternatively, we note that our subsequent analysis can in fact be phrased entirely in terms of three parameters of the (i.i.d.)~honest devices: the probability that they win the CHSH game given uniformly random inputs in $\{0,1\}$, the entropy of Alice's output conditioned on Bob's output for input pair $(0,2)$, and the same for input pair $(1,3)$. Hence if necessary, one could instead specify these three parameters independently, but for ease of presentation in this work we simply use the single-parameter depolarizing-noise model.

\paragraph{Error correction.} Because we allow some noise in the honest devices, the client and receiver's output strings from the devices will not be perfectly correlated even in the honest case. We shall accommodate this in our protocol by incorporating standard error correction procedures (see e.g.~\cite{Ari10,Ren05,TMP+17}), which we briefly summarize as follows:
\begin{fact}\label{fc:EC}
Suppose Alice and Bob hold classical random variables $CC'$ following an i.i.d. distribution $\sfP_{C_iC'_i}^{\otimes l}$.
We shall refer to an error correction procedure as a process in which Alice computes a syndrome value $\syn(C)$ of length $\lsyn$ (in bits) to send to Bob, who uses it together with $C'$ to produce a guess for $C$.
Then for any $\xi>1$ and $\beta<1/2$, there exists an efficient error-correction procedure such that the syndrome length is
\begin{align}
\label{eq:EC}
\lsyn = \xi \sfH(C_i|C'_i) l,
\end{align}
and the probability of Bob's guess being wrong is upper bounded by a function
$\lEC(l) = O(2^{-l^\beta})$.
\end{fact}

With this in mind, we shall lay out the error correction procedure used in our DI-VKECM scheme. Here we shall describe the error correction procedure in terms of arbitrary values for the protocol parameters $\gamma,\alpha,\qhon$.
The protocol steps corresponding to the error correction procedure are step~\ref{alg:syn}, in which a syndrome is computed for a string $\eA$, and step~\ref{alg:guess}, where the receiver uses it to produce a guess for $\eA$. We note that in the honest case, when the full process $\Dec\circ\KeyR\circ\Enc$ in Scheme \ref{prot:DI-VKECM} is carried out, at the point of step~\ref{alg:syn} the receiver holds an output string $\eS$ from the device along with the values $U,\eX$. Furthermore, the device behaviour is i.i.d.~in the honest case.
With this, we take the error-correction procedure to be performed as specified in Fact~\ref{fc:EC}, interpreting the distribution on $(C_i,C'_i)$ in that procedure to be the distribution on $(\eA_i, \eS_i U_i \eX_i)$ for the \emph{honest} devices. 
Note that this means the value $\sfH(C_i|C'_i)$ in~\eqref{eq:EC} in this context has the following value (letting $h_2$ denote the binary entropy function):
\begin{align} 
\sfH(\eA_i|\eS_i U_i \eX_i) 
&= \sum_{u_i \tilde{x}_i} \Pr[U_i \eX_i = u_i \tilde{x}_i] \sfH\!\left(\eA_i|\eS_i; U_i \eX_i = u_i \tilde{x}_i \right) 
= (1-\gamma)(1-\alpha) h_2(\qhon), \label{eq:hABhon}
\end{align}
where we have used the facts that
$\sfH\!\left(\eA_i|\eS_i; U_i \eX_i = u_i \tilde{x}_i\right) = 0$ whenever 
$u_i \neq \mathrm{keep}$ or $\tilde{x}_i = \bot$ (because $\eA_i$ is set to a deterministic value in those cases), and $\sfH\!\left(\eA_i|\eS_i; U_i \eX_i = u_i \tilde{x}_i\right)  = h_2(\qhon)$ otherwise (because in that case we have $u_i = \mathrm{keep}$ and $\tilde{x}_i = x_i+2$, i.e.~for the honest devices this corresponds to the client and receiver measuring in the same basis on a Werner state~\eqref{eq:werner}, hence their output values $(\eA_i,\eS_i)$ have uniform marginal distributions, and the probability that they differ is $\qhon$).

\paragraph{Parameter choices.} We now specify the parameter choices for our DI-VKECM scheme. 
The input/output length $l$ for the devices is set equal to the security parameter $\lambda$ as mentioned before, and the message is also required to be a bit string of length $\lambda$, i.e.~we have $\clM=\{0,1\}^\lambda$. (This is a slightly redundant parametrization in this case since $\lambda,l,\log|\clM|$ are all equal to each other, but we present it this way to maintain flexibility for potentially choosing $l$ and $\clM$ differently in terms of $\lambda$, as shall return to in Section~\ref{sec:bittritPA}.)
Take any choice of values for $\xi>1$ and $\gamma, \alpha \in (0,1)$, and set the parameter $\delta_{\gamma, \alpha}$ appearing in the protocol to be the corresponding value from Theorem~\ref{thm:cl-win}. Letting $\kappa$ be the constant from Theorem~\ref{thm:cl-parrep}, we require the honest devices to have depolarizing-noise values $\qhon$ satisfying
\begin{equation}\label{eq:par-cons-1}
\qhon < \frac{\dga}{4}, \qquad 2\xi(1-\gamma)(1-\alpha)h_2(\qhon) < \kappa\dga^3\alpha^4 ,
\end{equation}
which is always possible by taking a sufficiently small value of $\qhon$ since $\gamma,\alpha,\dga,\kappa$ have been fixed. 
With this, all parameters necessary to specify the error correction procedure (in steps~\ref{alg:syn} and~\ref{alg:guess}) have been fixed; in particular, the length of $\syn(\eA)$ in that procedure is given by
\begin{equation}\label{eq:par-cons-2}
\lsyn = \xi (1-\gamma)(1-\alpha)h_2(\qhon) l.
\end{equation}
Note that since the error-correction procedure is based on the honest behaviour, $\lsyn$ is a protocol parameter, not a random variable.

\subsection{Completeness of Scheme~\ref{prot:DI-VKECM}}\label{subsec:corr-proof}
It is easy to see that Scheme~\ref{prot:DI-VKECM} satisfies the property~\eqref{eq:prodstate} that when $F=\text{\xmark}$, the ciphertext state is independent of the message. We now show that it also satisfies the completeness property~\eqref{eq:complete}.
\begin{theorem}
With the parameter choices as specified in \eqref{eq:par-cons-1}-\eqref{eq:par-cons-2}, Scheme~\ref{prot:DI-VKECM} satisfies the completeness property~\eqref{eq:complete}, with the bound on the right-hand-side being $1-\left(e^{-\dga^2 \lambda/8} + \lEC(\lambda)\right)$ where $\lEC$ is the function described in Fact~\ref{fc:EC}.
\end{theorem}

\begin{proof}
Since 
\begin{align*}
&\Pr\left[\left(F = \text{\cmark}\right)\land\left(\Dec\circ\KeyR\circ \Enc(1^\lambda, m) = m\right)\right] \\
=& 1-\Pr\left[\left(F = \text{\xmark}\right)\lor\left(\Dec\circ\KeyR\circ\Enc(1^\lambda, m) \neq m\right)\right] \\
=& 1 - \left(\Pr\left[F = \text{\xmark}\right] + \Pr\left[\left(F = \text{\cmark}\right)\land\left(\Dec\circ\KeyR\circ\Enc(1^\lambda, m) \neq m\right)\right]\right),
\end{align*}
to prove completeness it suffices to bound the two probabilities in the last line when the devices are honest.

To bound $\Pr\left[F = \text{\xmark}\right]$, we recall that our model of the honest devices gives a CHSH winning probability of $(1-2\qhon) \omega^*(\CHSH) + \qhon > \omega^*(\CHSH) - \qhon$. 
Therefore, by the distribution of $U_i$ in the protocol, in each round we have
\begin{align*}
\Pr[(U_i \neq \mathrm{keep}) \land (A_i\oplus S_i \neq X_i\cdot U_i)] 
&= \gamma \Pr[A_i\oplus S_i \neq X_i\cdot U_i | U_i \neq \mathrm{keep}] \\
&\leq \gamma \left(1 - \omega^*(\CHSH) + \qhon\right) \\
&\leq \gamma \left(1 - \omega^*(\CHSH) \right) + \frac{\dga}{4} \quad \text{since $\qhon<\frac{\dga}{4}$ and $\gamma<1$}.
\end{align*}
Since
this value is ${\dga}/{4}$ less than the threshold fraction $\gamma (1-\omega^*(\CHSH)) + {\dga }/{2}$ in the 
step~\ref{alg:test} 
test, and all the instances are i.i.d. in the honest case, we can apply the Chernoff bound to get that the probability of failing the test (i.e.~$F = \text{\xmark}$) is at most $e^{
-\dga^2 l/8
}$.

As for $\Pr\left[\Dec\circ \Enc(1^\lambda, m) \neq m\right]$,
first observe that when $F = \text{\cmark}$, the honest receiver's decrypted value is equal to $m$ whenever $\gA = \eA$, and hence 
\[
\Pr\left[\left(F = \text{\cmark} \right) \land \left(\Dec\circ\KeyR\circ\Enc(1^\lambda, m) \neq m\right) \right] 
\leq 
\Pr\left[\left(F = \text{\cmark} \right) \land \left(\gA \neq \eA\right) \right] 
\leq \Pr\left[ \gA \neq \eA \right].
\]
Since the honest behaviour is i.i.d., and the error-correction step was taken to be as described in Fact~\ref{fc:EC} based on the honest behaviour, we have that the probability of the receiver's guess $\gA$ being wrong (i.e.~$\gA \neq \eA$) is at most $\lEC(l)$ in the honest case. This gives the claimed result, recalling that $l=\lambda$.
\end{proof}

\subsection{Security of Scheme~\ref{prot:DI-VKECM}}\label{subsec:secur-proof}

We begin with a simple observation about the behaviour of one-time-pads (which can be verified with a straightforward calculation):
\begin{fact}\label{fc:otp}
Consider an arbitrary state $\rho_{E C_1}$ where $C_1$ is a classical $n$-bit register. Suppose we generate an independent uniformly random $n$-bit string $C_2$, and set another register $C_3 = C_1 \oplus C_2$. Then the resulting state $\rho_{E C_1 C_2 C_3}$ is exactly the same as though we had instead generated $C_3$ as an independent uniformly random $n$-bit string, and set $C_2 = C_1 \oplus C_3$.

In particular, this implies that if we trace out $C_2$ from that state $\rho_{E C_1 C_2 C_3}$, the resulting reduced state has the product form $\rho_{E C_1 C_3} = \rho_{E C_1} \otimes \Id_{C_3}/2^{n}$, i.e.~it could equivalently be generated by setting $C_3$ to be a uniformly random value independent of $\rho_{E C_1}$.
\end{fact}

From this, we immediately obtain indistinguishable security for our protocol:
\begin{theorem}\label{thm:indsec}
Scheme~\ref{prot:DI-VKECM} is indistinguishable-secure. 
\end{theorem}
\begin{proof}
By definition, indistinguishable security only involves the output of the encryption procedure excluding the private key $\Kpriv$.
We observe that in the encryption procedure, the only step that depends on the message $M$ is step~\ref{alg:otp} (when $F=\text{\cmark}$), where basically a one-time-pad $\otp$ is independently generated uniformly at random and applied to the message. 
Furthermore, since we exclude the private key in analyzing indistinguishable security, we can trace out $\otp$ immediately after that step. Hence we can apply Fact~\ref{fc:otp} (identifying $M,\otp$ with $C_1,C_2$ respectively) to conclude that the ciphertext register $C$ is completely independent of the message.
Since this is the only register that could potentially depend on the message here, it clearly follows that indistinguishable security holds.
\end{proof}

Next, we turn to proving uncloneable security. We begin by first 
considering a modified version of a cloning attack in which some registers are not provided to the dishonest parties, and
proving a bound on the probability that Bob and Charlie can simultaneously guess some parts of the ``internal'' values $\eA$ produced in step~\ref{alg:erase} of the $\KeyR$ procedure:

\newcommand{\eAB}{\widetilde{A}^\B}
\newcommand{\eAC}{\widetilde{A}^\C}
\newcommand{\KdecB}{K^\B_\mathrm{dec}}
\newcommand{\KdecC}{K^\C_\mathrm{dec}}
\newcommand{\gB}{G^\B}
\newcommand{\gC}{G^\C}
\newcommand{\Esucc}{\clE_\mathrm{succ}}

\begin{lemma}\label{lem:game2prot}
We introduce the following notation for some registers produced in a cloning attack (as defined in Definition~\ref{def:cl-sec}): when $\KeyR$ is performed for Bob (resp.~Charlie), let $Y$ (resp.~$Z$) denote the value of $\eX$ produced in step~\ref{alg:decinputs}, and let $\eAB$ (resp.~$\eAC$) denote the value of $\eA$ produced in step~\ref{alg:erase}. Let $\clT$ denote the subset of instances $i\in[l]$ such that $Y_i$ and $Z_i$ are \emph{both} equal to $X_i$.

Now consider a cloning attack, except with the following modifications: 
\begin{itemize}
\item When the encryption procedure $\Enc$ is performed, the receiver is not given the register $C$ in step~\ref{alg:otp} (or step~\ref{alg:dummy}).
\item When the key-release procedure $\KeyR$ is performed for Bob (resp.~Charlie), the values $\otp \oplus \eAB$ and $\syn(\eAB)$ (resp.~$\otp \oplus \eAC$ and $\syn(\eAC)$) are omitted from the decryption key in step~\ref{alg:dkey}.
\item Instead of Bob (resp.~Charlie) producing a guess for the message $M$, he produces a guess
$\gB$ (resp.~$\gC$) for $\eAB$ (resp.~$\eAC$).
\end{itemize}
Then if there is no leakage between the client and receiver's devices during $\Enc(1^\lambda, M)$, or between Bob and Charlie's devices after they receive their decryption keys, the following bound holds:
\[
\Pr[(F=\text{\cmark}) \land (\eAB_\clT = \gB_\clT) \land (\eAC_\clT = \gC_\clT)] \leq 2^{-\kappa\dga^3\alpha^4l},
\]
where $\kappa$ is the constant from Theorem~\ref{thm:cl-parrep}, and the notation $W_\clT$ for any string $W$ denotes the substring of $W$ on the instances in $\clT$.
\end{lemma}

\begin{proof}
We first briefly summarize the main processes involving the dishonest parties in this modified cloning-attack scenario. The receiver begins with some share of a quantum state in their device, then gets the classical value $U$ from the client and uses it to produce a value $S$ to send to the receiver (in step~\ref{alg:recmessage}). 
Without any further communication from the client (since the register $C$ has been omitted), the receiver
distributes states between Bob and Charlie.
Finally Bob and Charlie receive values $Y$ and $Z$ respectively (since the values $\otp \oplus \eAB,\syn(\eAB),\otp \oplus \eAC,\syn(\eAC)$ have been omitted from their decryption keys), and measure their states to produce their guesses $\gB,\gC$. 
Our proof proceeds by comparing this situation to the game $\cl_{\gamma,\alpha}^{t/l}$ (with $t$ arbitrary for now; we shall fix a suitable value at the end).

We observe that in terms of the input distribution (and the order in which they are supplied relative to the state distribution step), the distribution of $X,U,Y,Z$ in this scenario is exactly the same as in $\cl_{\gamma,\alpha}^{t/l}$. Furthermore, the client's output string $A$ produced during the protocol is produced the same way as Alice's output in $\cl_{\gamma,\alpha}^{t/l}$, the receiver's output string $S$ is produced the same way as Barlie's first-round output in $\cl_{\gamma,\alpha}^{t/l}$, and Bob and Charlie's guesses $\gB,\gC$ are produced the same way as their second-round outputs in $\cl_{\gamma,\alpha}^{t/l}$. Hence we shall treat these values as their outputs in the game $\cl_{\gamma,\alpha}^{t/l}$.

Our goal is to bound the probability of the event $(F=\text{\cmark}) \land (\eAB_\clT = \gB_\clT) \land (\eAC_\clT = \gC_\clT)$. 
We shall first argue that this event implies that the game $\cl_{\gamma,\alpha}^{t/l}$ is won on all instances $i\in[l]$ where $U_i = \mathrm{keep}$. 
To do so, consider any such instance $i$, and observe that if $i\notin\clT$ (i.e.~at least one of $Y_i,Z_i$ is $\bot$), then $\cl_{\gamma,\alpha}^{t/l}$ is automatically won on that instance, recalling the win conditions~\eqref{eq:wincondfull}. 
On the other hand, if $i\in\clT$, then the event $(F=\text{\cmark}) \land (\eAB_\clT = \gB_\clT) \land (\eAC_\clT = \gC_\clT)$ implies that we have 
$(\eAB_i = \gB_i) \land (\eAC_i = \gC_i)$. 
Furthermore, since $i\in\clT$ we have that $\eAB_i$ and $\eAC_i$ are equal to $A_i$ (recalling how they were constructed in step~\ref{alg:erase} of the protocol). So we see that $\eA_i = \gB_i = \gC_i$, i.e.~the win condition of $\cl_{\gamma,\alpha}^{t/l}$ is indeed fulfilled on all such instances as well. 

Hence the only instances on which the win condition of $\cl_{\gamma,\alpha}^{t/l}$ might not be fulfilled are those in which $U_i \neq \mathrm{keep}$. Again referring back to the win conditions~\eqref{eq:wincondfull}, we see that such instances fail the win condition if and only if $A_i\oplus S_i \neq X_i\cdot U_i$. But recalling the definition of the ``testing'' step~\ref{alg:test} in the protocol, we see that $F=\text{\cmark}$ is the statement that the number of such instances is at most $(\gamma (1-\omega^*(\CHSH)) + {\dga }/{2})l$. Thus overall, we conclude that the event $(F=\text{\cmark}) \land (\eAB_\clT = \gB_\clT) \land (\eAC_\clT = \gC_\clT)$ implies the number of instances satisfying the win condition is at least
\[
l - \left(\gamma (1-\omega^*(\CHSH)) + \frac{\dga }{2}\right)l = \left((1 - \gamma) + \gamma\omega^*(\CHSH) - \frac{\dga }{2}\right) l.
\]

Therefore, the probability of $(F=\text{\cmark}) \land (\eAB_\clT = \gB_\clT) \land (\eAC_\clT = \gC_\clT)$ is at most the probability of having at least $t = \left((1 - \gamma) + \gamma\omega^*(\CHSH) - {\dga }/{2}\right) l $ win instances in $\cl_{\gamma,\alpha}^{t/l}$.
By Theorem~\ref{thm:cl-parrep},
the latter probability is upper bounded by $2^{-\kappa\dga^3\alpha^4l}$,  
which gives the desired result.
\end{proof}

With this, we now prove a bound on the probability Bob and Charlie can both guess the values $\eAB,\eAC$ described above, but using all the registers involved in an actual cloning attack. Informally, the idea is that to compensate for the syndromes $\syn(\eAB),\syn(\eAC)$ we simply multiply the guessing probability by a factor corresponding to the lengths of the syndromes, whereas for the registers $C, \otp \oplus \eAB, \otp \oplus \eAC$ we shall argue that they are mostly ``decoupled'' from $\eAB\eAC$ and hence have limited effects on the optimal guessing probability.

\begin{lemma}\label{lem:rawkeyguess}
We follow the notation as defined in Lemma~\ref{lem:game2prot}.
Consider a cloning attack with the modification that instead of Bob (resp.~Charlie) producing a guess for the message $M$, he produces a guess $\gB$ (resp.~$\gC$) for $\eAB$ (resp.~$\eAC$).
Then if there is no leakage between the client and receiver's devices during $\Enc(1^\lambda, M)$, or between Bob and Charlie's devices after they receive their decryption keys, the following bound holds:
\begin{equation}\label{eq:rawkeyguess}
\Pr[(F=\text{\cmark}) \land (\eAB = \gB) \land (\eAC = \gC)] \leq 2^{-\kappa\dga^3\alpha^4l+2\lsyn},
\end{equation}
where $\kappa$ is the constant from Theorem~\ref{thm:cl-parrep}, and $\lsyn$ given by \eqref{eq:par-cons-2}.
\end{lemma}

\begin{proof}
The difference between this scenario as compared to that described in Lemma~\ref{lem:game2prot} is that here, the dishonest parties have access to more registers: specifically, the receiver is given the register $C$ in the ciphertext, and Bob (resp.~Charlie) is given $\otp \oplus \eAB,\syn(\eAB)$ (resp.~$\otp \oplus \eAC,\syn(\eAC)$) in his decryption key. For brevity, let $D^\B$ and $H^\B$ denote the values $\otp \oplus \eAB$ and $\syn(\eAB)$ for Bob, and define $D^\C$ and $H^\C$ analogously for Charlie.
To prove the claimed bound~\eqref{eq:rawkeyguess} in this scenario, we shall analyze a sequence of modified scenarios in which we progressively remove these registers from consideration, finally arriving at the scenario  in Lemma~\ref{lem:game2prot}. (Note that the probability we are bounding in this lemma is also slightly different from that in Lemma~\ref{lem:game2prot}, because it involves the full strings $\eAB, \gB, \eAC, \gC$ rather than just the substrings on $\clT$. We return to this point later in this proof.)

To facilitate this analysis, we introduce the following notation. For any subset $Q$ of the registers mentioned above (i.e.~$Q \subseteq \{C,D^\B,D^\C,H^\B,H^\C\}$), let $\clC_Q$ denote the set of all ``modified cloning attacks'' in the sense described in the Lemma~\ref{lem:game2prot} statement 
(i.e.~at the end Bob and Charlie produce guesses for $\eAB$ and $\eAC$),
except that the dishonest parties additionally have access 
to the registers $Q$ (with the implicit understanding that these registers become available to the dishonest parties at the same points as in a standard cloning attack). In terms of this notation, Lemma~\ref{lem:game2prot} is considering all attacks in the set $\clC_\emptyset$ ($\emptyset$ denotes the empty set), while this lemma statement we aim to prove here is simply the statement that
\begin{equation} \label{eq:guessopt}
\max_{\clC_{C D^\B D^\C H^\B H^\C}} \Pr[\Esucc] \leq 2^{-\kappa\dga^3\alpha^4l+2\lsyn}, \quad \text{where } \Esucc = (F=\text{\cmark}) \land (\eAB = \gB) \land (\eAC = \gC).
\end{equation}

As the first step, we argue that Bob and Charlie having access to the registers $H^\B H^\C$ could only have increased the attack's maximum success probability by at most a factor of $2^{2\lsyn}$, i.e.~in terms of the above notation we have
\[
\max_{\clC_{C D^\B D^\C H^\B H^\C}} \Pr[\Esucc] \leq 2^{2\lsyn} \max_{\clC_{C D^\B D^\C}} \Pr[\Esucc] .
\]
This follows from a standard guessing-strategy argument: given any attack in which Bob and Charlie use $H^\B H^\C$ to achieve $\Pr[\Esucc] = p$ for some $p\in[0,1]$, there is always another attack in which Bob and Charlie achieve $\Pr[\Esucc] \geq 2^{-2\lsyn}p$ \emph{without} access to $H^\B H^\C$, as follows: Bob and Charlie independently produce uniformly random guesses for $H^\B$ and $H^\C$ respectively, then proceed with the original strategy as though these guesses were the true values of $H^\B H^\C$.
Note that the probability of these guesses being (both) equal to the true syndrome values is always exactly $2^{-2\lsyn}$, hence the described attack succeeds with probability at least $2^{-2\lsyn}p$, as claimed.

With this bound we can remove the registers $H^\B H^\C$ from consideration, in the sense that to prove~\eqref{eq:guessopt}, it would suffice to show that
\begin{equation}\label{eq:pguessCDD}
\max_{\clC_{C D^\B D^\C}} \Pr[\Esucc] \leq 2^{-\kappa\dga^3\alpha^4l}.
\end{equation}
To do so, we now proceed to remove the register $C$ as well, by arguing that 
\[
\max_{\clC_{C D^\B D^\C}} \Pr[\Esucc] = \max_{\clC_{D^\B D^\C}} \Pr[\Esucc] .
\]
This follows by observing that in the definition of a cloning attack (or modified cloning attack, in this context), the message $M$ is chosen uniformly at random and independently of everything else. This means we can effectively view it as playing the role of a one-time-pad in the process of generating $C$ --- note that this is ``reversed'' from our earlier Theorem~\ref{thm:indsec} proof, in that here we shall view $M$ rather than $\otp$ as the one-time-pad. To put this rigorously: before step~\ref{alg:otp}, in this context $M$ is a uniformly random value that was generated independently of everything else. Furthermore, for the event $\Esucc$ that we are considering, the only role played by the message $M$ is in generating the value $C=M\oplus \otp$ in step~\ref{alg:otp} (focusing on the $F=\text{\cmark}$ case, since otherwise $C$ is produced from a ``dummy message'' $M^\mathrm{fake}$ independent of the true message), and hence we can just trace out $M$ immediately after that step.
With this, we apply the last statement in Fact~\ref{fc:otp} (identifying $\otp,M$ with $C_1,C_2$ respectively --- again, we stress that the roles are ``reversed'' from our Theorem~\ref{thm:indsec} proof) to conclude that the state produced after step~\ref{alg:otp} is \emph{exactly} the same as though $C$ was generated as a uniformly random value independent of everything else. Since a dishonest receiver could generate such a $C$ on their own anyway, we can absorb it into the actions of a dishonest receiver when considering the set of attacks $\clC_{D^\B D^\C}$, and conclude that the maximum value of $\Pr[\Esucc]$ over attacks in the set $\clC_{C D^\B D^\C}$ is the same as over the set $\clC_{D^\B D^\C}$.

The last step is to remove the registers $D^\B D^\C$ and connect $\Esucc$ to the event considered in Lemma~\ref{lem:game2prot}. 
We start by observing that whenever Bob and Charlie have access to $D^\B = \otp \oplus \eAB$ and $D^\C = \otp \oplus \eAC$, they can simultaneously guess (respectively) $\eAB$ and $\eAC$ correctly \emph{if and only if} they can simultaneously guess $\otp$ --- this follows by observing that if e.g.~Bob guesses $\otp$ correctly, he can get $\eAB$ by computing $\otp \oplus D^\B$; conversely if he guesses $\eAB$ correctly he can get $\otp$ from $\eAB \oplus D^\B$ (the analysis for Charlie is analogous). 
To discuss this more easily, let $\widehat{\clC}_Q$ denote the set of all modified cloning attacks in a similar fashion to $\clC_Q$, except with one further modification that rather than having Bob (resp.~Charlie) produce a guess
$\gB$ (resp.~$\gC$) for $\eAB$ (resp.~$\eAC$), he is to produce a guess
$\widehat{G}^\B$ (resp.~$\widehat{G}^\C$) for $\otp$.\footnote{Technically, although we have described $\gB$ as a guess for $\eAB$ and $\widehat{G}^\B$ as a guess for $\otp$ (and analogously for Charlie), this is not strictly necessary from a purely mathematical standpoint --- abstractly, 
Bob and Charlie are just producing some values in $\{0,1\}^l$ regardless of whether we are considering $\clC_Q$ or $\widehat{\clC}_Q$; the only difference is how we choose to label and interpret these values. However, using the notation we have presented makes the argument easier to describe.} 
(Note that this means Bob and Charlie are both trying to guess the same value in this case.) Then the preceding observation tells us that we have 
\[
\max_{\clC_{D^\B D^\C}} \Pr[\Esucc] = \max_{\widehat{\clC}_{D^\B D^\C}} \Pr[(F=\text{\cmark}) \land (\otp = \widehat{G}^\B = \widehat{G}^\C)] .
\]

On the right-hand-side of the above equation, Bob and Charlie are trying to guess $\otp$ given access to $D^\B D^\C$. Now consider the set $\clT$ as defined in the Lemma~\ref{lem:game2prot} statement, and let $\overline{\clT} = [l]\setminus\clT$. Note that because $\otp$ is a uniformly random $l$-bit string, all the bits within it are independent of each other as well, and we can discuss them individually. Suppose that for all the instances $i\in\overline{\clT}$, rather than giving Bob (resp.~Charlie) the bit $D_i^\B$ (resp.~$D_i^\C$) we were to simply give them $\otp_i$ directly. This would improve their probability of guessing $\otp$, hence if we let $\widehat{\clC}_{D^\B_\clT \otp_{\overline{\clT}}, D^\C_\clT \otp_{\overline{\clT}}}$ denote a scenario where Bob and Charlie get $D^\B_\clT \otp_{\overline{\clT}}$ and $D^\C_\clT \otp_{\overline{\clT}}$ respectively (following the Lemma~\ref{lem:game2prot} notation) rather than simply $D^\B$ and $D^\C$, we can write
\begin{align}
\max_{\widehat{\clC}_{D^\B D^\C}} \Pr[(F=\text{\cmark}) \land (\otp = \widehat{G}^\B = \widehat{G}^\C)] 
&\leq \max_{\widehat{\clC}_{D^\B_\clT \otp_{\overline{\clT}}, D^\C_\clT \otp_{\overline{\clT}}}} \Pr[(F=\text{\cmark}) \land (\otp = \widehat{G}^\B = \widehat{G}^\C)] \label{eq:pguessDRDR} \\
&= \max_{\widehat{\clC}_{D^\B_\clT D^\C_\clT}} \Pr[(F=\text{\cmark}) \land (\otp_\clT = \widehat{G}^\B_\clT = \widehat{G}^\C_\clT)] \nonumber\\
&= \max_{\clC_{D^\B_\clT D^\C_\clT}} \Pr[(F=\text{\cmark}) \land (\eAB_\clT = \gB_\clT) \land (\eAC_\clT = \gC_\clT)] ,\nonumber
\end{align}
where the second line holds because when e.g.~Bob has access to $\otp_{\overline{\clT}}$ he can guess $\otp$ if and only if he can guess $\otp_\clT$ (similarly for Charlie), and the third line holds because once again, 
when e.g.~Bob has access to $D^\B_\clT=\otp_\clT \oplus \eAB_\clT$ he can guess $\otp_\clT$ if and only if he can guess $\eAB_\clT$ (similarly for Charlie). 

Finally, we observe that for attacks in $\clC_{D^\B_\clT D^\C_\clT}$, the dishonest parties do not have access to the register $C$, and thus the register $\otp$ is not involved in anything until step~\ref{alg:dkey}. We can hence say that $\otp$ is generated at that point, uniformly at random and independently of everything else, then used to compute $D^\B_\clT=D^\C_\clT=\otp_\clT \oplus \eA_\clT$ (recalling steps~\ref{alg:erase} and \ref{alg:dkey} of the protocol) and traced out immediately afterwards (since it is not involved in any subsequent steps when considering $\Pr[(F=\text{\cmark}) \land (\eAB_\clT = \gB_\clT) \land (\eAC_\clT = \gC_\clT)]$). Therefore we can once again apply the last statement in Fact~\ref{fc:otp} (this time identifying $\otp_\clT,\eA_\clT$ with $C_1,C_2$ respectively, which is more similar to our Theorem~\ref{thm:indsec} proof) to conclude that the state produced after step~\ref{alg:dkey} is {exactly} the same as though $D^\B_\clT$ was generated as a uniformly random value independent of everything else, and $D^\C_\clT$ set equal to it. Since this is something that a dishonest receiver could have done on their own (before distributing $D^\B_\clT$ and $D^\C_\clT$ to Bob and Charlie), we conclude that 
\[
\max_{\clC_{D^\B_\clT D^\C_\clT}} \Pr[(F=\text{\cmark}) \land (\eAB_\clT = \gB_\clT) \land (\eAC_\clT = \gC_\clT)] = \max_{\clC_{\emptyset}} \Pr[(F=\text{\cmark}) \land (\eAB_\clT = \gB_\clT) \land (\eAC_\clT = \gC_\clT)] .
\]
Lemma~\ref{lem:game2prot} is precisely the statement that the right-hand-side of the above expression\footnote{Notice that the original event of interest $\Esucc$ is a ``stricter'' condition than the event we have finally ended up considering, namely $(F=\text{\cmark}) \land (\eAB_\clT = \gB_\clT) \land (\eAC_\clT = \gC_\clT)$ in Lemma~\ref{lem:game2prot} (which only requires Bob and Charlie to guess correctly on $\clT$, not all the instances). This distinction is basically reflected in the inequality~\eqref{eq:pguessDRDR} in our proof here, where informally speaking we have allowed Bob and Charlie to ``win for free'' on $\overline{\clT}$ by simply giving them all the values $\otp_{\overline{\clT}}$.} is at most $2^{-\kappa\dga^3\alpha^4l}$, hence proving the desired bound~\eqref{eq:pguessCDD}.
\end{proof}

\begin{remark}
In principle, an alternative approach to the above arguments is possible, by instead defining the set $\clT$ to also include all the instances with $Y_i=Z_i=\bot$. This does not significantly change the overall structure --- Lemma~\ref{lem:game2prot} still holds with this definition of $\clT$ (because making $\clT$ a bigger set just makes $(F=\text{\cmark}) \land (\eAB_\clT = \gB_\clT) \land (\eAC_\clT = \gC_\clT)$ a stricter condition, i.e.~the probability of that event can only decrease), and in the proof of Lemma~\ref{lem:rawkeyguess}, this definition of $\clT$ still has the property that $D^\B_\clT=D^\C_\clT=\otp_\clT \oplus \eA_\clT$ (in fact this is precisely the ``largest'' choice of $\clT$ on which we can be sure this property holds). Still, we have chosen our definition of $\clT$ as presented because the Lemma~\ref{lem:game2prot} proof is slightly easier to describe with that choice.
\end{remark}

Finally, with the above lemma we can straightforwardly bound the probability of Bob and Charlie simultaneously guessing the message, hence obtaining uncloneable security. We first present a version without leakage between the devices:
\begin{theorem} \label{thm:unc_noleak}
With the parameter choices as specified in \eqref{eq:par-cons-1}-\eqref{eq:par-cons-2}, Scheme~\ref{prot:DI-VKECM} is $(t(\lambda),0)$-uncloneable-secure with 
\[
t(\lambda) = (1 - \kappa\dga^3\alpha^4 + 2\xi (1-\gamma)(1-\alpha)h_2(\qhon))\lambda,
\]
if there is no leakage between the client and the receiver's devices during $\Enc(1^\lambda, M)$, and between the two parties Bob and Charlie after the ciphertext is distributed between them in the cloning attack.
\end{theorem}

\begin{proof}
Uncloneable security is defined in terms of the probability $\Pr[(F=\text{\cmark}) \land (M=M^\B=M^\C)]$ only, hence we can focus only on the case where $F=\text{\cmark}$ in the protocol. We note that when this happens, the receiver gets the classical value $C=M\oplus\otp$, and without loss of generality we suppose they distribute copies of it to Bob and Charlie in the cloning attack. Furthermore, Bob and Charlie get the classical values $D^\B = \otp \oplus \eAB$ and $D^\C = \otp \oplus \eAC$ respectively, where $\eAB,\eAC$ are as described in the statement of Lemma~\ref{lem:rawkeyguess}. From this, we see that Bob and Charlie can simultaneously guess $M$ correctly if and only if Bob can guess $\eAB$ correctly and Charlie can guess $\eAC$ correctly (because e.g.~if Bob guesses $M$ correctly, he can get $\eAB$ from $M \oplus C \oplus D^\B$; conversely if he guesses $\eAB$ correctly he can get $M$ from $\eAB \oplus C \oplus D^\B$). The probability of them doing the latter (and having $F=\text{\cmark}$) is precisely the probability $\Pr[(F=\text{\cmark}) \land (\eAB = \gB) \land (\eAC = \gC)]$ in Lemma~\ref{lem:rawkeyguess}. Therefore that lemma gives us (recalling that the message length is $|\clM|=2^l$, and substituting~\eqref{eq:par-cons-2} for the syndrome length):
\begin{equation}\label{eq:bnd_noleak}
\Pr[(F=\text{\cmark}) \land (M=M^\B=M^\C)] \leq 2^{-\kappa\dga^3\alpha^4l+2\lsyn} = \frac{2^{l-\kappa\dga^3\alpha^4l+2\xi (1-\gamma)(1-\alpha)h_2(\qhon) l}}{|\clM|},
\end{equation}
which is the desired result since $l=\lambda$.
\end{proof}

We can also obtain a similar statement in the presence of leakage, hence proving Theorem~\ref{thm:main}: 
\begin{theorem}\label{thm:unc_leak}
With the parameter choices as specified in \eqref{eq:par-cons-1}-\eqref{eq:par-cons-2}, Scheme~\ref{prot:DI-VKECM} is $(t(\lambda),0)$-uncloneable-secure with 
\[
t(\lambda) = (1-\kappa\dga^3\alpha^4 + 2\xi (1-\gamma)(1-\alpha)h_2(\qhon) +\nu)\lambda,
\]
if the total leakage between the client and the receiver during $\Enc(1^\lambda, M)$, and between Bob and Charlie during the cloning attack, is $\nu l$ bits. 
\end{theorem}
\begin{proof}
Recall from the description in Section~\ref{sec:secdefn} that the leakage can be interactive with arbitrarily many rounds of communication between the client and receiver during $\Enc(1^\lambda, M)$, and between Bob and Charlie during the cloning attack, as long as the total number of bits leaked is bounded. As stated there, we only consider leakage that happens after the devices receive their inputs --- for the client the only input is $X$, and for the receiver and Bob and Charlie, the inputs can be any information that they get. We shall also take the leakage to happen before the devices produce what we consider their outputs, which is $A$ for the client's device, $S$ for the receiver's device, and $\tA^\B$ and $\tA^\C$ for Bob and Charlie's devices. Note that in the protocol description we have the client getting $X$, $A$ in a single step and the receiver only getting $U$ after that, but these steps need not be so strictly time-ordered --- we could have the client inputting $X$ into their device and the receiver inputting $U$ into their device at the same time.\footnote{There is a slight complication for the receiver in this leakage model because the receiver actually gets information in two rounds; they first receive $U$ and have to produce $S$, and then they receive $C'$ --- there could also be some leakage from the client to the receiver after the receiver receives $C$, which the receiver could then pass on to Bob and Charlie. But this kind of leakage can be handled by the simple argument we used to deal with $\syn(\tA)$, so we are ignoring that here. Nevertheless, if such leakage happens, it should count towards the total number of bits leaked, and the dependence of $t(\lambda)$ on the number of bits leaked would be the same.} Moreover, in each round of communication the following happens: the device that receives the message communicated does some measurement on their quantum state depending on this message, their input, and previous messages and previous measurement outcomes, and sends a message to the other device depending on the output of this measurement and their input. Let the number of bits leaked during the encryption phase be $\nu_1l$, and the number of bits leaked afterwards (which can be e.g.~between Bob and Charlie) be $\nu_2l$, with $\nu_1+\nu_2=\nu$.

To bound $\Pr[(F=\text{\cmark})\land(M=M^\B=M^\C)]$, we use an argument essentially similar to the analysis of the syndromes in the Lemma~\ref{lem:rawkeyguess} proof: given any strategy in which the leakage bits are used to achieve $\Pr[(F=\text{\cmark}) \land (M=M_B=M_C)] = p$ for some $p\in[0,1]$, there is always another strategy that achieves $\Pr[(F=\text{\cmark}) \land (M=M_B=M_C)] \geq 2^{-\nu l}p$ \emph{without} using the leakage bits, which yields the desired result since we have already proven (under the no-leakage scenario) that the latter probability is upper bounded by~\eqref{eq:bnd_noleak}. Explicitly, the strategy is as follows: the client and receiver's devices will share the same initial state they did in the protocol with leakage, and $\nu_1 l$ extra bits of randomness (which can be simulated by shared entanglement), which will be broken up into blocks corresponding to each round of communication in the situation with leakage. Similarly Bob and Charlie's devices will share the same state and $\nu_2 l$ extra bits of randomness divided into blocks. The idea is that the devices will behave just as they did in the original strategy, by using the shared randomness to simulate the communication received from the other device.

We shall denote the $j$-th block of randomness shared between the client and receiver's devices by $s_j$. For the rest of this argument, for the sake of brevity, we shall talk about the client and receiver doing things instead of the client and receiver's devices. Suppose the client was supposed to communicate in the $j$-th round. In the new strategy, the client assumes that $s_{j-1}$ is the message they received from the receiver in the $(j-1)$-th round, and does the same measurement they would have done in the original strategy for this simulated message (the measurement also depends on the client's input, their previous measurement outcomes, and previous simulated messages). If the outcome of the measurement is not equal to $s_j$, then the client records a ``failure" for this round. The receiver behaves similarly in rounds where they were supposed to communicate. After, all the rounds are done, if the client and receiver have not recorded ``failure" at any point, they output as they would have done in the original strategy; otherwise they provide a random output. Once the outputs of the measurements are fixed, the protocol is deterministic, so a transcript of messages that is compatible with the client's outputs, and separately compatible with the receiver's outputs, is compatible with both of them. For any such fixed transcript, the shared randomness between the client and receiver is equal to it with probability $2^{-\nu_1l}$, and therefore with this probability the client and receiver actually output according to the original strategy.

Bob and Charlie behave similarly during their part of the new strategy, and output according to the old strategy with probability $2^{-\nu_2 l}$. Since the probability of satisfying $(F=\text{\cmark})\land(M=M^\B=M^\C)$ in the old strategy is $p$, the overall probability of satisfying in the new strategy is at least $2^{-\nu l}p$.
\end{proof}

Recall from the discussion below Definition~\ref{def:cl-sec} that in order for $(t(\lambda),0)$-uncloneability to be a nontrivial property, we require $t(\lambda)<\lambda$. For our scheme with the specified parameter choices, this property indeed holds as long as the leakage fraction $\nu$ satisfies (note that the right-hand-side in the expression below is strictly positive due to the condition~\eqref{eq:par-cons-1} in the protocol parameter specifications)
\begin{equation}\label{eq:maxleak}
\nu < \kappa\dga^3\alpha^4 - 2\xi(1-\gamma)(1-\alpha)h_2(\qhon).
\end{equation}
In other words, regarding the claim in our main theorem (Theorem~\ref{thm:main}), our scheme achieves nontrivial uncloneable security against dishonest devices with any value of $\nu$ up to the bound in~\eqref{eq:maxleak}.

\section{Single-decryptor encryption}
\label{sec:DISDECM}
In this section, we prove the following achievability result for DI-SDECM, via a suitable modification of our above DI-VKECM scheme:
\begin{theorem}\label{thm:main-dec}
There is a scheme for DI-SDECM with message space $\clM=\{0,1\}^\lambda$, such that:
\begin{enumerate}
\item It satisfies the completeness property~\eqref{eq:complete-2} given i.i.d.~honest devices with a constant level of noise; 
\item It achieves indistinguishable security;
\item There exists some $\nu>0$ such that the scheme achieves $(t(\lambda),0)$-anti-piracy security as per both Definition \ref{def:pir-sec} and Definition \ref{def:pir-sec-ind}, against dishonest devices with $\nu\lambda$ bits of leakage, where $t(\lambda)$ is a function satisfying $t(\lambda)<\lambda$ for all sufficiently large $\lambda$.
\end{enumerate}
\end{theorem}

Our protocol for achieving this is presented as Scheme~\ref{prot:DI-SDECM} below. In the description of Scheme~\ref{prot:DI-SDECM}, the devices, error-correction procedure, and parameter choices are to be understood as being the same as in Section~\ref{sec:DIVKECM}. 
Informally, the changes as compared to Scheme~\ref{prot:DI-VKECM} for DI-VKECM are simple: we have just isolated the parts of the previous encryption procedure that did not involve the message $m$ and placed them into the new key generation procedure $\KeyG(1^\lambda)$, while the parts that involved $m$ have been placed into the new encryption procedure $\Enc(m, \kenc, f)$ (together with the steps that involved selecting a random subset of the device outputs to XOR with other strings, hence moving this randomness into the encryption procedure rather than the previous ``key release'' procedure).
Another change is that in the new key generation procedure $\KeyG(1^\lambda)$ we do not have the message $m$ and thus cannot produce a string $m \oplus \otp$ (in contrast to the previous encryption procedure), so instead we simply output the one-time-pad $\otp$ by itself, which will be XOR'd with the message $m$ later during the new encryption procedure $\Enc(m, \kenc, f)$.\footnote{Intuitively, by the ``symmetry'' between the one-time-pad and its generated ciphertext as expressed in Fact~\ref{fc:otp}, this change should not modify the actual states produced in any manner that is relevant to the security proofs, as we shall discuss below.}
\begin{algorithm}[!h]
\caption{DI-SDECM with security parameter $\lambda$ and messages in $\clM = \{0,1\}^\lambda$}
\label{prot:DI-SDECM}
\vspace{0.3cm}
\begin{algorithmic}[1]
\Algphase{$\KeyG(1^\lambda)$:}
\State Devices of the form described in Section~\ref{sec:DIVKECM} are distributed between the client and receiver, with $l=\lambda$ \;
\State The client samples strings $X,U$ as follows: for each $i\in[l]$, set $X_i \in \{0,1\}$ uniformly at random, and independently set $U_i=\mathrm{keep},0,1$ with probabilities $1-\gamma,\gamma/2,\gamma/2$ respectively \;
\State The client inputs $X$ into their device and receives an output string $A\in\{0,1\}^l$ \;
\State The client sends $U$ to the receiver \;
\State The receiver inputs $U$ into their device, interpreting $U_i=\mathrm{keep}$ as $\perp$ for each $i$, and receives an output string $S\in \{0,1\}^l$ \;
\State The receiver sends $S$ to the client \;
\State The client tests if the number of $i\in [l]$ such that $U_i \neq \mathrm{keep}$ and $A_i\oplus S_i \neq X_i\cdot U_i$ is at most $(\gamma (1-\omega^*(\CHSH)) + {\dga }/{2})l$ \;
\State If the test passes then the client sets the flag to $F=\text{\cmark}$; otherwise the client sets the flag to $F=\text{\xmark}$ \;
\State The client samples $\otp\in \clM$ uniformly at random and sends it to the receiver \label{alg:SDotp} \;
\State The client stores $(X,U,A,\otp)$ as the encryption key and stores the value of the flag $F$; the receiver stores $\rho = \rho'\otimes\state{\otp}$ as the decryption key state, where $\rho'$ is the quantum state in the receiver's share of the devices\;

\Algphase{$\Enc(m, \kenc, f)$:}
\State Interpret the encryption key as $\kenc = (X,U,A,\otp)$ \;
\State Sample a string $\eX \in \{0, 1, \perp\}^l$ as follows: for each $i\in[l]$, set $\eX_i=\bot$ with probability $\alpha$, and otherwise set $\eX_i=X_i+2$ \;
\State Set a string $\eA \in \{0,1\}^l$ as follows: for each $i\in[l]$, set $\eA_i = 0$ if $U_i \neq \mathrm{keep}$ or $\eX_i=\bot$, and otherwise set $\eA_i = A_i$ \;
\State Compute the syndrome $\syn(\eA)$ following the error-correction procedure described above \;
\If{$f=\text{\cmark}$}
\State Output the ciphertext string $(m \oplus \otp \oplus \eA, \syn(\eA), \eX)$ \label{alg:SDciphertext} \;
\Else
\State Sample $M^\mathrm{fake}\in \clM$ uniformly at random and output the ciphertext string $(M^\mathrm{fake} \oplus \otp \oplus \eA, \syn(\eA), \eX)$ \; \;
\EndIf

\Algphase{$\Dec(c, \rho)$:}
\State Interpret $\rho$ as $\rho'\otimes\state{\otp}$ where $\rho'$ has $l$ qubit registers and $\otp \in \clM$; interpret the ciphertext string $c$ as $(N,\syn(\eA),\eX)$ \;
\State Input $\eX$ into the receiver's device and obtain the output string $\rS \in \{0,1\}^l $ \;
\State Set a string $\eS \in \{0,1\}^l$ as follows: for each $i\in[l]$, set $\eS_i = 0$ if $U_i \neq \mathrm{keep}$ or $\eX_i=\bot$, and otherwise set $\eS_i = \rS_i$ \;
\State Use $\eS,U,\eX$ and $\syn(\eA)$ to compute a guess $\gA$ for $\eA$ \;
\State Output $\widetilde{M} = \otp \oplus N \oplus \gA$ \label{alg:SDdecrypt} \;
\vspace{0.3cm}
\end{algorithmic}
\end{algorithm}
\newpage 

\subsection{Completeness and security of Scheme~\ref{prot:DI-SDECM}}
\label{sec:SDECM-proofs}

It is straightforward to see that Scheme~\ref{prot:DI-SDECM} satisfies the completeness requirements in Definition~\ref{def:sdecm}, since in the honest case the various registers produced are identical to Scheme~\ref{prot:DI-VKECM} up to some minor rearrangements regarding $\otp$, and hence in the completeness analysis in Section~\ref{subsec:corr-proof} carries over directly. As for security, we clearly have indistinguishable security, and it is also not too hard to modify our previous arguments to show anti-piracy security:
\begin{lemma}
Scheme~\ref{prot:DI-VKECM} is indistinguishable-secure. 
\end{lemma}
\begin{proof}
Indistinguishable security only involves the ciphertext produced by the encryption procedure, without having access to the decryption key $\Kdec$.
Observe that in the encryption procedure (focusing on the $F=\text{\cmark}$ case, since in the $F=\text{\xmark}$ case the ciphertext is trivially independent of $m$), the only part of the ciphertext that depends on the message $m$ is the $m \oplus \otp \oplus \eA$ part. However, recall that $\otp$ was generated as a uniformly random string independent of everything else (including the other parts of the ciphertext), and hence serves as a perfect one-time-pad. Hence indistinguishable security clearly holds (explicitly: one could e.g.~apply Fact~\ref{fc:otp} to conclude that $m \oplus \otp \oplus \eA$ is completely independent of $m \oplus \eA$, which is the only quantity depending on $m$).
\end{proof}

\begin{lemma}\label{lemma:DI-SDECM}
With the parameter choices as specified in \eqref{eq:par-cons-1}-\eqref{eq:par-cons-2}, Scheme~\ref{prot:DI-SDECM} is $(t(\lambda),0)$-anti-piracy-secure as per Definition \ref{def:pir-sec} and Definition \ref{def:pir-sec-ind} with 
\[
t(\lambda) = (1-\kappa\dga^3\alpha^4 + 2\xi (1-\gamma)(1-\alpha)h_2(\qhon) +\nu)\lambda,
\]
if the total leakage between the client and the receiver during $\KeyG(1^\lambda)$, and between Bob and Charlie during the pirating attack, is $\nu l$ bits. 
\end{lemma}
\begin{proof}
To show security by Definition \ref{def:pir-sec}, observe that for a pirating attack applied to Scheme~\ref{prot:DI-SDECM}, the steps performed by the dishonest parties and the registers available to them at each step are exactly the same as in a cloning attack applied to Scheme~\ref{prot:DI-VKECM}, except with the sole modification that the one-time-pad $\otp$ and the padded message $M \oplus \otp$ have had their roles interchanged. However, by Fact~\ref{fc:otp} these two values are precisely interchangeable, and hence our analysis in Section~\ref{subsec:secur-proof} for Scheme~\ref{prot:DI-VKECM} carries over exactly (noting that the final event of interest is also the same between the definitions of $(t(\lambda),0)$-anti-piracy security and $(t(\lambda),0)$-uncloneable security).

To show security by Definition \ref{def:pir-sec-ind}, we also just need to basically repeat the analysis in Section~\ref{subsec:secur-proof}, although we need to modify specific parts of the proofs rather than simply invoking the symmetry between $\otp$ and $M \oplus \otp$ as in the preceding case (since in this case the two independent messages make the register values somewhat different as compared to a cloning attack on Scheme~\ref{prot:DI-VKECM}). To highlight the key points: first observe that following the Theorem~\ref{thm:unc_noleak} proof, we again have that the maximum probability that $F=\text{\cmark}$ and Bob and Charlie can simultaneously guess their respective message values is the same as the maximum probability that $F=\text{\cmark}$ and they can simultaneously guess $\eAB$ and $\eAC$ respectively. Hence it suffices to bound the latter, which can be done by following a similar argument as the Lemma~\ref{lem:rawkeyguess} proof, now defining the sets $\clC_{\otp P^\B P^\C H^\B H^\C}, \dots ,\clC_{\emptyset}$ with respect to pirating attacks rather than cloning attacks (and letting $P^\B$ denote the value $M_1 \oplus \otp \oplus \eAB$ Bob receives in his ciphertext, and analogously $P^\C = M_2 \oplus \otp \oplus \eAC$ for Charlie); our aim is then to bound $\max_{\clC_{\otp P^\B P^\C H^\B H^\C}} \Pr[\Esucc]$. To do so, we handle the syndromes $H^\B H^\C$ by the same argument, to get an upper bound in terms of $\max_{\clC_{\otp P^\B P^\C}} \Pr[\Esucc]$. This time however, handling the registers $\otp P^\B P^\C$ is actually easier: in this scenario $P^\B$ is just a value $\otp \oplus \eAB$ one-time-padded with the value $M_1$ (which is not involved with any other register in this scenario), hence Fact~\ref{fc:otp} implies Bob could have generated it locally, and it can be removed in the sense that we have $\max_{\clC_{\otp P^\B P^\C}} \Pr[\Esucc] = \max_{\clC_{\otp P^\C}} \Pr[\Esucc]$.\footnote{This proof structure did not work for the Lemma~\ref{lem:rawkeyguess} proof because the values $D^\B = \otp \oplus \eAB$, $D^\C = \otp \oplus \eAC$ were padded with a \emph{common} value $\otp$, preventing us from applying Fact~\ref{fc:otp}. This issue was handled in that proof via the more elaborate argument where we allowed Bob and Charlie to ``win for free'' on $\overline{\clT}$ by giving them all the values $\otp_{\overline{\clT}}$.} Similarly, $P^\C$ is just a value $\otp \oplus \eAC$ one-time-padded with the value $M_2$ and can also be removed. This leaves only $\otp$, which (with $P^\B P^\C$ excluded from consideration) is just a register that could have been generated locally by the receiver before distributing states to Bob and Charlie, and hence can also be removed. Summarizing, this gives $\max_{\clC_{\otp P^\B P^\C}} \Pr[\Esucc] = \max_{\clC_{\emptyset}} \Pr[\Esucc]$, yielding the desired bound via Lemma~\ref{lem:game2prot} (recalling that $\Esucc$ is a ``stricter'' event than the one considered in that lemma).
\end{proof}

\section{Single-decryptor encryption of bits and trits}
\label{sec:bittritPA}

We now describe how Scheme~\ref{prot:DI-SDECM} can be modified to achieve ``perfect'' anti-piracy security under Definition \ref{def:pir-sec-ind} for the cases where the message space is a single bit or trit (i.e.~$\clM=\bbF_2$ or $\bbF_3$; later in this section we describe some obstacles faced in generalizing this result to larger $\bbF_p$, as well as a technical difficulty in applying it to DI-VKECM or Definition \ref{def:pir-sec} of anti-piracy security). Specifically, we obtain the following result: 
\begin{theorem}\label{thm:main-dec-bit}
There is a scheme for DI-SDECM with message space $\clM=\bbF_2$ or $\clM=\bbF_3$ which achieves the first two properties as in Theorem \ref{thm:main-dec} and $(0,\negl(\lambda))$-anti-piracy security as per Definition \ref{def:pir-sec-ind}, against dishonest devices with $\nu\lambda$ bits of leakage, as long as
\begin{equation}\label{eq:0pirmaxleak}
3\nu < \kappa\dga^3\alpha^4 - 2 (1-\gamma)(1-\alpha)h_2(\qhon).
\end{equation}
 \end{theorem}

Qualitatively, the idea is to implement a form of randomness extraction or privacy amplification~\cite{TSS+11,DPV+12}.
In our setting, from Lemma~\ref{lem:rawkeyguess} we have a bound on the probability of Bob and Charlie being able to simultaneously guess the ``raw key'' strings $\eAB,\eAC$. Our goal here to process $\eAB,\eAC$ into shorter strings such that their probability of being able to simultaneously guess the shorter strings is only negligibly larger than the trivial guessing probability, from which we could achieve $(0,\negl(\lambda))$-anti-piracy security.

Explicitly, the modification to Scheme~\ref{prot:DI-SDECM} is as follows. The input/output string length $l$ for the devices is still set equal to the security parameter $\lambda$ as before, but the message space is now $\clM=\bbF_p$ for $p=2$ or $3$. 
The following steps in the protocol are modified:
\begin{itemize}
\item In step~\ref{alg:SDotp} of $\KeyG$, the ``one-time-pad'' $\otp$ is instead drawn uniformly at random in $\bbF_p$. 
\item In step~\ref{alg:SDciphertext} of $\Enc$, instead of setting 
$(m \oplus \otp \oplus \eA, \syn(\eA), \eX)$ as the ciphertext, the following procedure is performed: a uniformly random value $\rVec\in\bbF_p^l$ is generated (independently of everything else), and the ciphertext is set as $(m \oplus \otp \oplus \eA\cdot\rVec, \syn(\eA), \eX, \rVec)$, where $\eA\cdot\rVec$ denotes inner product with respect to $\bbF_p$, and $\oplus$ is computed modulo $p$. Note that the random value $\rVec$ is included in the ciphertext.
\item In step~\ref{alg:SDdecrypt} of $\Dec$, the final output is instead computed as $\widetilde{M} = \otp \oplus D \oplus \gA\cdot\rVec$, where $\gA\cdot\rVec$ is computed using the value of $\rVec$ included in the ciphertext.
\end{itemize}
Qualitatively, the randomized value $\rVec$ serves as a method to ``extract'' the randomness in $\eA$. We highlight that similar to the standard setting for strong extractors (viewing $\rVec$ as the extractor seed), $\rVec$ can be revealed to the party trying to guess the inner-product value (and this is a necessary property in our context, since an honest receiver would need to use it in decrypting the message).

To show that this modified protocol indeed achieves $(0,\negl(\lambda))$-anti-piracy security, we basically need to show that if Bob and Charlie try to simultaneously guess their corresponding inner-product values, they only have a negligible advantage over the trivial success probability of $1/p$. This brings us to the main technical result of this section, which can be thought of as a simultaneous non-local version of the quantum Goldreich-Levin theorem \cite{AC02}:
\begin{lemma}\label{lem:bittrit-PA}
Consider a CQ state $\rho_{X^\B X^\C BC}$, where $X^\B$ and $X^\C$ are both classical registers taking values in $\bbF_p^l$, while $B$ and $C$ are quantum registers held by Bob and Charlie respectively. Further, suppose for any measurements Bob and Charlie can do on their quantum registers to produce outputs $G^\B$ and $G^\C$, we have
\[ \Pr[(G^\B = X^\B) \land (G^\C = X^\C)] \leq \delta,\]
where the probability is taken over the distribution of $X^\B X^\C$. Consider $\rVec^\B, \rVec^\C$ which are indepedently and uniformly distributed in $\bbF_p^l$. If $\rVec^\B$ and $\rVec^\C$ are given to Bob and Charlie respectively, and they do measurements on their quantum registers depending on $\rVec^\B$ and $\rVec^\C$, to produce outputs $G^\B(\rVec^\B)$ and $G^\C(\rVec^\C)$, then we have for $p=2, 3$,
\[ \Pr\left[\bigvee_{\substack{j, k: \\ j+k=0\mod p}}(G^\B(\rVec^\B) = X^\B\cdot \rVec^\B+j)\land(G^\C(\rVec^\C) = X^\C\cdot \rVec^\C+k)\right] \leq \frac{1}{p} + O(\delta^{1/3}), \]
where the probability is taken over the distribution of $X^\B X^\C, \rVec^\B \rVec^\C$.
\end{lemma}
Note that the event whose probability is being upper bounded in the lemma for $p=2$ is that Bob and Charlie both guess $X^\B\cdot\rVec^\B$ and $X^\C\cdot\rVec^\C$ respectively right, or they both guess wrong. If $p=3$, the event is that they both guess right, or one of their guesses is off by 1, and the other's guess is off by 2. An upper bound on the probability of this event obviously implies an upper bound on the probability of them both guessing right. But as we shall see later, in order to prove a version of this lemma with leakage, we shall actually need the fact that we can prove an upper bound on the probability of this larger event in the case without leakage.
\begin{proof}[Proof of Lemma \ref{lem:bittrit-PA}]
We shall assume that for any fixed $X^\B X^\C = x^\B x^\C$, $\rho_{BC|x^\B x^\C}$ is a pure state $\state{\rho}_{BC|x^\B x^\C}$. This is without loss of generality, because we can always consider a purification of the state --- Bob and Charlie may not have access to the purifying registers, but that will not matter for the purposes of our argument.

The proof will be via contradiction: supposing Bob and Charlie have a strategy in which they use $\rVec^\B$ and $\rVec^\C$ to produce outputs $G^\B(\rVec^\B)$ and $G^\C(\rVec^\C)$ such that
\begin{equation}\label{eq:avg-cont}
\Pr_{X^\B X^\C \rVec^\B \rVec^\C}\left[\bigvee_{\substack{j, k: \\ j+k=0\mod p}}(G^\B(\rVec^\B) = X^\B\cdot \rVec^\B+j)\land(G^\C(\rVec^\C) = X^\C\cdot \rVec^\C+k)\right] > \frac{1}{p} + \eps,
\end{equation}
we shall construct a procedure for Bob and Charlie to learn $X^\B$ and $X^\C$ with probability greater than $\Omega(\eps^3)$. To begin, note that for each value of $x^\B, x^\C, \rvec^\B, \rvec^\C$, without loss of generality we can model Bob and Charlie's procedure to obtain their guesses for $x^\B\cdot \rvec^\B$ and $x^\C \cdot \rvec^\C$ (using $\ket{\rho}_{BC|x^\B x^\C}$ and $\rvec^\B, \rvec^\C$) as follows: they each attach a $p$-dimensional ancilla, which is to be their output register, to their halves of $\ket{\rho}_{BC|x^\B x^\C}$, and apply unitaries controlled on $\rvec^\B$ and $\rvec^\C$ respectively to the state and ancillas. The output value is then obtained by measuring their ancilla registers in the computational basis. Suppose the action of the unitaries, which we shall call $U^\B_{\IP}$ and $U^\C_{\IP}$, is as follows:
\begin{align*}
& U^\B_{\IP}\otimes U^\C_{\IP}\ket{\rvec^\B \rvec^\C}_{\rVec^\B \rVec^\C}\ket{00}_{Z^\B Z^\C}\ket{\rho}_{BC|x^\B x^\C} \\
= & \ket{\rvec^\B \rvec^\C}_{\rVec^\B \rVec^\C}\sum_{j,k\in \bbF_p}\alpha_{x^\B x^\C jk}^{\rvec^\B \rvec^\C}\ket{x^\B\cdot \rvec^\B + j}_{Z^\B}\ket{x^\C\cdot \rvec^\C + k}_{Z^\C}\ket{\sigma^{\rvec^\B \rvec^\C}}_{BC|x^\B x^\C jk}.
\end{align*}
In the above, $U^\B_{\IP}$ acts on the registers $\rVec^\B Z^\B B$ and $U^\C_{\IP}$ acts on $\rVec^\C Z^\C C$. After their action, the $\rVec^\B \rVec^\C$ registers remain unchanged; there is some superposition of answers of the form $x^\B\cdot \rvec^\B + j$ and $x^\C\cdot \rvec^\C + k$ (where all additions are modulo $p$) on the answer registers $Z^\B, Z^\C$, and corresponding to these answers, the state on the $BC$ registers is $\ket{\sigma^{\rvec^\B \rvec^\C}}_{BC|x^\B x^\C jk}$. The probability of getting answers $x^\B \cdot x^\C+j$ and $x^\C \cdot \rvec^\C + k$ is $|\alpha^{\rvec^\B \rvec^\C}_{x^\B x^\C jk}|^2$. 

Also, note that
by applying Markov's inequality on \eqref{eq:avg-cont} we have that
\begin{align*}
& \Pr_{X^\B X^\C}\left[\Pr_{\rVec^\B \rVec^\C}\left[\bigvee_{\substack{j, k: \\ j+k=0\mod p}}(G^\B(\rVec^\B) = X^\B\cdot \rVec^\B+j)\land(G^\C(\rVec^\C) = X^\C\cdot \rVec^\C+k)\right] \leq \frac{1}{p} + \frac{\eps}{2}\right] \\
\leq & \frac{1-\frac{1}{p} - \eps}{1-\frac{1}{p} - \frac{\eps}{2}} \leq 1 - \frac{\eps}{2}.
\end{align*}
This means that with probability at least $\frac{\eps}{2}$ over the distribution of $X^\B X^\C$, Bob and Charlie's average probability (over the distribution of $\rVec^\B \rVec^\C$) of outputting $x^\B \cdot \rvec^\B+j$ and $x^\C \cdot \rvec^\C+k$ for $j+k = 0 \mod p$ is at least $\frac{1}{p} + \frac{\eps}{2}$. We shall call pairs $x^\B x^\C$ for which this is true ``good'' pairs. We shall now concentrate on a good pair $x^\B x^\C$, and henceforth in the analysis, we shall drop the $x^\B x^\C$ dependence from $\ket{\rho}_{BC|x^\B x^\C}$ and $\ket{\sigma^{\rvec^\B \rvec^\C}}_{BC|x^\B x^\C jk}$ and $\alpha^{\rvec^\B \rvec^\C}_{x^\B x^\C jk}$ (with the understanding that we are now only focusing on a fixed good pair $x^\B x^\C$). 

Since $x^\B x^\C$ is a good pair, we have that
\begin{equation}\label{eq:good-bias}
\frac{1}{p^{2l}}\sum_{\rvec^\B \rvec^\C}\sum_{j+k=0}\Big|\alpha^{\rvec^\B \rvec^\C}_{jk}\Big|^2 \geq \frac{1}{p} + \frac{\eps}{2}.
\end{equation}
We shall now show there exists a procedure independent of $x^\B x^\C$ that Bob and Charlie can perform on $\ket{\rho}_{BC}$, such that (given that $x^\B x^\C$ is a good pair) Bob and Charlie output $x^\B x^\C$ with probability at least $\Omega(\eps^2)$.
Essentially, Bob and Charlie will carry out the Bernstein-Vazirani algorithm independently, using $U^\B_{\IP}$ and $U^\C_{\IP}$ as noisy oracles. The circuit for doing so is depicted below in Figure~\ref{fig:BV-ckt}.
\begin{figure}[!h]
\centering
\begin{tikzpicture}
\draw (0.1,0) rectangle (2.9,4) node(IP1) [midway] {$U^\B_{\IP}\otimes U^\C_{\IP}$};
\draw (6.9,0) rectangle (9.7,4) node(IP2) [midway] {$(U^\B_{\IP})^\dagger\otimes (U^\C_{\IP})^\dagger$};
\draw (-1.7,2) rectangle (-0.6,4) node(F1) [midway] {$F_p^{\otimes 2l}$};
\draw (-1.7,-1.8) rectangle (-0.6,-0.4) node(F2) [midway] {$F_p^{\otimes 2}$};
\draw (3.3,-1.8) rectangle (6.5,-0.4) node(ADD) [midway] {$\mathrm{ADD}^\B\otimes \mathrm{ADD}^\C$};
\draw (10.4,2) rectangle (11.5,4) node(F3) [midway] {$F_p^{\dagger \otimes 2l}$};
\draw (10.4, -1.8) rectangle (11.5, -0.4) node(F4) [midway] {$F_p^{\otimes 2}$};

\draw (-2.2, 3.4) -- node [at start, xshift=-0.5cm] {$\ket{0^l}_{\rVec^\B}$} (-1.7,3.4);
\draw (-0.6, 3.4) -- (0.1, 3.4);
\draw (2.9, 3.4) -- (6.9,3.4);
\draw (9.7, 3.4) -- (10.4, 3.4);
\draw (11.5, 3.4) -- (12, 3.4);

\draw (-2.2, 2.6) -- node [at start, xshift=-0.5cm] {$\ket{0^l}_{\rVec^\C}$} (-1.7,2.6);
\draw (-0.6,2.6) -- (0.1, 2.6);
\draw (2.9, 2.6) -- (6.9,2.6);
\draw (9.7, 2.6) -- (10.4, 2.6);
\draw (11.5, 2.6) -- (12, 2.6);

\draw(-2.2, 1.6) -- node [at start, xshift=-0.5cm] {$\ket{0}_{Z^\B}$} (0.1, 1.6);
\draw (-0.6, 1.6) -- (0.1, 1.6);
\draw (2.9, 1.6) -- (6.9, 1.6);
\draw (9.7, 1.6) -- (12, 1.6);

\draw(-2.2, 1) -- node [at start, xshift=-0.5cm] {$\ket{0}_{Z^\C}$} (0.1, 1);
\draw (-0.6, 1) -- (0.1, 1);
\draw (2.9, 1) -- (6.9, 1);
\draw (9.7, 1) -- (12, 1);

\draw(-2.2, 0.4) -- node [at start, xshift=-0.5cm] {$\ket{\rho}_{BC}$} (0.1, 0.4);
\draw (-0.6, 0.4) -- (0.1, 0.4);
\draw (2.9, 0.4) -- (6.9, 0.4);
\draw (9.7, 0.4) -- (12, 0.4);

\draw (-2.2, -0.8) -- node [at start, xshift=-0.8cm] {$\ket{p-1}_{\tZ^\B}$} (-1.7,-0.8);
\draw (-0.6,-0.8) -- (3.3,-0.8);
\draw (6.5,-0.8) -- (10.4, -0.8);
\draw (11.5, -0.8) -- (12, -0.8);

\draw (-2.2, -1.4) -- node [at start, xshift=-0.8cm] {$\ket{p-1}_{\tZ^\C}$} (-1.7,-1.4);
\draw (-0.6,-1.4) -- (3.3,-1.4);
\draw (6.5,-1.4) -- (10.4, -1.4);
\draw (11.5, -1.4) -- (12, -1.4);

\draw (4,1.6) -- node [at start, yshift=-0.03cm] {$\bullet$} (4,-0.4);
\draw (5.8,1) -- node [at start, yshift=-0.03cm] {$\bullet$} (5.8,-0.4);
\end{tikzpicture}
\caption{Circuit to compute $x^\B, x^\C$ in good set}
\label{fig:BV-ckt}
\end{figure}

Bob's registers in this circuit are $\rVec^\B Z^\B \tZ^\B B$, and Charlie's registers are $\rVec^\C Z^\C \tZ^\C B$. The circuit essentially does the following: it prepares a uniform superposition of $\rvec^\B$ and $\rvec^\C$ by applying $F_p^{\otimes l}$ on the $\ket{0^l}$ state in the $\rVec^\B$ and $\rVec^\C$ registers respectively. Here $F_p$ is the $\bbF_p$ Fourier transform, whose action on computational basis states is given by
\[ F_p\ket{j} = \frac{1}{\sqrt{p}}\sum_{k\in\bbF_p}\omega^{jk}\ket{k}, \]
with $\omega$ being the $p$-th root of unity. At the same time as $F_p^{\otimes l}$, the circuit applies $F_p^{\otimes 2}$ to the registers $\tZ^\B$ and $\tZ^\C$, which are initialized with $\ket{p-1}\ket{p-1}$. The effect of the Fourier transforms on these registers is
\[ F_p^{\otimes 2}\ket{p-1}_{\tZ^\B}\ket{p-1}_{\tZ^\C} = \frac{1}{p}\sum_{j',k'\in \bbF_p}\omega^{(p-1)j'+ (p-1)k'}\ket{j'k'}_{\tZ^\B \tZ^\C} = \frac{1}{p}\sum_{j',k'\in \bbF_p}\omega^{-j' - k'}\ket{j'k'}_{\tZ^\B \tZ^\C}.\]
After this, the unitaries $U^\B_{\IP}\otimes U^\C_{\IP}$ are acted on the registers $\rVec^\B \rVec^\C Z^\B Z^\C BC$, and then the $\mathrm{ADD}^\B$ and $\mathrm{ADD}^\C$ gates acting on the registers $Z^\B \tZ^\B$ and $Z^\C \tZ^\C$ registers respectively are applied. The $\mathrm{ADD}^\B$ gate essentially adds (modulo $p$) the value in the $Z^\B$ register to the value in the $\tZ^\B$ register, and similarly, the $\mathrm{ADD}^\C$ gate adds the value in the $Z^\C$ register to the $\tZ^\C$ register. After this, the inverses of the $U^\B_{\IP}\otimes U^\C_{\IP}$ gates and the $F_p^{\otimes (2l+2)}$ gates are added on their respective registers. Finally, Bob measures the $\rVec^\B$ register and Charlie measures the $\rVec^\C$ register, both in the computational basis. Note that all these steps can be carried out by Bob acting only on his registers, and Charlie acting only on his registers.

The probability that Bob and Charlie measure the $\rVec^\B \rVec^\C$ registers and both get the correct values $x^\B$ and $x^\C$ is at least 
\[ \left|\bra{x^\B x^\C}_{\rVec^\B \rVec^\C}\bra{00}_{Z^\B Z^\C}\bra{\rho}_{BC}\bra{p-1,p-1}_{\tZ^\B \tZ^\C}U_5U_4U_3U_2U_1\ket{0^l0^l}_{\rVec^\B \rVec^\C}\ket{00}_{Z^\B Z^\C}\ket{\rho}_{BC}\ket{p-1,p-1}_{\tZ^\B \tZ^\C}\right|^2, \]
where $U_1 = F_p^{\otimes (2l+2)}$, $U_2 = U^\B_{\IP}\otimes U^\C_{\IP}$, $U_3 = \mathrm{ADD}^\B\otimes\mathrm{ADD}^\C$, $U_4 = (U^\B_{\IP})^\dagger\otimes (U^\C_{\IP})^\dagger$ and $U_5 = F_p^{\otimes (2l+2)}$. To calculate this, we observe
\begin{align*}
& U_3U_2U_1\ket{0^l0^l}_{\rVec^\B \rVec^\C}\ket{00}_{Z^\B Z^\C}\ket{\rho}_{BC}\ket{p-1,p-1}_{\tZ^\B \tZ^\C} \\
= & U_3U_2\left(\frac{1}{p^{l+1}}\sum_{\rvec^\B,\rvec^\C \in \bbF_p^l}\ket{\rvec^\B \rvec^\C}_{\rVec^\B \rVec^\C}\ket{00}_{Z^\B Z^\C}\ket{\rho}_{BC}\sum_{j',k'\in \bbF_p}\omega^{-j'-k'}\ket{j',k'}_{\tZ^\B \tZ^\C}\right) \\
= & U_3\left(\frac{1}{p^{l+1}}\sum_{\rvec^\B, \rvec^\C \in \bbF_p^l}\ket{\rvec^\B \rvec^\C}\sum_{j,k\in \bbF_p}\alpha_{jk}^{\rvec^\B \rvec^\C}\ket{x^\B\cdot \rvec^\B + j,x^\C\cdot \rvec^\C + k}\ket{\sigma^{\rvec^\B \rvec^\C}}_{jk}\sum_{j',k'\in \bbF_p}\omega^{-j'-k'}\ket{j',k'}\right) \\
= & \frac{1}{p^{l+1}}\sum_{\rvec^\B, \rvec^\C}\ket{\rvec^\B \rvec^\C}\sum_{j,k}\alpha_{jk}^{\rvec^\B \rvec^\C}\ket{x^\B\cdot \rvec^\B + j,x^\C\cdot \rvec^\C + k}\ket{\sigma^{\rvec^\B \rvec^\C}}_{jk}\sum_{j',k'}\omega^{-j'-k'}\ket{x^\B\cdot \rvec^\B + j + j',x^\C \cdot \rvec^\C + k + k'} \\
= & \frac{1}{p^{l+1}}\sum_{\rvec^\B, \rvec^\C}\ket{\rvec^\B \rvec^\C}\sum_{j,k}\alpha_{jk}^{\rvec^\B \rvec^\C}\omega^{x^\B\cdot \rvec^\B+j + x^\C\cdot \rvec^\C + k}\ket{x^\B\cdot \rvec^\B + j,x^\C\cdot \rvec^\C + k}\ket{\sigma^{\rvec^\B \rvec^\C}}_{jk}\sum_{j',k'}\omega^{-j'-k'}\ket{j',k'}.
\end{align*}
Similarly,
\begin{align*}
& \bra{x^\B x^\C}_{\rVec^\B \rVec^\C}\bra{00}_{Z^\B Z^\C}\bra{\rho}_{BC}\bra{p-1,p-1}_{\tZ^\B \tZ^\C}U_5U_4 \\
= & \left(\frac{1}{p^{l+1}}\sum_{\rvec^\B, \rvec^\C}\omega^{-x^\B\cdot \rvec^\B - x^\C\cdot \rvec^\C}\bra{\rvec^\B \rvec^\C}\bra{00}_{Z^\B Z^\C}\bra{\rho}_{BC}\sum_{j',k'}\omega^{j'+k'}\bra{j',k'}_{\tZ^\B \tZ^\C}\right)U_4 \\
= & \frac{1}{p^{l+1}}\sum_{\rvec^\B, \rvec^\C}\omega^{-x^\B\cdot \rvec^\B - x^\C\cdot \rvec^\C}\bra{\rvec^\B \rvec^\C}\sum_{j,k}(\alpha^{\rvec^\B \rvec^\C}_{jk})^*\bra{x^\B\cdot \rvec^\B + j,x^\C \cdot \rvec^\C + k}\bra{\sigma^{\rvec^\B \rvec^\C}}_{jk}\sum_{j',k'}\omega^{j'+k'}\bra{j',k'}_{\tZ^\B \tZ^\C}.
\end{align*}
Note that the states $\ket{\rvec^\B \rvec^\C}\ket{x^\B\cdot \rvec^\B + j,x^\C\cdot \rvec^\C + k}\ket{\sigma^{\rvec^\B \rvec^\C}}_{jk}$ are orthogonal for different values of $\rvec^\B, \rvec^\C, j$ and $k$. Therefore we have,
\begin{align}
& \bra{x^\B x^\C}\bra{00}\bra{\rho}\bra{p-1,p-1}U_5U_4U_3U_2U_1\ket{0^l0^l}\ket{00}\ket{\rho}\ket{p-1,p-1} \nonumber \\
= & \frac{1}{p^{2l}}\sum_{\rvec^\B, \rvec^\C}\sum_{j,k \in \bbF_p}\left|\alpha^{\rvec^\B \rvec^\C}_{jk}\right|^2\omega^{j+k} \nonumber \\
= & \sum_{k'\in \bbF_p}\left(\sum_{j'\in \bbF_p}\frac{1}{p^{2l}}\sum_{\rvec^\B, \rvec^\C}\left|\alpha^{\rvec^\B \rvec^\C}_{j',k'-j'}\right|^2\right)\omega^{k'}, \label{eq:IP-omega}
\end{align}
where $k'-j'$ in the subscript of $\alpha^{\rvec^\B \rvec^\C}_{j',k'-j'}$ is meant to be interpreted modulo $p$.

For $p=2$, $\omega=-1$, so the above expression is
\[ \frac{1}{2^{2l}}\sum_{\rvec^\B \rvec^\C}\left|\alpha^{\rvec^\B \rvec^\C}_{00}\right|^2 + \frac{1}{2^{2l}}\sum_{\rvec^\B \rvec^\C}\left|\alpha^{\rvec^\B \rvec^\C}_{11}\right|^2 - \frac{1}{2^{2l}}\sum_{\rvec^\B \rvec^\C}\left|\alpha^{\rvec^\B \rvec^\C}_{01}\right|^2 - \frac{1}{2^{2l}}\sum_{\rvec^\B \rvec^\C}\left|\alpha^{\rvec^\B \rvec^\C}_{10}\right|^2\]
For a good $x^\B x^\C$, by \eqref{eq:good-bias}, we have $\dfrac{1}{2^{2l}}\displaystyle\sum_{\rvec^\B \rvec^\C}(|\alpha^{\rvec^\B \rvec^\C}_{00}|^2 + |\alpha^{\rvec^\B\rvec^\C}|^2) \geq \frac{1}{2} + \frac{\eps}{2}$, which means that $\dfrac{1}{2^{2l}}\displaystyle\sum_{\rvec^\B \rvec^\C}(|\alpha^{\rvec^\B \rvec^\C}_{01}|^2 + |\alpha^{\rvec^\B \rvec^\C}_{10}|^2)$ is at most $\frac{1}{2} -\frac{\eps}{2}$. Therefore, the probability of Bob and Charlie learning $x^\B x^\C$ from the good set is at least
\[ \left|\frac{1}{2} + \frac{\eps}{2} - \frac{1}{2} + \frac{\eps}{2}\right|^2 \geq \eps^2.\]
Since the probability of $X^\B X^\C$ being from the good set is at least $\frac{\eps}{2}$, the overall probability of Bob and Charlie learning $X^\B X^\C$ is at least $\frac{\eps^3}{2} = \Omega(\eps^3)$ as claimed.

For $p=3$, \eqref{eq:IP-omega} instead becomes
\begin{align*}
& \frac{1}{3^{2l}}\sum_{\rvec^\B \rvec^\C}\bigg(\left|\alpha^{\rvec^\B \rvec^\C}_{00}\right|^2 + \left|\alpha^{\rvec^\B \rvec^\C}_{12}\right|^2 + \left|\alpha^{\rvec^\B \rvec^\C}_{21}\right|^2 + \omega\left(\left|\alpha^{\rvec^\B \rvec^\C}_{01}\right|^2 + \left|\alpha^{\rvec^\B \rvec^\C}_{10}\right|^2 + \left|\alpha^{\rvec^\B \rvec^\C}_{22}\right|^2\right) \\
& \quad \quad + \omega^2\left(\left|\alpha^{\rvec^\B \rvec^\C}_{11}\right|^2 + \left|\alpha^{\rvec^\B \rvec^\C}_{02}\right|^2 + \left|\alpha^{\rvec^\B \rvec^\C}_{20}\right|^2\right) \bigg).
\end{align*}
We can thus write the probability of learning a good $x^\B x^\C$ as
\[ |a_0 + a_1\omega + a_2\omega^2|^2 = a_0^2 + a_1^2 + a_2^2 - a_0a_1 - a_0a_2 -a_1a_2,\]
where $a_0 + a_1 + a_2=1$, and we know from \eqref{eq:good-bias} that $a_0 = \frac{1}{3^{2l}}\displaystyle\sum_{\rvec^\B \rvec^\C}\left(\left|\alpha^{\rvec^\B \rvec^\C}_{00}\right|^2 + \left|\alpha^{\rvec^\B \rvec^\C}_{12}\right|^2 + \left|\alpha^{\rvec^\B \rvec^\C}_{21}\right|^2\right)$ $\geq \frac{1}{3} + \frac{\eps}{2}$. Writing $a_2$ in terms of $a_0, a_1$, the above expression attains its minimum value w.r.t. $a_1$ when its derivative w.r.t. $a_1$ is 0. This happens at $a_1 = \frac{1-a_0}{2}$, and the corresponding value of the expression is $\frac{1}{4}(3a_0 - 1)^2$. Substituting $a_0 \geq \frac{1}{3} + \frac{\eps}{2}$, we get that the probability is always at least $\frac{9\eps^2}{16}$ for $x^\B x^\C$ in the good set. Thus the probability of learning $X^\B X^\C$ overall is at least $\frac{9\eps^3}{32} = \Omega(\eps^3)$.
\end{proof}
We make some observations from the proof of Lemma \ref{lem:bittrit-PA}. First, the lower bound for learning a good $x^\B x^\C$ in the general case is $\left|\sum_{j=0}^{p-1} a_j \omega^j\right|^2$, where the $a_j$-s form a probability distribution, and $a_0 \geq \frac{1}{p} + \frac{\eps}{2}$. Our proof works for $p=2, 3$ because we can show the quantity is bounded away from zero under the two constraints on the $a_j$-s that we have. However, this does not seem to hold for $p>3$, and hence our proof approach here does not work straightforwardly for such $p$. In particular, for any even-valued $p > 3$, note that one of the powers of $\omega$ is $-1$; hence, if we make the $a_j$-s corresponding to the $-1$ root have the same value as $a_0$, and give the rest of the $a_j$-s equal values, then $\sum_{j=0}^{p-1} a_j \omega^j = 0$. As for odd-valued $p > 3$, we note that for $p=5$, viewing the terms $a_j \omega^j$ as vectors in the complex plane lets us see geometrically that (as long as $\eps$ is not too large) there is also a feasible choice of $a_j$ values such that $\sum_{j=0}^{p-1} a_j \omega^j = 0$, and the construction should also generalize to all other odd-valued $p>3$. Explicitly\footnote{The geometric intuition here is that $\sum_{j=0}^{p-1} a_j \omega^j$ is the point in the complex plane given by head-to-tail summation of the vectors $a_j \omega^j$, which basically form a ``non-closed pentagon'' with side lengths $a_j$. The specified $a_j$ values yield an endpoint of this vector sum that (due to the symmetry in the $a_1,a_2,a_3,a_4$ choices) always lies on the real axis, and moves from the positive half to the negative half as $t$ ranges from $0$ to $1-a_0$.}: set $a_0 = \frac{1}{5} + \frac{\eps}{2}$, and let $a_1=a_4=\frac{1-a_0-t}{4}$, $a_2=a_3=\frac{1-a_0+t}{4}$ for a parameter $t\in[0,1-a_0]$ whose exact value we shall choose later.
Observe that since $a_1=a_4$ and $a_2=a_3$, these values always sum to a value with zero imaginary part, i.e.~$\Im\left(\sum_{j=0}^{p-1} a_j \omega^j\right) = 0$. Furthermore, if $t=0$ (i.e.~$a_1=a_2=a_3=a_4=\frac{1-a_0}{4}<\frac{1}{5}$) then $\Re\left(\sum_{j=0}^{p-1} a_j \omega^j\right) > 0$, whereas if $t=1-a_0$ (i.e.~$a_1=a_4=0$ and $a_2=a_3=\frac{1-a_0}{2}$) then $\Re\left(\sum_{j=0}^{p-1} a_j \omega^j\right) < 0$ (as long as $\eps$ is not too large). Hence by continuity in $t$, there exists some $t\in[0,1-a_0]$ such that $\Re\left(\sum_{j=0}^{p-1} a_j \omega^j\right) = 0$ exactly, yielding the desired counterexample.\footnote{We remark that this argument basically relies on having (at least) one ``free parameter'' $t$ to adjust the vector sum endpoint; therefore, it should generalize to larger odd-valued $p$ as well, but an analogous construction is not available in the $p=3$ case because if e.g.~we were to try setting $a_0 = \frac{1}{3} + \frac{\eps}{2}$ and $a_1=a_2$, there are no ``degrees of freedom'' left after accounting for normalization.} 
Still, we currently do not know whether this difficulty for the $p>3$ case is simply a limitation of this proof approach, or whether Lemma~\ref{lem:bittrit-PA} fundamentally does not hold in that case.

Additionally, we note that it is fine for the purposes of the proof if the initial upper bound on Bob and Charlie being able to guess $X^\B$ and $X^\C$ was obtained in the presence of some leakage between Bob and Charlie. The bound on the probability of Bob and Charlie guessing $X^\B\cdot\rVec^\B$ and $X^\C\cdot\rVec^\C$ holds in the presence of the same amount of leakage about $X^\B$ and $X^\C$, as long as additional leakage about $\rVec^\B$ and $\rVec^\C$ does not happen. Note however that the above proof really does not work if $\rVec^\B$ is fully leaked to Charlie and $\rVec^\C$ is fully leaked to Bob. This is because, in the contradiction step of the proof, we would then have to assume that $U^\B_{\IP}$ also takes a copy of $\rVec^\C$ as input and $U^\C_{\IP}$ takes a copy of $R^\B$ as input. But to run the Bernstein-Vazirani algorithm, Bob needs to have a uniform superposition of over $\rVec^\B$, which is uncorrelated with everything else, and Charlie needs to have a uniform superposition over $\rVec^\C$ which is uncorrelated with everything else. Of course, for similar reasons, the proof does not work if we try to take inner product with the same $\rVec$ for both Bob and Charlie.

However, for a bounded amount of leakage about $X^\B, X^\C$ or $\rVec^\B, \rVec^\C$, similar to Theorem \ref{thm:unc_leak}, we can still come up with a new strategy for guessing $X^\B\cdot \rVec^\B$ and $X^\C\cdot\rVec^\C$ without leakage, given a strategy to guess them with leakage. This allows us to convert an upper bound on the ``success probability'' for the latter to an upper bound for the former. The analysis needs to be more fine-grained here however, since the kind of bound we are hoping to get with leakage is $\frac{1}{p} + \negl(l)$, instead of $\negl(l)$, and we need to make use of the fact that the final bound in Lemma \ref{lem:bittrit-PA} is an upper bound on the probability of Bob and Charlie both guessing right \emph{or} both guessing wrong (that is, for $p=2$, as described previously; the $p=3$ case is similar but involves the $j+k=0\mod p$ condition more directly). We do this analysis in the following corollary.
\begin{cor}\label{cor:PA-leak}
In the same setting as Lemma \ref{lem:bittrit-PA}, if the bound $\delta$ for Bob and Charlie guessing $X^\B$ and $X^\C$ holds without any leakage, then with $\nu l$ bits of leakage, we have,
\[ \Pr\left[\bigvee_{\substack{j, k: \\ j+k=0\mod p}}(G^\B(\rVec^\B) = X^\B\cdot \rVec^\B+j)\land(G^\C(\rVec^\C) = X^\C\cdot \rVec^\C+k)\right] \leq \frac{1}{p} + 2^{\nu l}\cdot O(\delta^{1/3}).\]
\end{cor}
\begin{proof}
We shall only present the analysis for $p=2$; the $p=3$ analysis is very similar. As before, if $\nu l$ bits were leaked in the original strategy $\clP$, then to obtain a new strategy without leakage, Bob and Charlie will share $\nu l$ bits of randomness, which they will use to simulate messages of $\clP$ without communicating. At the end, if the randomness is consistent with all their measurement outcomes, they will output according to $\clP$, otherwise they will output uniformly at random (using independent uniform bits). This means there are three different possibilities: Bob and Charlie both output a uniformly at random, one of them outputs uniformly at random and the other outputs according to $\clP$, and both of them output according to $\clP$.

We now compute the probability of Bob and Charlie both being right or being wrong in this new strategy without leakage.
The probability that they both output according to $\clP$ is of course $2^{-\nu l}$, and conditioned on them doing so, the probability of them both being right or both being wrong is the same as in $\clP$, which is, say, $\frac{1}{2} + \eps$. Now suppose the probability that they both output uniformly at random is $q_1$, and the probability that one of them outputs uniformly at random and the other according to $\clP$ is $q_2$, where $q_1 + q_2 = 1 - 2^{-\nu l}$. In the first case, the probability that they both output right or they both output wrong is $\frac{1}{2}\cdot \frac{1}{2} + \frac{1}{2}\cdot\frac{1}{2} = \frac{1}{2}$. In the latter case, let us first focus on the party who outputs according to $\clP$: we don't know the actual value of the probability that their output is correct, but call this value $\beta$. With this (and using the fact that the other party simply produces an independent uniform output in this case), the probability of both the parties outputting right or both of them outputting wrong in this case is $\frac{1}{2}\cdot \beta + \frac{1}{2}\cdot\left(1-\beta\right) = \frac{1}{2}$. 
Thus, their overall probability of both outputting right or both outputting wrong in this strategy without leakage is
\[ q_1\cdot \frac{1}{2} + q_2\cdot\frac{1}{2} + 2^{-\nu l}\cdot\left(\frac{1}{2}+\eps\right) = \frac{1}{2} + 2^{-\nu l}\cdot\eps. \]
Since we know from Lemma \ref{lem:bittrit-PA} that this probability is at most $\frac{1}{2} + O(\delta^{1/3})$, this gives us $\eps \leq 2^{\nu l}\cdot O(\delta^{1/3})$.
\end{proof}

\subsection{$(0,\negl(\lambda))$-anti-piracy security for bits and trits}

With the above result, we can now prove fairly straightforwardly that the modified protocol achieves $(0,\negl(\lambda))$-anti-piracy security under Definition~\ref{def:pir-sec-ind} (independent random challenge plaintexts), thus achieving Theorem \ref{thm:main-dec-bit} (the fact that it also satisfies completeness and indistinguishable security is easily seen from similar arguments as in Section~\ref{sec:SDECM-proofs}). 
\begin{lemma}
Let the total number of leaked bits between the client and the receiver during $\Enc(1^\lambda, M)$, and between Bob and Charlie during the pirating attack, be $\nu l$. 
Then with the parameter choices as specified in \eqref{eq:par-cons-1}-\eqref{eq:par-cons-2}, Scheme~\ref{prot:DI-SDECM} with the modifications described in this section is $(0,\negl(\lambda))$-anti-piracy-secure under Definition~\ref{def:pir-sec-ind}, as long as $\nu$ satisfies
\begin{equation}
3\nu < \kappa\dga^3\alpha^4 - 2 (1-\gamma)(1-\alpha)h_2(\qhon).
\end{equation}
\end{lemma}

\begin{proof}
The structure of this proof is basically the same as our analysis in Section~\ref{subsec:secur-proof} for DI-VKECM, except we invoke Corollary~\ref{cor:PA-leak} to ``amplify'' the bounds on the guessing probabilities, and also we make some modifications to fit the pirating-attack scenario (analogous to the modifications described in the Lemma~\ref{lemma:DI-SDECM} proof).

Consider a pirating attack as described in Definition~\ref{def:pir-sec-ind}.
Let the number of bits leaked during the key generation phase be $\nu_1l$, and the number of bits leaked afterwards be $\nu_2l$, with $\nu_1+\nu_2=\nu$.
In order to properly apply Corollary~\ref{cor:PA-leak} later, we suppose without loss of generality that all the $\nu_2l$ bits leaked in the second part occur after Bob and Charlie get their ciphertexts (as remarked in the Theorem~\ref{thm:unc_leak} proof, any leakage in the second part before they get their ciphertexts can be handled via the same form of argument as for the first $\nu_1l$ bits, or alternatively by simply absorbing such leakage into the dishonest receiver's actions before distributing the state between Bob and Charlie).
Let the ciphertext value that Bob receives during the encryption phase be denoted as $(P^\B, \syn(\eA^\B), \eX^\B, \rVec^\B)$ where $P^\B = M_1 \oplus \otp \oplus \eAB\cdot\rVec^\B$; analogously, let Charlie's values be denoted as $(P^\C, \syn(\eA^\C), \eX^\C, \rVec^\C)$ where $P^\C = M_2 \oplus \otp \oplus \eAC\cdot\rVec^\C$.

To begin, let us first prove a rough analogue of Lemma~\ref{lem:rawkeyguess}; specifically, consider a scenario that is defined the same way as a pirating attack on this protocol, except with the following modifications:
\begin{itemize}
\item Bob and Charlie's goal is instead to guess their respective ``raw key'' strings $\eAB$ and $\eAC$ (rather than their messages).
\item During the encryption phase, we {only} give them the values $(\syn(\eA^\B), \eX^\B)$ and $(\syn(\eA^\C), \eX^\C)$ in their respective ciphertexts (i.e.~$(P^\B,\rVec^\B)$ and $(P^\C,\rVec^\C)$ are omitted), and also require them to produce their guesses ``immediately'' after receiving these values, \emph{without} performing the subsequent $\nu_2l$ bits of leakage.
\end{itemize}
We aim to upper bound the probability that their guesses are correct in this scenario.
To do so, we note that for the special case $\nu_1=0$ (i.e.~no bits are leaked up until their guesses for $\eAB,\eAC$ are produced), by the same arguments as the Lemma~\ref{lem:rawkeyguess} proof, the probability that $F=\text{\cmark}$ and they both guess correctly is at most $2^{-\kappa\dga^3\alpha^4l+2\xi (1-\gamma)(1-\alpha)h_2(\qhon)l}$, substituting in the formula~\eqref{eq:par-cons-2} for the syndrome length $\lsyn$. (Here we have invoked the fact that in this scenario Bob and Charlie do not get $P^\B,\rVec^\B$ and $P^\C,\rVec^\C$, in which case the register $\otp$ they received in the key generation phase is something they could have generated on their own using shared randomness, so the only ``useful'' classical information they receive to produce their guesses are $(\syn(\eA^\B), \eX^\B)$ and $(\syn(\eA^\C), \eX^\C)$.) Now, to handle the general case where $\nu_1>0$, we can again use the same arguments as in the Theorem~\ref{thm:unc_leak} proof to conclude that this probability is instead upper bounded by $2^{-\kappa\dga^3\alpha^4l+2\xi (1-\gamma)(1-\alpha)h_2(\qhon)l + \nu_1 l}$.
This means we have that the probability they both guess correctly \emph{conditioned} on $F=\text{\cmark}$ satisfies 
\begin{align}
\Pr[(\eAB = \gB) \land (\eAC = \gC) | F=\text{\cmark}] &= \frac{1}{\Pr[F=\text{\cmark}]}\Pr[(F=\text{\cmark}) \land (\eAB = \gB) \land (\eAC = \gC)] \nonumber\\
&\leq \frac{2^{-\kappa\dga^3\alpha^4l+2\xi (1-\gamma)(1-\alpha)h_2(\qhon)l + \nu_1 l}}{\Pr[F=\text{\cmark}]}. \label{eq:rawkeyguesscond}
\end{align}

We now use the above bound to study another modified pirating attack scenario, where now the only modification from an actual pirating attack is that Bob and Charlie try to respectively produce guesses for $\eAB\cdot\rVec^\B$ and $\eAC\cdot\rVec^\C$ instead of their messages (i.e.~here we allow them to have access to their full ciphertexts and leakage bits). First note that  because the messages $M_1,M_2$ were initially chosen uniformly and independently of everything else (by Definition~\ref{def:pir-sec-ind} for pirating attacks), and play no further role in this modified scenario after being used to produce $P^\B$ and $P^\C$, we can invoke Fact~\ref{fc:otp} to say that $P^\B$ could have been locally generated by Bob, and analogously $P^\C$ by Charlie, which means we can omit them from consideration. (This is the same argument as in the Lemma~\ref{lemma:DI-SDECM} proof with respect to Definition~\ref{def:pir-sec-ind}.) With this, we are precisely in a scenario where Lemma~\ref{lem:bittrit-PA} applies, identifying $X^\B,X^\C$ in the lemma statement with $\eAB,\eAC$ respectively in this scenario: the above bound~\eqref{eq:rawkeyguesscond} upper-bounds the probability that (for the state conditioned on $F=\text{\cmark}$) Bob and Charlie can guess the ``raw key'' values $\eAB,\eAC$, and we are interested in bounding the probability that they can guess $\eAB\cdot\rVec^\B$ and $\eAC\cdot\rVec^\C$ after receiving $\rVec^\B,\rVec^\C$ and having $\nu_2l$ bits of leakage. Explicitly, Lemma~\ref{lem:bittrit-PA} tells us that this probability (for the conditional state) is at most
\[
\frac{1}{p} + \frac{2^{\nu_2 \lambda}}{\Pr[F=\text{\cmark}]^{\frac{1}{3}}}O\left(2^{-\frac{1}{3}(\kappa\dga^3\alpha^4 - 2\xi (1-\gamma)(1-\alpha)h_2(\qhon) + \nu_1) \lambda}\right),
\]
substituting $l=\lambda$.

Finally, we turn to the actual pirating attack scenario. We make the simple observation that while Bob's goal here is to guess his message $M_1$, the fact that he has the values $\otp$ and $P^\B = M_1 \oplus \otp \oplus \eAB\cdot\rVec^\B$ means that this is equivalent to producing a guess for $\eAB\cdot\rVec^\B$; an analogous statement holds for Charlie. 
Hence the above bound also holds for the probability that they can both guess their messages (conditioned on $F=\text{\cmark}$), which lets us write
\begin{align*}
& \Pr[(F=\text{\cmark}) \land (M_1=M^\B)\land(M_2=M^\C)] \\
= & \Pr[F=\text{\cmark}] \Pr[(M_1=M^\B)\land(M_2=M^\C) | F=\text{\cmark}] \\
\leq & \Pr[F=\text{\cmark}] \left(\frac{1}{p} + \frac{2^{\nu_2 \lambda}}{\Pr[F=\text{\cmark}]^{\frac{1}{3}}}O\left(2^{-\frac{1}{3}(\kappa\dga^3\alpha^4 - 2\xi (1-\gamma)(1-\alpha)h_2(\qhon) + \nu_1) \lambda}\right)\right) \\
\leq & \frac{1}{p} + O\left(2^{-\frac{1}{3}(\kappa\dga^3\alpha^4 - 2\xi(1-\gamma)(1-\alpha)h_2(\qhon) - 3\nu) \lambda}\right).
\end{align*}
Given the condition~\eqref{eq:0pirmaxleak}, the $O\left(2^{-\frac{1}{3}(\kappa\dga^3\alpha^4 - 2\xi(1-\gamma)(1-\alpha)h_2(\qhon) - 3\nu) \lambda}\right)$ term is a negligible function of $\lambda$, yielding the desired result.
\end{proof}

\begin{remark}\label{remark:OTPreuse}
Unlike the previous section, we were not able to prove here that this protocol satisfies $(0,\negl(\lambda))$-anti-piracy security in the sense of Definition~\ref{def:pir-sec}, due to a subtle issue in handling the information available to Bob and Charlie. Specifically, under that definition, the values $P^\B$ and $P^\C$ in the above analysis would instead be $P^\B = M \oplus \otp \oplus \eAB\cdot\rVec^\B$ and $P^\C = M \oplus \otp \oplus \eAB\cdot\rVec^\C$ where the message $M$ is the same in both terms. Due to this, we no longer have the property that $P^\B P^\C$ can be locally generated by Bob and Charlie without access to the client's registers, and the above argument no longer straightforwardly works.\footnote{Qualitatively, one way to understand this issue is that in our above proof, we have basically treated the messages $M_1$ and $M_2$ as independent one-time-pads on the quantities $\otp \oplus \eAB\cdot\rVec^\B$ and $\otp \oplus \eAB\cdot\rVec^\C$, but attempting an analogous argument under Definition~\ref{def:pir-sec} with $M$ as the ``one-time-pad'' runs into the issue that it is reused across the two terms, which disrupts the usual security properties of one-time-padding.} A similar obstacle is encountered when trying to obtain a DI-VKECM protocol with $(0,\negl(\lambda))$-uncloneability via this approach.\footnote{In the context of the previous protocols without the ``extraction'' step, our security proofs overcame this issue via the bound~\eqref{eq:pguessDRDR}, in which we essentially exploited the fact that the dishonest parties' goal in the security definition is a somewhat ``stricter'' win condition than the parallel-repeated game we actually bound the winning probability of. This approach does not seem to straightforwardly work after the extraction step has taken place.}

Still, we note that it does seem possible that the protocol may in fact satisfy Definition~\ref{def:pir-sec}; it is just that the above proof technique does not suffice to show this. If a proof of this were to be found, it seems likely that it would also yield an uncloneable encryption protocol with $(0,\negl(\lambda))$-uncloneability for single bits or trits. (A property that may be convenient in a prospective proof is the fact that e.g.~for the bit-valued case (i.e.~$\clM=\bbF_2$), the values $P^\B P^\C$ can be equivalently viewed as a pair of bits uniformly distributed across the two values satisfying $P^\B \oplus P^\C = \eAB\cdot\rVec^\B \oplus \eAB\cdot\rVec^\C$. Put another way, in some sense they encode the XOR of the ``secret'' values $\eAB\cdot\rVec^\B$ and $\eAB\cdot\rVec^\C$, but in a ``distributed'' way across Bob and Charlie rather than being available to either of them locally. If it could be shown that supplying Bob and Charlie with these values does not increase their joint probability of guessing the message $M$, then it would be sufficient to obtain the desired result.)
\end{remark}

\section{Parallel repetition of the cloning game}\label{sec:parrep}
In this section, we prove the following theorem.
\begin{theorem}\label{thm:parrep}
Let $G_\alpha$ be a game as described in Section \ref{subsec:game-def}, satisfying properties \ref{prop:dist-1}-\ref{prop:dist-2}, and with $\omega^*(G_\alpha)=1-\eps$. Let $\clA, \clS, \clB, \clC$ be the output sets of $G_\alpha$. Then for $t = (1-\eps + \eta)l$, the parallel-repeated $G_\alpha$ satisfies
\begin{align*}
\omega^*(G^l_\alpha) & = \left(1-\frac{\eps}{2}\right)^{\Omega\left(\frac{\eps^2\alpha^4l}{\log(|\clA|\cdot|\clS|\cdot|\clB|\cdot|\clC|)}\right)} \\
\omega^*(G^{t/l}_\alpha) & = \left(1-\frac{\eta}{2}\right)^{\Omega\left(\frac{\eta^2\alpha^4l}{\log(|\clA|\cdot|\clS|\cdot|\clB|\cdot|\clC|)}\right)}.
\end{align*}
\end{theorem}
\noindent Theorem~\ref{thm:cl-parrep} is a simplified version of the second bound after applying the inequality $1-\kappa \leq 2^{-\kappa}$.

We shall use the following results in our proof.
\begin{fact}[\cite{Hol09}]\label{fc:hol-cond}
Let $\sfP_{QM_1\ldots M_lN} = \sfP_Q\sfP_{M_1|Q}\sfP_{M_2|Q}\ldots\sfP_{M_l|Q}\sfP_{N|QM_1\ldots M_l}$ be a probability distribution over $\mathcal{Q}\times\mathcal{M}^l\times\mathcal{N}$, and let $\clE$ be any event. Then,
\[ \sum_{i=1}^l\Vert\sfP_{QM_iN|\clE}-\sfP_{QN|\clE}\sfP_{M_i|Q}\Vert_1 \leq \sqrt{l\left(\log(|\mathcal{N}|) + \log\left(\frac{1}{\Pr[\clE]}\right)\right)}.\]
\end{fact}
\begin{fact}[\cite{JPY14}, Lemma III.1]\label{fc:jpy-cond}
Suppose $\rho$ and $\sigma$ are CQ states satisfying $\rho = \delta\sigma + (1-\delta)\sigma'$ for some other state $\sigma'$. Suppose $Z$ is a classical register of size $|\clZ|$ in $\rho$ and $\sigma$ such that the distribution on $Z$ in $\sigma$ is $\sfP_Z$, then
\[ \bbE_{\sfP_Z}\sfD(\sigma_z\Vert\rho) \leq \log(1/\delta) + \log|\clZ|.\]
\end{fact}
\begin{fact}[Quantum Raz's Lemma, \cite{BVY17}]\label{fc:qraz}
Let $\rho_{XY}$ and $\sigma_{XY}$ be two CQ states with $X = X_1\ldots X_l$ being classical, and $\sigma$ being product across all registers. Then,
\[ \sum_{i=1}^l\sfI(X_i:Y)_\rho \leq \sfD(\rho_{XY}\Vert \sigma_{XY}).\]
\end{fact}
\begin{fact}[\cite{KT23}, Lemma 32]\label{fc:anchor-t*}
Suppose $\sfP_{ST}$ and $\sfP_{S'T'R'}$ are distributions such that 
for some $t^*$, we have $\sfP_{ST}(s,t^*) = \beta\cdot\sfP_S(s)$ for all $s$. If $\norm{\sfP_{ST}-\sfP_{S'T'}}_1 \leq \beta$, then,
\begin{enumerate}[(i)]
\item $ \Vert\sfP_{S'R'|t^*} - \sfP_{S'R'}\Vert_1 \leq \dfrac{2}{\beta}\Vert\sfP_{S'T'R'} - \sfP_{S'R'}\sfP_{T|S}\Vert_1 + \dfrac{5}{\beta}\norm{\sfP_{S'T'} - \sfP_{ST}}_1$;
\vspace{-1cm}
\item $\displaystyle \begin{aligned}
 & \\[0.4cm]
\norm{\sfP_{S'T'R'} - \sfP_{ST}\sfP_{R'|t^*}}_1 & \leq \frac{2}{\beta}\Big(\norm{\sfP_{S'T'R'} - \sfP_{T'R'}\sfP_{S|T}}_1 + \norm{\sfP_{S'T'R'} - \sfP_{S'R'}\sfP_{T|S}}_1\Big)  \\
& \quad + \frac{7}{\beta}\norm{\sfP_{S'T'} - \sfP_{ST}}_1.
\end{aligned}$
\end{enumerate}
\end{fact}

\subsection{Setup}
Consider a protocol $\clP$ for $l$ copies of $G_\alpha$. Alice and Barlie have inputs $X=X_1\ldots X_l$ and $U=U_1\ldots U_l$ in the first round. Before the game starts, they share some entangled state; we shall assume that after they receive their outputs, they do some unitaries and then measure in the computational basis to produce their outputs $A=A_1\ldots A_l$ and $S=S_1\ldots S_l$. Suppose the state held by Alice and Barlie after they produce the outputs is $\ket{\sigma}$ on registers $A\tA S\tS E^\A E^{\B\C}$, with Alice holding $A\tA E^\A$ and Barlie holding the rest.\footnote{Barlie can do an additional unitary on his registers before passing them on to Bob and Charlie --- here we are absorbing that unitary, which does not change the distribution of $AS$, into $\ket{\sigma}$.} Here $A=A_1\ldots A_l$ and $S=S_1\ldots S_l$ will be the registers in which Alice and Barlie's outputs are measured in the computational basis; $\tA, \tS$ are registers onto which the contents of $A, S$ are copied --- we can always assume the outputs are copied since they are classical. We define the following pure state to represent the inputs, outputs and other registers in the protocol at this stage:
\begin{align*}
\ket{\rho}_{X\tX U\tU Y\tY Z\tZ A\tA S\tS E^\A E^{\B\C}} = & \sum_{x,u,y,z}\sqrt{\sfP_{XUYZ}(x,u,y,z)}\ket{xx}_{X\tX}\ket{uu}_{U\tU}\ket{yy}_{Y\tY}\ket{zz}_{Z\tZ} \otimes \\
 & \quad \sum_{a,s}\sqrt{\sfP_{AS|xu}(as)}\ket{aa}_{A\tA}\ket{ss}_{S\tS}\ket{\rho}_{E^\A E^{\B\C}|xuas}.
\end{align*}
Here we have included the $Y\tY Z\tZ$ registers in this state even though they have not been revealed yet or used in the protocol; the state in the entangled registers has no dependence on $z$ because of this. Here $\sfP_{AS|xu}(a,s)$ is the probability of Alice and Barlie obtaining outputs $(a,s)$ on inputs $(x,u)$ in the first round. In the actual protocol, the registers $A$ and $S$ would be measured at this stage, and the outputs can be used as classical inputs for the next round, but for the sake of this analysis we shall keep everything coherent.

In the second round, the registers $S\tS E^{\B\C}$ are distributed in some way between Bob and Charlie. They will then receive inputs $Y, Z$ respectively, and additionally have access to $U$. Now Bob and Charlie again do some unitaries and produce their outputs by measuring in the computational basis. Suppose their shared state when they produce their outputs is $\ket{\sigma}$ on registers $BCE^\B E^\C$, with Bob holding $BE^\B$ and Charlie holding the rest. For convenience, we shall also divide up the classical information that they both have copies of in the following way: Bob gets $US$, and Charlie gets $\tU \tS$. We denote the state of the protocol at this stage by
\begin{align*}
\ket{\sigma}_{X\tX U\tU Y\tY Z\tZ A\tA S\tS B C E^\A E^\B E^\C} = & \sum_{x,u,y,z}\sqrt{\sfP_{XUYZ}(x,u,y,z)}\ket{xx}_{X\tX}\ket{uu}_{U\tU}\ket{yy}_{Y\tY}\ket{zz}_{Z\tZ} \otimes \\
 & \quad \sum_{a,s}\sqrt{\sfP_{AS|xu}(as)}\ket{aa}_{A\tA}\ket{ss}_{S\tS}\sum_{bc}\sqrt{\sfP_{BC|xuyzas}(bc)}\ket{b}_{B}\ket{c}_{C} \otimes \\
 & \quad \ket{\sigma}_{E^\A E^\B E^\C|xuyzasbc}.
\end{align*}
Note that to get $\ket{\rho}$ to $\ket{\sigma}$, the $A\tA$ registers are not touched at all, and $S\tS$ are used only as a control registers for Bob and Charlie's unitaries, which is why the marginal distribution of $AS$ is the same in $\ket{\rho}$ and $\ket{\sigma}$.


We shall use the following lemma, whose proof is given later, to prove Theorem \ref{thm:parrep}.
\begin{lemma}\label{lem:parrep-ind}
Let $\omega^*(G) = 1 - \eps$. For $i\in[l]$, let $J_i=\sfV_1(X_iU_i,A_iS_i)\cdot\sfV_2(X_iU_iY_iZ_i,A_iS_iB_iC_i)$ in a protocol $\clP$ for $l$ copies of $G$ (here $\sfV_1$ and $\sfV_2$ are the first and second round predicates respectively). If $\clE_T$ is the event $\prod_{i\in T}J_i=1$ and $\oT$ denotes $[l]\setminus T$ for $T\subseteq [l]$, then
\[ \bbE_{i\in\oT}\Pr[J_i=1|\clE_T] \leq 1-\eps + O\left(\frac{\sqrt{\delta_T}}{\alpha^2}\right), \]
where
\[ \delta_T = \frac{|T|\cdot\log(|\clA|\cdot|\clS|\cdot|\clB|\cdot|\clC|) + \log(1/\Pr[\clE_T])}{l}.\]
\end{lemma}
It is possible to give an explicit constant in place of the big $O$ in the above lemma statement, by tracking all the constants in the proof. We shall not be doing so here, but it can be seen from our proof that the constant we get is certainly bigger than 4. Therefore, it is sufficient to prove the lemma when $\delta_T \leq \alpha^4/8$, as the lemma statement is trivial otherwise.

To prove the bound on $\omega^*(G^l)$ in Theorem \ref{thm:parrep}, we shall use the above lemma to choose a random subset $T$ of $[l]$, such that the probability of winning in the random subset is
\[ \left(1-\frac{\eps}{2}\right)^{\left(\frac{\Omega(\eps^2\alpha^4)}{\log(|\clA|\cdot|\clS|\cdot|\clB|\cdot|\clC|)}\right)l},\]
which immediately implies that the same bound holds for $\omega^*(G^l_\alpha)$. We start with $T=\emptyset$, and construct $T$ by choosing a uniformly random element outside the current $T$, until the final set satisfies $\delta_T \geq \eps^2\alpha^4/4K^2$, where $K$ is the constant in the big $O$ of Lemma \ref{lem:parrep-ind}. Every instance added by this procedure, except the very last one, satisfies the bound in Lemma \ref{lem:parrep-ind} with $K\cdot\frac{\sqrt{\delta_T}}{\alpha^2} \leq \frac{\eps}{2}$. If the final picked set has $\log(1/\Pr[\clE_T]) \geq \left(\frac{\eps^2\alpha^4}{4K^2}\right)\cdot\frac{l}{2\log(|\clA|\cdot|\clS|\cdot|\clB|\cdot|\clC|)}$, then we are already done. Otherwise we have,
\[ |T| \geq \left(\frac{\eps^2\alpha^4}{4K^2} - \log(1/\Pr[\clE_T])\right)\cdot\frac{l}{\log(|\clA|\cdot|\clS|\cdot|\clB|\cdot|\clC|)},\]
which makes $\Pr[\clE_T] \leq \left(1-\frac{\eps}{2}\right)^{|T|-1} \leq \left(1-\frac{\eps}{2}\right)^{\left(\frac{\eps^2\alpha^4}{8K^2}\right)\cdot\frac{l}{\log(|\clA|\cdot|\clS|\cdot|\clB|\cdot|\clC|)}}$, which is the required bound.

Next, to prove the bound on $\omega^*(G^{t/l})$, where $t=(1-\eps+\eta)l$, for some $\gamma$ to be determined later, suppose $T$ is such that $\Pr[\clE_T] \geq 2^{-\gamma^2l - |T|\cdot\log(|\clA|\cdot|\clS|\cdot|\clB|\cdot|\clC|)}$. Then we have, $\delta_T \leq \gamma^2$. Let
\[ n = \frac{\gamma^2l}{\log\left(\frac{\log(|\clA|\cdot|\clS|\cdot|\clB|\cdot|\clC|)}{1-\eps + \frac{K\gamma}{\alpha^2}}\right)},\]
which satisfies $2^{-\gamma^2l - n\cdot\log(|\clA|\cdot|\clS|\cdot|\clB|\cdot|\clC|)} = \left(1-\eps+\frac{K\gamma}{\alpha^2}\right)^n$. For a random set of size $n$, we can then say by the previous inductive argument that its winning probability is upper bounded by $\left(1-\eps+\frac{K\gamma}{\alpha^2}\right)^n$, i.e.,
\begin{equation}\label{eq:thres-bound}
\sum_{T\subseteq [l]: |T|=n}\frac{1}{{l\choose n}}\Pr[\clE_T] \leq \left(1-\eps+\frac{K\gamma}{\alpha^2}\right)^n.
\end{equation}
Now set $\gamma=\frac{\alpha^2\eta}{4K}$, which makes $n < (1-\eps+\eta)l$. If $(1-\eps+\eta)l$ games are won, we can pick a random subset of $n$ instances out of the $(1-\eps+\eta)l$ instances on which the game is won, and say that the probability of winning the $(1-\eps+\eta)l$ instances is upper bounded by the probability of winning these $n$ instances. Using \eqref{eq:thres-bound} we therefore have,
\begin{equation}\label{eq:thres-bound2}
\omega^*(G^{t/l}) \leq \sum_{T\subseteq [(1-\eps+\eta)l]:|T|=n}\frac{1}{{(1-\eps+\eta)l \choose n}}\Pr[\clE_T] \leq \left(1-\eps+\frac{\eta}{4}\right)^n\cdot\frac{{l \choose n}}{{(1-\eps +\eta)l \choose n}}.
\end{equation}
We can simplify the second factor in the above expression as
\[ \frac{{l \choose n}}{{(1-\eps +\eta)l \choose n}} \leq \left(\frac{l}{(1-\eps+\eta)l-n}\right)^n \leq \left(\frac{1}{1-\eps+\frac{3\eta}{4}}\right)^n,\]
where in the last inequality we have used the definition of $n$. Putting this into \eqref{eq:thres-bound2} we get,
\[ \omega^*(G^{t/l}) \leq \left(\frac{1-\eps+\frac{\eta}{4}}{1-\eps+\frac{3\eta}{4}}\right)^n = \left(1-\frac{\frac{\eta}{2}}{1-\eps+\frac{3\eta}{4}}\right)^n \leq \left(1-\frac{\eta}{2}\right)^n,\]
which proves the theorem after substituting the value of $n$.

\subsection{Proof of Lemma \ref{lem:parrep-ind}}
From this section onwards, for the sake of brevity, we shall drop the subscript $T$ from $\delta_T$ and $\clE_T$ as defined in Lemma \ref{lem:parrep-ind}, and refer to these quantities as simply $\delta$ and $\clE$. 

For each $i\in [l]$, we define the random variable $D_iF_i$ as described in Section \ref{subsec:game-def}: for each $i$ $D_i$ is a uniformly random bit; $F_i=X_iU_i$ if $D_i=0$ and $F_i=Y_iZ_i$ if $D_i=1$. We can consider the states $\ket{\rho}$ and $\ket{\sigma}$ conditioned on particular values of $DF=d_1\ldots d_lf_1\ldots f_l=df$, and this simply means that the distributions of $XUYZ$ are conditioned on $DF=df$.

Conditioned on $DF=df$, we define the state $\ket{\vph}_{df}$, which is $\ket{\sigma}_{df}$ conditioned on success in $T$:
\begin{align*}
\ket{\vph}_{X\tX U\tU Y\tY Z\tZ A\tA S\tS B C E^\A E^\B E^\C|df} & = \frac{1}{\sqrt{\gamma_{df}}}\sum_{x,u,y,z}\sqrt{\sfP_{XUYZ|df}(x,u,y,z)}\ket{xx}_{X\tX}\ket{uu}_{U\tU}\ket{yy}_{Y\tY}\ket{zz}_{Z\tZ} \otimes \\
& \quad \sum_{\substack{a,s,b,c: \\ (x_T,u_T,y_T,z_T,a_T,s_T,b_T,c_T) \\ \text{win } G^{|T|}}}\sqrt{\sfP_{ASBC|xuyz}(asbc)}\ket{aa}_{A\tA}\ket{ss}_{S\tS}\ket{b}_{B}\ket{c}_{C}\otimes \\
& \quad \ket{\sigma}_{E^\A E^\B E^\C|xuyzasbc}
\end{align*}
where $\gamma_{df}$ is the probability of winning in $T$ conditioned on $DF=df$, in $\clP$. It is easy to see that $\sfP_{XUYZASBC|\clE,df}$ is the distribution on the registers $XUYZASBC$ in $\ket{\vph}_{df}$. For $i\in [l]$, we shall use $\ket{\vph}_{x_iu_iy_iz_id_{-i}f_{-i}}$ (where $d_{-i}$ stands for $d_1\ldots d_{i-1}d_{i+1}\ldots d_l$ and similar notation is used for $f$), $\ket{\vph}_{u_iy_iz_id_{-i}f_{-i}}$, $\ket{\vph}_{x_iy_iz_id_{-i}f_{-i}}$, $\ket{\vph}_{y_iz_id_{-i}f_{-i}}$ to refer to $\ket{\vph}$ with values of $X_iU_iY_iZ_iD_{-i}F_{-i}$, $U_iY_iZ_iD_{-i}F_{-i}$, $X_iY_iZ_iD_{-i}F_{-i}$ and $Y_iZ_iD_{-i}F_{-i}$ respectively conditioned on, and similar notation for other variables and subsets of $[l]$ as well. $\ket{\vph}_{\bot,\bot,d_{-i}f_{-i}}$ will be used to refer to a state where at an index implied from context (in this case $i$), $Y_iZ_i$ are both conditioned on the value $\bot$; when only one of $Y_i$ or $Z_i$ are conditioned on, we shall write the state explicitly as $\ket{\vph}_{Y_i=\bot,d_{-i}f_{-i}}$ or $\ket{\vph}_{Z_i=\bot,d_{-i}f_{-i}}$. 

We shall use the following lemma, whose proof we give later, to prove Lemma \ref{lem:parrep-ind}.
\begin{lemma}\label{lem:parrep-conds}
If $\delta \leq \alpha^4/8$ for $\delta$ as defined in Lemma \ref{lem:parrep-ind} (called $\delta_T$ in the lemma statement), then using $R_i = X_TU_TY_TZ_TA_TS_TB_TC_T D_{-i}F_{-i}$, the following conditions hold:
\begin{enumerate}[(i)]
\item $\bbE_{i\in\oT}\norm{\sfP_{X_iU_iY_iZ_iR_i|\clE} - \sfP_{X_iU_iY_iZ_i}\sfP_{R_i|\clE,\perp\perp}}_1 \leq O\left(\dfrac{\sqrt{\delta}}{\alpha^2}\right)$;
\item For each $i\in \oT$, there exist unitaries $\{V^\A_{i,x_ir_i}\}_{x_ir_i}$ acting on $X_{\oT}\tX_{\oT}E^\A A_{\oT}\tA_{\oT}$, and $\{V^{\B\C}_{i,y_ir_i}\}_{y_ir_i}$ acting on $U_{\oT}\tU_{\oT}Y_{\oT}\tY_{\oT}Z_{\oT}\tZ_{\oT}S_{\oT}\tS_{\oT}B_{\oT}C_{\oT}E^\B E^\C$ such that
\[ \bbE_{i\in\oT}\bbE_{\sfP_{X_iU_iR_i|\clE}}\norm{V^\A_{i,x_ir_i}\otimes V^{\B\C}_{i,u_ir_i}\state{\vph}_{\perp \perp r_i}(V^\A_{i,x_ir_i})^\dagger\otimes(V^{\B \C}_{i,u_ir_i})^\dagger - \state{\vph}_{x_iu_i\perp\perp r_i}}_1 \leq O\left(\dfrac{\sqrt{\delta}}{\alpha^2}\right); \]
\item $\bbE_{i\in\oT}\norm{\sfP_{X_iU_iY_iZ_iR_i|\clE}\left(\sfP_{A_iS_i|\clE,X_iU_iY_iZ_iR_i} - \sfP_{A_iS_i|\clE,X_iU_i,\perp \perp,R_i}\right)}_1 \leq O\left(\dfrac{\sqrt{\delta}}{\alpha^2}\right)$;
\item For each $i\in\oT$, there exist unitaries $\{V^\B_{i,u_iy_ir_i}\}_{u_iy_ir_i}$ acting on the registers $Y_{\oT}\tY_{\oT}U_{\oT}S_{\oT}E^\B B_{\oT}$ and $\{V^\C_{i,u_iz_ir_i}\}_{u_iz_ir_i}$ acting on $Z_{\oT}\tZ_{\oT}\tU_{\oT}\tS_{\oT}E^\C C_{\oT}$ such that
\begin{align*}
& \bbE_{i\in\oT}\bbE_{\sfP_{X_iU_iY_iZ_iA_iS_iR_i|\clE}}\norm{V^\B_{i,u_iy_ir_i}\otimes V^\C_{i,u_iz_ir_i}\state{\vph}_{x_iu_i\perp\perp a_is_ir_i}(V^\B_{i,u_iy_ir_i})^\dagger\otimes (V^\C_{i,u_iz_ir_i})^\dagger - \state{\vph}_{x_iuy_iz_ia_is_ir_i}}_1 \\
& \leq O\left(\frac{\sqrt{\delta}}{\alpha^2}\right).
\end{align*}
\end{enumerate}
\end{lemma}

As noted before, the statement of Lemma \ref{lem:parrep-ind} is trivial if $\delta \geq \alpha^4/8$, so we shall use Lemma \ref{lem:parrep-conds} to prove Lemma \ref{lem:parrep-ind} in the case that $\delta \leq \alpha^4/8$. Making use of the conditions in Lemma \ref{lem:parrep-conds}, we give a strategy for a single copy of $G_\alpha$ as follows:
\begin{itemize}
\item Alice and Barlie share $\log|\oT|$ uniform bits, for each $i\in\oT$, $\sfP_{R_i|\clE,\perp \perp}$ as randomness, and for each $R_i=r_i$, the state $\ket{\vph}_{\perp\perp r_i}$ as entanglement, with Alice holding registers $X_{\oT}\tX_{\oT} A_{\oT}\tA_{\oT}E^\A$ and Barlie holding registers $U_{\oT}\tU_{\oT}Y_{\oT}\tY_{\oT}Z_{\oT}\tZ_{\oT}S_{\oT}\tS_{\oT}B_{\oT}C_{\oT}E^\B E^\C$ (the rest of the registers have fixed values due to $r_i$).
\item Alice and Barlie use their shared randomness to sample $i\in\oT$ uniformly, and in the first round, apply $V^\A_{i,x_ir_i}$, $V^{\B\C}_{i,u_ir_i}$ on their parts of the shared entangled state according to their shared randomness from $\sfP_{R_i|\clE,\perp\perp }$ and their first round inputs.
\item Alice and Bob measure the $A_i,S_i$ registers of the resulting state to give their first round outputs.
\item Before the second round starts, Barlie passes the registers $Y_{\oT}\tY_{\oT}U_{\oT}S_{\oT}BE^\B$ to Bob, and $Z_{\oT}\tZ_{\oT}\tU_{\oT}\tS_{\oT}CE^\C$ to Charlie. He also gives each of them his first round input $u_i$, and the $i, r_i$ he sampled in the first round.
\item After receiving the second round inputs, Bob and Charlie apply $V^\B_{i,u_iy_ir_i}$ and $V^\C_{i,u_iz_ir_i}$ respectively to their registers according to their inputs and $u_ir_i$ received from Barlie.
\item Bob and Charlie measure the $B_i, C_i$ registers of the resulting state to give their second round outputs.
\end{itemize}

We shall first analyse the success probability of this strategy assuming the inputs and shared randomness are distributed according to $\sfP_{X_iY_iZ_iR_i|\clE}$ for $i$. Let $\sfP_{\hA_i\hS_i|X_iU_iY_iZ_iR_i}$ denote the conditional distribution Alice and Barlie get after the first round (note that $\hA_i\hS_i$ are actually independent of $Y_iZ_i$ given $X_iU_i$, but we are still writing $Y_iZ_i$ in the conditioning), and $\sfP_{\hB_i\hC_i|X_iU_iY_iZ_iR_i\hA_i\hS_i}$ denote their conditional distribution after the second round. Since $\sfP_{\hA_i\hS_i|X_iY_iZ_iR_i}$ is obtained by measuring the $A_iS_i$ registers of the state $V^\A_{i,x_ir_i}\otimes V^{\B\C}_{i,u_ir_i}\ket{\vph}_{\perp\perp r_i}$, and $\sfP_{A_iS_i|\clE,X_iU_i,\perp\perp,R_i}$ is obtained by measuring the same registers of $\ket{\vph}_{x_iu_i\perp\perp r_i}$, from item (ii) of Lemma~\ref{lem:parrep-conds} and Fact \ref{fc:chan-l1} we have,
\[ \bbE_{i\in\oT}\norm{\sfP_{X_iU_iY_iZ_iR_i|\clE}\left(\sfP_{\hA_i\hS_i|X_iU_iY_iZ_iR_i} - \sfP_{A_iS_i|\clE,X_iU_i,\perp\perp,R_i}\right)}_1 \leq O\left(\frac{\sqrt{\delta}}{\alpha^2}\right). \]
Combining this with item (iii) of the lemma we have,
\[ \bbE_{i\in\oT}\norm{\sfP_{X_iU_iY_iZ_iR_i|\clE}\left(\sfP_{\hA_i\hS_i|X_iU_iY_iZ_iR_i} - \sfP_{A_iS_i|\clE,X_iU_iY_iZ_iR_i}\right)}_1 \leq O\left(\frac{\sqrt{\delta}}{\alpha^2}\right).\]
By similar reasoning, we have from item (iv),
\[ \bbE_{i\in\oT}\norm{\sfP_{X_iU_iY_iZ_iR_iA_iS_i}\left(\sfP_{\hB_i\hC_i|X_iU_iY_iZ_iR_i\hA_i\hS_i} - \sfP_{B_iC_i|\clE,X_iY_iZ_iR_iA_iS_i}\right)}_1 \leq O\left(\frac{\sqrt{\delta}}{\alpha^2}\right).\]
Now, since our actual distribution of inputs and randomness if $\sfP_{X_iU_iY_iZ_i}\sfP_{R_i|\clE,\perp\perp}$, our actual input, randomness and output distribution is $\sfP_{X_iU_iY_iZ_i}\sfP_{R_i|\clE,\perp\perp }\sfP_{\hA_i\hS_i\hB_i\hC_i|X_iY_iZ_iR_i}$. Combing the above equations with item (i), we get that this distribution satisfies,
\[ \bbE_{i\in\oT}\norm{\sfP_{X_iU_iY_iZ_i}\sfP_{R_i|\clE,\perp\perp }\sfP_{\hA_i\hS_i\hB_i\hC_i|X_iU_iY_iZ_iR_i} - \sfP_{X_iU_iY_iZ_iR_iA_iS_iB_iC_i|\clE}}_1 \leq O\left(\frac{\sqrt{\delta}}{\alpha^2}\right). \]
Suppose the constant in the above big $O$ is $K$. Since $\Pr[J_i=1|\clE]$ is the probability that the distribution $\sfP_{X_iU_iY_iZ_iR_iA_iS_iB_iC_i|\clE}$ wins a single copy of the game $G_\alpha$, if $\bbE_{i\in\oT}\Pr[J_i=1|\clE] > 1-\eps+\frac{K}{2}\cdot\frac{\sqrt{\delta}}{\alpha^2}$, then the winning probability of our constructed strategy is more than
\[ 1-\eps + \frac{K}{2}\cdot\frac{\sqrt{\delta}}{\alpha^2} - \frac{1}{2}\bbE_{i\in\oT}\norm{\sfP_{X_iU_iY_iZ_i}\sfP_{R_i|\clE,\perp\perp }\sfP_{\hA_i\hS_i\hB_i\hC_i|X_iU_iY_iZ_iR_i} - \sfP_{X_iU_iY_iZ_iR_iA_iS_iB_iC_i|\clE}}_1 \geq \omega^*(G),\]
which is a contradiction. Therefore we must have $\bbE_{i\in\oT}\Pr[J_i=1|\clE] \leq 1-\eps+O\left(\frac{\sqrt{\delta}}{\alpha^2}\right)$.

\subsection{Proof of Lemma \ref{lem:parrep-conds}}
\textbf{Closeness of distributions.} Applying Fact \ref{fc:hol-cond} with $Q, N$ being trivial and $M_i = X_iU_iY_iZ_i$ we get,
\begin{equation}
\bbE_{i \in \oT }\Vert\sfP_{X_iU_iY_iZ_i|\clE} - \sfP_{X_iU_iY_iZ_i}\Vert_1 \leq \frac{1}{l-|T|}\sqrt{(l-|T|)\cdot \log(1/\Pr[\clE])} \leq \sqrt{2\delta}, \label{eq:XYZ}
\end{equation}
recalling we are taking $\delta$ to be the value $\delta_T$ defined in Lemma~\ref{lem:parrep-ind}. In particular, the last line of the above equation is obtained by recalling that we have required $\delta \leq \alpha^4/8$, which implies $|T| \leq l/2$. 

Also, applying Fact \ref{fc:hol-cond} again with $M_i$ the same, $Q=X_TU_TY_TZ_TDF$ and $N=A_TS_TB_TC_T$, we get
\begin{align}
\sqrt{2\delta} & \geq \frac{1}{l-|T|}\sqrt{(l-|T|)(\log(1/\Pr[\clE]) + |T|\cdot\log(|\clA|\cdot|\clS|\cdot|\clB|\cdot|\clC|)} \nonumber \\
& \geq \bbE_{i\in\oT}\norm{\sfP_{X_iU_iY_iZ_iX_TU_TY_TZ_TDFA_TS_TB_TC_T|\clE} - \sfP_{X_TU_TY_TZ_TA_TS_TB_TC_TDF|\clE}\sfP_{X_iU_iY_iZ_i|X_TU_TY_TZ_TDF}}_1 \nonumber \\
& = \bbE_{i\in\oT }\Vert\sfP_{X_iU_iY_iZ_iD_iF_iR_i|\clE} - \sfP_{D_iF_iR_i|\clE}\sfP_{X_iU_iY_iZ_i|D_iF_i}\Vert_1 \nonumber \\
 & = \frac{1}{2}\bbE_{i\in\oT }\left(\Vert\sfP_{X_iU_iY_iZ_iR_i|\clE} - \sfP_{X_iU_iR_i|\clE}\sfP_{Y_iZ_i|X_iU_i}\Vert_1 + \Vert\sfP_{X_iU_iY_iZ_iR_i|\clE} - \sfP_{Y_iZ_iR_i|\clE}\sfP_{X_iU_i|Y_iZ_i}\Vert_1\right), \label{eq:XYZR-1}
\end{align}
where in the third line we have used the definition of $R_i=X_TU_TY_TZ_TA_TS_TB_TC_TD_{-i}G_{-i}$, and the last line is obtained by conditioning on values $D_i=0$ and $D_i=1$ (which happen with probability $\frac{1}{2}$ even after conditioning on $\clE$).

Note that for all $x_i, u_i$, we have $\sfP_{X_iU_iY_iZ_i}(x_i,u_i,\perp,\perp) = \alpha^2\cdot\sfP_{X_iU_i}(x_iu_i)$. This and the bound~\eqref{eq:XYZ} allows us to apply item (ii) of Fact~\ref{fc:anchor-t*}, with $X_iU_i=S$, $Y_iZ_i=T$, $R_i=R$, and the corresponding variables conditioned on $\clE$ being the primed variables in the lemma statement. This gives
\begin{align*}
\bbE_{i\in\oT}\norm{\sfP_{X_iU_iY_iZ_iR_i|\clE} - \sfP_{X_iU_iY_iZ_i}\sfP_{R_i|\clE,\perp \perp}}_1  & \leq \frac{2}{\alpha^2}\bbE_{i\in\oT }\big(\Vert\sfP_{X_iU_iY_iZ_iR_i|\clE} - \sfP_{X_iU_iR_i|\clE}\sfP_{Y_iZ_i|X_iU_i}\Vert_1 \nonumber \\
& \quad + \Vert\sfP_{X_iU_iY_iZ_iR_i|\clE} - \sfP_{Y_iZ_iR_i|\clE}\sfP_{X_iU_i|Y_iZ_i}\Vert_1\big) \\
& \quad + \frac{7}{\alpha^2}\bbE_{i \in \oT }\Vert\sfP_{X_iU_iY_iZ_i|\clE} - \sfP_{X_iU_iY_iZ_i}\Vert_1.
\end{align*}
Applying \eqref{eq:XYZ} and \eqref{eq:XYZR-1} to the terms on the right-hand side yields item (i) of the lemma.

To show item (iii) of the lemma, we shall apply Fact \ref{fc:jpy-cond} to the states $\vph_{YZX\tX U\tU A_{\oT}\tA_{\oT} S_{\oT}\tS_{\oT}|a_Ts_Tb_Tc_Tdf}$ and $\sigma_{YZX\tX U\tU A_{\oT}\tA_{\oT} S_{\oT}\tS_{\oT}|df}$ (with $z$ being $a_Ts_Tb_Tc_T$) to get,
\begin{align*}
& \bbE_{\sfP_{Y_TZ_TX_TU_TA_TS_TC_TC_TDF|\clE}}\sfD\Big(\vph_{Y_{\oT }Z_{\oT}X_{\oT }\tX_{\oT}U_{\oT}\tU_{\oT} A_{\oT}\tA_{\oT}S_{\oT}\tS_{\oT}|x_Tu_Ty_Tz_Ta_Ts_Tb_Tc_Tdf}\Big\Vert \sigma_{Y_{\oT }Z_{\oT}X_{\oT }\tX_{\oT}U_{\oT}\tU_{\oT} A_{\oT}\tA_{\oT}S_{\oT}\tS_{\oT}|x_Tu_Ty_Tz_Tdf}\Big) \\
\leq & \bbE_{\sfP_{A_TS_TB_TC_TDF}}\sfD\left(\vph_{YZX\tX U\tU A_{\oT}\tA_{\oT} S_{\oT}\tS_{\oT}|a_Ts_Tb_Tc_Tdf}\middle\Vert\sigma_{YZX\tX U\tU A_{\oT}\tA_{\oT} S_{\oT}\tS_{\oT}|df}\right) \\
\leq & \bbE_{\sfP_{DF|\clE}}(\log(1/\gamma_{df}) + \log(|\clA|^{|T|}\cdot|\clS|^{|T|}\cdot|\clB|^{|T|}\cdot|\clC|^{|T|})) \\
\leq &  \log\Big(1/\bbE_{\sfP_{DF|\clE}}\gamma_{df}\Big) + |T|\cdot\log(|\clA|\cdot|\clS|\cdot|\clB|\cdot|\clC|) \\
= & \log(1/\Pr[\clE]) + |T|\cdot\log(|\clA|\cdot|\clS|\cdot|\clB|\cdot|\clC|) = \delta l.
\end{align*}
Now we notice that the state $\sigma_{Y_{\oT }Z_{\oT}X_{\oT }\tX_{\oT}U_{\oT}\tU_{\oT} A_{\oT}\tA_{\oT}S_{\oT}\tS_{\oT}|x_Tu_Ty_Tz_Tdf}$ is product across $Y_{\oT}Z_{\oT}$ and the rest of the registers. This is because conditioned on $df$, $YZ$ is independent of $XY$, and $YZ$ was not involved in the first round at all, so $YZ$ is definitely in product with $A\tA S\tS E^\A$ in $\rho$. The unitary that produced $\sigma$ from $\rho$ does not touch the registers $X\tX A\tA E^\A$ at all, and only uses $U\tU S\tS$ as control registers. Therefore, there are no correlations between $YZ$ and these registers conditioned on $df$ in $\sigma$ either. The product structure obviously also holds true if we trace out the $X_{\oT}\tX_{\oT}U_{\oT}\tU_{\oT}$ registers. Therefore, we can apply Quantum Raz's lemma (Fact \ref{fc:qraz}) to the above bound to say
\begin{align*}
\delta l & \geq \bbE_{\sfP_{X_TU_TY_TZ_TA_TS_TB_TC_TDF|\clE}}\sfD\Big(\vph_{Y_{\oT }Z_{\oT}A_{\oT}\tA_{\oT}S_{\oT}\tS_{\oT}|x_Tu_Ty_Tz_Ta_Ts_Tb_Tc_Tdf}\Big\Vert \sigma_{Y_{\oT }Z_{\oT}A_{\oT}\tA_{\oT} S_{\oT}\tS_{\oT}|x_Tu_Ty_Tz_Tdf}\Big) \\
 & \geq \sum_{i\in\oT}\sfI(Y_iZ_i:A_{\oT}\tA_{\oT} S_{\oT}\tS_{\oT}|X_TU_TY_TZ_TA_TS_TB_TC_TDF)_\vph \\
 & \geq \frac{l}{2}\bbE_{i\in\oT }\bbE_{\sfP_{D_iF_iR_i|\clE}}\sfI(Y_iZ_i:A_{\oT}\tA_{\oT} S_{\oT}\tS_{\oT})_{\vph_{d_if_ir_i}} \\
& \geq \frac{l}{2}\cdot\frac{1}{2}\bbE_{i\in\oT }\bbE_{\sfP_{X_iU_iR_i|\clE}}\sfI(Y_iZ_i:A_{\oT}\tA_{\oT} S_{\oT}\tS_{\oT})_{\vph_{x_iu_ir_i}} \\
& = \frac{l}{4}\bbE_{i\in\oT }\bbE_{\sfP_{X_iU_iY_iZ_iR_i|\clE}} \sfD\left(\vph_{A_{\oT}\tA_{\oT} S_{\oT}\tS_{\oT}|x_iy_iz_ir_i} \middle\Vert \vph_{A_{\oT}\tA_{\oT} S_{\oT}\tS_{\oT}|x_iu_ir_i}\right) \\
& \geq \frac{l}{4}\bbE_{i\in\oT }\bbE_{\sfP_{X_iY_iZ_iR_i|\clE}} \sfB\left(\vph_{A_{\oT}\tA_{\oT} S_{\oT}\tS_{\oT}|x_iu_iy_iz_ir_i}, \vph_{A_{\oT}\tA_{\oT} S_{\oT}\tS_{\oT}|x_iu_ir_i}\right)^2,
\end{align*}
where we have used Pinsker's inequality in the last step. Using Jensen's inequality on the above, we then have,
\begin{equation}\label{eq:AS-close}
\bbE_{i\in\oT }\bbE_{\sfP_{X_iY_iZ_iR_i|\clE}} \sfB\left(\vph_{A_{\oT}\tA_{\oT} S_{\oT}\tS_{\oT}|x_iy_iz_ir_i}, \vph_{A_{\oT}\tA_{\oT} S_{\oT}\tS_{\oT}|x_iu_ir_i}\right) \leq 2\sqrt{\delta}.
\end{equation}
Now \eqref{eq:XYZ} implies
\[
\bbE_{i\in\oT}\norm{\sfP_{X_iU_iY_iZR_i|\clE} - \sfP_{Y_iZ_i}\sfP_{X_iU_iR_i|\clE,Y_iZ_i}}_1 = \bbE_{i\in\oT}\norm{\sfP_{Y_iZ_i|\clE} - \sfP_{Y_iZ_i}}_1 \leq \sqrt{2\delta}.
\]
Therefore,
\begin{align*}
& \bbE_{i\in\oT }\bbE_{\sfP_{Y_iZ_i}\sfP_{X_iR_i|\clE,Y_iZ_i}} \sfB\left(\vph_{A_{\oT}\tA_{\oT} S_{\oT}\tS_{\oT}|x_iu_iy_iz_ir_i}, \vph_{A_{\oT}\tA_{\oT} S_{\oT}\tS_{\oT}|x_iu_ir_i}\right) \\
\leq & \bbE_{i\in\oT }\bbE_{\sfP_{X_iU_iR_i|\clE,Y_iZ_i}} \sfB\left(\vph_{A_{\oT}\tA_{\oT} S_{\oT}\tS_{\oT}|x_iu_iy_iz_ir_i}, \vph_{A_{\oT}\tA_{\oT} S_{\oT}\tS_{\oT}|x_iu_ir_i}\right) + \bbE_{i\in\oT}\norm{\sfP_{X_iU_iY_iZR_i|\clE} - \sfP_{Y_iZ_i}\sfP_{X_iU_iR_i|\clE,Y_iZ_i}}_1 \\
\leq & 4\sqrt{\delta}.
\end{align*}
Since $\sfP_{Y_iZ_i}(\perp,\perp)=\alpha^2$, we have,
\begin{equation}\label{eq:AS-perp-close}
\bbE_{i\in\oT }\bbE_{\sfP_{X_iU_iR_i|\clE,\perp,\perp}} \sfB\left(\vph_{A_{\oT}\tA_{\oT} S_{\oT}\tS_{\oT}|x_iu_i\perp \perp r_i}, \vph_{A_{\oT}\tA_{\oT} S_{\oT}\tS_{\oT}|x_iu_i r_i}\right) \leq \frac{4\sqrt{\delta}}{\alpha^2}.
\end{equation}

Applying the Fuchs-van de Graaf inequality to \eqref{eq:AS-close} and \eqref{eq:AS-perp-close}, and tracing out registers besides the $i$-th one we get,
\begin{gather}
\bbE_{i\in\oT}\norm{\sfP_{X_iU_iY_iZ_iR_i|\clE}(\sfP_{A_iS_i|\clE, X_iU_iY_iZ_iR_i} - \sfP_{A_iS_i|\clE,X_iU_iR_i})}_1 \leq 4\sqrt{2\delta}, \label{eq:AS-l1} \\
\bbE_{i\in\oT}\norm{\sfP_{X_iU_iR_i|\clE,\perp\perp}(\sfP_{A_iS_i|\clE, X_iU_i\perp\perp,R_i} - \sfP_{A_iS_i|\clE,X_iU_iR_i})}_1 \leq \frac{8\sqrt{2\delta}}{\alpha^2}. \label{eq:AS-perp-l1}
\end{gather}
Finally, we can apply item (i) of Fact \ref{fc:anchor-t*} with $X_iU_i=S$, $Y_iZ_i=T$, $R_i=R$, and the variables conditioned on $\clE$ being the corresponding primed variables, since $\sfP_{X_iU_iY_iZ_i}(x_i,u_i,\perp,\perp) = \alpha^2\cdot\sfP_{X_iU_i}(x_i,u_i)$ for all $x_i,u_i$. Using \eqref{eq:XYZ} and \eqref{eq:XYZR-1}, this gives us
\begin{align}
& \bbE_{i\in\oT}\norm{\sfP_{X_iU_iR_i|\clE} - \sfP_{X_iU_iR_i|\clE,\perp\perp}}_1 \nonumber \\
\leq & \frac{2}{\alpha^2} \bbE_{i\in\oT} \norm{ \sfP_{X_iU_iY_iZ_iR_i|\clE} - \sfP_{X_iU_iR_i|\clE} \sfP_{Y_iZ_i|X_iU_i} }_1 + 
\frac{5}{\alpha^2} \bbE_{i\in\oT} \norm{ \sfP_{X_iU_iY_iZ_i|\clE} - \sfP_{X_iU_iY_iZ_i} }_1 \nonumber \\
\leq & O\left(\frac{\sqrt{\delta}}{\alpha^2}\right). \label{eq:XYR-perp}
\end{align}
Applying the triangle inequality to the above, as well as \eqref{eq:AS-l1} and \eqref{eq:AS-perp-l1} we get,
\begin{align*}
& \bbE_{i\in\oT}\norm{\sfP_{X_iU_iY_iZ_iR_i|\clE}\left(\sfP_{A_iS_i|\clE,X_iU_iY_iZ_iR_i} - \sfP_{A_iS_i|\clE,X_iU_i,\perp\perp,R_i}\right)}_1 \\
\leq & \bbE_{i\in\oT}\Big(\norm{\sfP_{X_iU_iY_iZ_iR_i|\clE}(\sfP_{A_iS_i|\clE, X_iU_iY_iZ_iR_i} - \sfP_{A_iS_i|\clE,X_iU_iR_i})}_1 \\
& \quad + \norm{\sfP_{X_iU_iR_i|\clE,\perp\perp}(\sfP_{A_iS_i|\clE, X_iU_i\perp\perp,R_i} - \sfP_{A_iS_i|\clE,X_iU_iR_i})}_1 + 2\norm{\sfP_{X_iU_iR_i|\clE} - \sfP_{X_iU_iR_i|\clE,\perp\perp}}_1\Big) \\
\leq & O\left(\frac{\sqrt{\delta}}{\alpha^2}\right).
\end{align*}
This shows item (iii) of the lemma.

For later calculations, we note that the state $\sigma_{YZ\tZ A\tA S\tS|df}$ is also product across $Y$ and the rest of the registers. Therefore, it is possible to get bounds on $\bbE_{i\in\oT}\bbE_{\sfP_{D_iF_iR_i}}\sfI(Y_i:Z_{\oT}\tZ_{\oT}A_{\oT}\tA_{\oT})_{\vph_{d_if_ir_i}}$ in the exact same way. Doing the same calculation as above with this quantity and conditioning $Y_i=\perp$ (which happens with probability $\alpha$ under $\sfP_{Y_i}$), we can get
\begin{equation}\label{eq:ASZ-perp}
\bbE_{i\in\oT}\norm{\sfP_{X_iU_iY_iR_i|\clE}(\sfP_{Z_iA_iS_i|\clE,X_iU_iY_iR_i} - \sfP_{Z_iA_iS_i|\clE,X_iU_i,Y_i=\perp,R_i})}_1 \leq O\left(\frac{\sqrt{\delta}}{\alpha}\right).
\end{equation}
Note that we have had to apply Fact \ref{fc:anchor-t*} with $X_iU_iZ_i=S, Y_i=T$ and $R_i=R$ to get the above inequality, which is possible because $\sfP_{X_iU_iY_iZ_i}(x_i,u_i,\perp,z_i) = \alpha\cdot\sfP_{X_iU_iZ_i}(x_i,u_i,z_i)$ for all $x_i,u_i, z_i$. Similarly, on Charlie's side we have,
\begin{equation}\label{eq:ASY-perp}
\bbE_{i\in\oT}\norm{\sfP_{X_iU_iZ_iR_i|\clE}(\sfP_{Y_iA_iS_i|\clE,X_iU_iZ_iR_i} - \sfP_{Y_iA_iS_i|\clE,X_iU_i,Z_i=\perp,R_i})}_1 \leq O\left(\frac{\sqrt{\delta}}{\alpha}\right).
\end{equation}

\vspace{0.5cm}
\textbf{Existence of unitaries $V^{\A}_{i,x_ir_i}$ and $V^{\B \C}_{i,u_ir_i}$.} The proof of item (ii) of the lemma will be very similar to the analogous step in the proof of the parallel repetition theorem in \cite{KT23}.
Applying Fact \ref{fc:jpy-cond} on the states $\vph_{XU\tU Y\tY Z\tZ S_{\oT}\tS_{\oT}B_{\oT }C_{\oT}E^\B E^\C|a_Ts_Tb_Tc_Tdf}$ and $\sigma_{XU\tU Y\tY Z\tZ S_{\oT}\tS_{\oT}B_{\oT }\tC_{\oT }E^\B E^\C|df}$ with $z$ being $a_Ts_Tb_Tc_T$ as before, we get,
\begin{align*}
& \bbE_{\sfP_{X_TU_TY_TZ_TA_TS_TB_TC_TDF|\clE}}\sfD\Big(\vph_{X_{\oT }U_{\oT}\tU_{\oT}Y_{\oT }\tY_{\oT }Z_{\oT }\tZ_{\oT } S_{\oT}\tS_{\oT}B_{\oT }C_{\oT}E^\B E^\C|x_Tu_Ty_Tz_Ta_Ts_Tb_Tc_Tdf}\Big\Vert\sigma_{X_{\oT}U_{\oT}\tU_{\oT}Y_{\oT }\tY_{\oT }Z_{\oT }\tZ_{\oT } S_{\oT}\tS_{\oT}B_{\oT }\tC_{\oT }E^\B E^\C|x_Tu_Ty_Tz_Tdf}\Big) \\
\leq & \bbE_{\sfP_{A_TS_TB_TC_TDF|\clE}}\sfD\Big(\vph_{XU\tU Y\tY Z\tZ S_{\oT}\tS_{\oT}B_{\oT }C_{\oT}E^\B E^\C|a_Ts_Tb_Tc_Tdf}\Big\Vert \sigma_{XU\tU Y\tY Z\tZ S_{\oT}\tS_{\oT}B_{\oT }\tC_{\oT }E^\B E^\C|df}\Big) \\
\leq & \bbE_{\sfP_{DF|\clE}}(\log(1/\gamma_{df}) + \log(|\clA|^{|T|}\cdot|\clS|^{|T|}\cdot|\clB|^{|T|}\cdot|\clC|^{|T|})) \\
\leq & \delta l.
\end{align*}
Notice that $\sigma_{X_{\oT}U_{\oT}\tU_{\oT}Y_{\oT }\tY_{\oT }Z_{\oT }\tZ_{\oT } S_{\oT}\tS_{\oT}B_{\oT }\tC_{\oT }E^\B E^\C|x_Tu_Ty_Tz_Tdf}$ is product across $X_{\oT}$ and the rest of the registers. This is because conditioned on $df$, $X_{\oT}$ is in product with the rest of the input registers, and the registers $S_{\oT}\tS_{\oT}B_{\oT }\tC_{\oT }E^\B E^\C$ are acted upon by unitaries that are conditioned on the rest of the input registers only. Henceforth, we shall use $\tE^{\B\C}$ to refer to the registers $U_{\oT}\tU_{\oT}Y_{\oT }\tY_{\oT }Z_{\oT }\tZ_{\oT } S_{\oT}\tS_{\oT}B_{\oT }C_{\oT}E^\B E^\C$ for brevity. We can apply Quantum Raz's Lemma (Fact \ref{fc:qraz}) to say as before,
\begin{align*}
\delta l & \geq \bbE_{\sfP_{X_TU_TY_TZ_TA_TS_TB_TC_TDF|\clE}}\sfD\Big(\vph_{X_{\oT }\tE^{\B\C}|x_Tu_Ty_Tz_Ta_Ts_Tb_Tc_Tdf}\Big\Vert \sigma_{X_{\oT}\tE^{\B\C}|x_Tu_Ty_Tz_Tdf}\Big) \\
& \geq \sum_{i\in\oT }\sfI(X_i:\tE^{\B\C}|X_TU_TY_TZ_TA_TS_TB_TC_TDF)_\vph \\
& \geq \frac{l}{4}\bbE_{i\in\oT }\bbE_{\sfP_{X_iY_iZ_iR_i|\clE}} \sfB\left(\vph_{\tE^{\B\C}|x_iy_iz_ir_i}, \vph_{\tE^{\B\C}|y_iz_ir_i}\right)^2.
\end{align*}
Using Jensen's inequality on the above, we then have,
\[ \bbE_{i\in\oT }\bbE_{\sfP_{X_iY_iZ_iR_i|\clE}} \sfB\left(\vph_{\tE^{\B\C}|x_iy_iz_ir_i}, \vph_{\tE^{\B\C}|y_iz_ir_i}\right) \leq 2\sqrt{\delta}.\]
Shifting the expectation from $\sfP_{X_iY_iZ_iR_i|\clE}$ to $\sfP_{Y_iZ_i}\sfP_{X_iU_iR_i|\clE,Y_iZ_i}$ and conditioning on $Y_iZ_i=(\perp,\perp)$ as before, we then have,
\[ \bbE_{i\in\oT }\bbE_{\sfP_{X_iR_i|\clE,\perp,\perp}} \sfB\left(\vph_{\tE^{\B\C}|x_i\perp \perp r_i}, \vph_{\tE^{\B\C}|\perp \perp r_i}\right) \leq \frac{4\sqrt{\delta}}{\alpha^2}.\]
Finally, since $\ket{\vph}_{x_i\perp\perp r_i}$ and $\ket{\vph}_{\perp\perp r_i}$ are purifications of $\vph_{\tE^{\B\C}|x_i\perp\perp r_i}$ and $\vph_{\tE^{\B\C}|\perp\perp r_i}$, by Uhlmann's theorem and the Fuchs-van de Graaf inequality, there exist unitaries $V^{\A}_{i,x_ir_i}$ acting on registers outside of $\tE^{\B\C}$, i.e., on $X_{\oT}\tX_{\oT}E^\A A_{\oT}\tA_{\oT}$ (the values in other registers are fixed by $r_i$) such that
\begin{equation}\label{eq:X-unitary}
\bbE_{i\in\oT}\bbE_{\sfP_{X_iR_i|\clE,\perp\perp}}\norm{V^{\A}_{i,x_ir_i}\otimes\Id\state{\vph}_{\perp\perp r_i}(V^{\A}_{i,x_ir_i})^\dagger\otimes\Id - \state{\vph}_{x_i\perp\perp r_i}}_1 \leq \frac{8\sqrt{2\delta}}{\alpha^2}.
\end{equation}

By the same analysis on Barlie's side, we can say that there exist unitaries $V^{\B\C}_{i,u_ir_i}$ acting on $U_{\oT}\tU_{\oT}Y_{\oT}\tY_{\oT}Z_{\oT}\tZ_{\oT}S_{\oT}\tS_{\oT}B_{\oT}C_{\oT}E^\B E^\C$ such that 
\begin{equation}\label{eq:U-unitary}
\bbE_{i\in\oT}\bbE_{\sfP_{U_iR_i|\clE,\perp\perp}}\norm{\Id\otimes V^{\B\C}_{i,u_ir_i}\state{\vph}_{\perp\perp r_i}\Id\otimes(V^{\B\C}_{i,u_ir_i})^\dagger - \state{\vph}_{u_i\perp\perp r_i}}_1 \leq \frac{8\sqrt{2\delta}}{\alpha^2}.
\end{equation}
Now, if $\clO_{X_i}$ is the channel that measures the $X_i$ register and records the outcome, this commutes with the $V^{\B\C}_{i,u_ir_i}$ unitaries. So,
\begin{align*}
\clO_{X_i}\left(\Id\otimes V^{\B\C}_{i,u_ir_i}\state{\vph}_{\perp\perp r_i}\Id\otimes(V^{\B\C}_{i,u_ir_i})^\dagger\right) & = \bbE_{\sfP_{X_i|\clE,\perp \perp r_i}}\state{x_i}\otimes\left(\Id\otimes V^{\B\C}_{i,u_ir_i}\state{\vph}_{x_i\perp\perp r_i}\Id\otimes(V^{\B\C}_{i,u_ir_i})^\dagger\right) \\
\clO_{X_i}\left(\state{\vph}_{u_i\perp\perp r_i}\right) & = \bbE_{\sfP_{X_i|\clE,u_i\perp\perp r_i}}\state{x_i}\otimes\state{\vph}_{x_iu_i\perp\perp r_i}.
\end{align*}
Therefore, applying Fact \ref{fc:l1-dec} to \eqref{eq:U-unitary} with the $\clO_{X_i}$ channel we get,
\begin{align*}
& \bbE_{i\in\oT }\bbE_{\sfP_{U_iR_i|\clE,\perp\perp}}\norm{\bbE_{\sfP_{X_i|\clE,\perp\perp r_i}}\state{x_i}\otimes\left(\Id\otimes V^{\B\C}_{i,u_ir_i}\state{\vph}_{x_i\perp\perp r_i}\Id\otimes(V^{\B\C}_{i,u_ir_i})^\dagger\right) - \bbE_{\sfP_{X_i|\clE,u_i\perp\perp r_i}}\state{x_i}\otimes\state{\vph}_{x_iy_i\perp r_i}}_1 \\
\leq & \frac{8\sqrt{2\delta}}{\alpha^2}.
\end{align*}
Therefore,
\begin{align}
& \bbE_{i\in\oT}\bbE_{\sfP_{X_iU_iR_i|\clE,\perp\perp}}\norm{\Id\otimes V^{\B\C}_{i,u_ir_i}\state{\vph}_{x_i\perp\perp r_i}\Id\otimes(V^{\B\C}_{i,u_ir_i})^\dagger - \state{\vph}_{x_iu_i\perp\perp r_i}}_1 \nonumber \\
\leq & \frac{8\sqrt{2\delta}}{\alpha^2} + 2\bbE_{i\in\oC}\norm{\sfP_{X_iU_iR_i|\clE,\perp\perp } - \sfP_{R_i|\clE,\perp\perp }\sfP_{X_i|\clE,\perp\perp ,R_i}\sfP_{U_i|\clE,\perp\perp,R_i}}_1. \label{eq:U-unitary-X}
\end{align}
Combining \eqref{eq:X-unitary} and \eqref{eq:U-unitary-X} we get,
\begin{align}
& \bbE_{i\in\oT}\bbE_{\sfP_{X_iU_iR_i|\clE,\perp\perp }}\norm{V^\A_{i,x_ir_i}\otimes V^{\B\C}_{i,u_ir_i}\state{\vph}_{\perp\perp r_i}(V^\A_{i,x_ir_i})^\dagger\otimes(V^{\B\C}_{i,u_ir_i})^\dagger - \state{\vph}_{x_iu_i\perp\perp r_i}}_1 \nonumber \\
\leq & \bbE_{i\in\oT}\bbE_{\sfP_{X_iU_iR_i|\clE,\perp\perp }}\norm{\Id\otimes V^{\B\C}_{i,u_ir_i}\left(V^\A_{i,x_ir_i}\otimes\Id\state{\vph}_{\perp\perp r_i}(V^\A_{i,x_ir_i})^\dagger\otimes\Id - \state{\vph}_{x_i\perp\perp r_i}\right)\Id\otimes (V^{\B\C}_{i,u_ir_i})^\dagger}_1  \nonumber \\
& \quad + \bbE_{i\in\oT}\bbE_{\sfP_{X_iU_iR_i|\clE,\perp\perp }}\norm{\Id\otimes V^{\B\C}_{i,u_ir_i}\state{\vph}_{x_i\perp\perp r_i}\Id\otimes(V^{\B\C}_{i,u_ir_i})^\dagger - \state{\vph}_{x_iu_i\perp\perp r_i}}_1 \nonumber \\
\leq & \frac{8\sqrt{2\delta}}{\alpha^2} + \frac{8\sqrt{2\delta}}{\alpha^2} + 2\bbE_{i\in\oT}\norm{\sfP_{X_iU_iR_i|\clE,\perp\perp } - \sfP_{R_i|\clE,\perp\perp }\sfP_{X_i|\clE,\perp\perp ,R_i}\sfP_{U_i|\clE,\perp\perp,R_i}}_1. \label{eq:XU-unitary-1}
\end{align}

By Fact \ref{fc:cond-prob} we have,
\begin{align*}
\bbE_{i\in\oT}\norm{\sfP_{X_iU_iR_i|\clE,\perp\perp } - \sfP_{X_iU_i|\perp\perp }\sfP_{R_i|\clE,\perp\perp }}_1 & \leq \frac{2}{\sfP_{Y_iZ_i}(\perp,\perp )}\bbE_{i\in\oT}\norm{\sfP_{X_iU_iY_iZ_iR_i|\clE} - \sfP_{X_iU_i|Y_iZ_i}\sfP_{Y_iZ_iR_i|\clE}}_1 \\
 & \leq \frac{4\sqrt{2\delta}}{\alpha^2},
\end{align*}
where we have used \eqref{eq:XYZR-1} in the last line. Noting that $\sfP_{X_iU_i|\perp\perp}=\sfP_{X_iU_i}=\sfP_{X_i}\sfP_{U_i}=\sfP_{X_i|\perp\perp}\sfP_{U_i|\perp\perp}$, we then have,
\begin{align*}
& \bbE_{i\in\oT}\norm{\sfP_{X_iU_iR_i|\clE,\perp\perp } - \sfP_{R_i|\clE,\perp\perp }\sfP_{X_i|\clE,\perp\perp ,R_i}\sfP_{U_i|\clE,\perp\perp,R_i}}_1 \\
\leq & \bbE_{i\in\oT}\Big(\left\Vert\sfP_{X_iU_iR_i|\clE,\perp\perp } - \sfP_{X_iU_i|\perp\perp }\sfP_{R_i|\clE,\perp\perp }\right\Vert_1 + \left\Vert(\sfP_{X_i|\perp\perp }\sfP_{R_i|\clE,\perp\perp } - \sfP_{X_iR_i|\clE,\perp\perp })\sfP_{U_i|\perp\perp }\right\Vert_1 \\
& \quad + \left\Vert(\sfP_{U_i|\perp\perp }\sfP_{R_i|\clE,\perp\perp } - \sfP_{U_iR_i|\clE,\perp\perp })\sfP_{X_i|\clE,\perp\perp ,R_i}\right\Vert_1\Big) \\
\leq & \bbE_{i\in\oT}\Big(\left\Vert\sfP_{X_iU_iR_i|\clE,\perp\perp } - \sfP_{X_iU_i|\perp\perp }\sfP_{R_i|\clE,\perp\perp }\right\Vert_1 + \left\Vert\sfP_{X_i|\perp\perp }\sfP_{R_i|\clE,\perp\perp } - \sfP_{X_iR_i|\clE,\perp\perp }\right\Vert_1 \\
& \quad + \left\Vert\sfP_{U_i|\perp\perp }\sfP_{R_i|\clE,\perp\perp } - \sfP_{U_iR_i|\clE,\perp\perp }\right\Vert_1\Big) \\
\leq & 3\bbE_{i\in\oT}\norm{\sfP_{X_iU_iR_i|\clE,\perp\perp } - \sfP_{X_iU_i|\perp\perp }\sfP_{R_i|\clE,\perp\perp }}_1 \leq \frac{12\sqrt{2\delta}}{\alpha^2}.
\end{align*}
Putting the above in \eqref{eq:XU-unitary-1} we get,
\begin{align}
& \bbE_{i\in\oT}\bbE_{\sfP_{X_iU_iR_i|\clE,\perp\perp }}\norm{V^\A_{i,x_ir_i}\otimes V^{\B\C}_{i,u_ir_i}\state{\vph}_{\perp\perp r_i}(V^\A_{i,x_ir_i})^\dagger\otimes(V^{\B\C}_{i,u_ir_i})^\dagger - \state{\vph}_{x_iu_i\perp\perp r_i}}_1 \nonumber \\
\leq & O\left(\frac{\sqrt{\delta}}{\alpha^2}\right). \label{eq:XU-unitary-2}
\end{align}
Finally, from \eqref{eq:XYR-perp} and \eqref{eq:XU-unitary-2} we get,
\begin{align*}
& \bbE_{i\in\oT}\bbE_{\sfP_{X_iU_iR_i|\clE}}\norm{V^\A_{i,x_ir_i}\otimes V^{\B\C}_{i,u_ir_i}\state{\vph}_{\perp\perp r_i}(V^\A_{i,x_ir_i})^\dagger\otimes(V^{\B\C}_{i,u_ir_i})^\dagger - \state{\vph}_{x_iu_i\perp\perp r_i}}_1 \\
\leq & \bbE_{i\in\oT}\bbE_{\sfP_{X_iU_iR_i|\clE,\perp\perp }}\norm{V^\A_{i,x_ir_i}\otimes V^{\B\C}_{i,u_ir_i}\state{\vph}_{\perp\perp r_i}(V^\A_{i,x_ir_i})^\dagger\otimes(V^{\B\C}_{i,u_ir_i})^\dagger - \state{\vph}_{x_iu_i\perp\perp r_i}}_1 \\
& \quad + \bbE_{i\in\oT}\norm{\sfP_{X_iU_iR_i|\clE} - \sfP_{X_iU_iR_i|\clE,\perp\perp}}_1 \\
\leq & O\left(\frac{\sqrt{\delta}}{\alpha^2}\right).
\end{align*}
This completes the proof of item (ii) of the lemma.

\vspace{0.5cm}
\textbf{Existence of unitaries $V^\B_{i,y_ir_i}$ and $V^\C_{i,z_ir_i}$.} We observe that $\sigma_{Y_{\oT}X_{\oT}\tX_{\oT}\tU_{\oT}\tS_{\oT}Z_{\oT}\tZ_{\oT}A_{\oT}\tA_{\oT}C_{\oT}E^\A E^\C|x_Tu_Ty_Tz_Tdf}$ is product across $Y_{\oT}$ and the rest of the registers. $Y_{\oT}$ is in product with Alice's registers conditioned on $x_Tu_Ty_Tz_Tdf$ here because that was the case in the state $\rho$ at the end of the first round; Alice's registers don't change from the first round, and $Y_{\oT}$ is only acted upon as a control register to get $\sigma$ from $\rho$. $Y_{\oT}$ was also in product with Barlie's input and output registers as well as Charlie's registers in $\rho$ conditioned on $x_tu_Ty_Tz_Tdf$, and remain so after Bob and Charlie's unitaries.
Hence using $\tE^{\A\C}$ to denote the registers $X_{\oT}\tX_{\oT}\tU_{\oT}\tS_{\oT}Z_{\oT}\tZ_{\oT}A_{\oT}\tA_{\oT}C_{\oT}E^\A E^\C$ and applying Facts \ref{fc:jpy-cond} and \ref{fc:qraz} again on $\vph_{Y_{\oT}\tE^{\A\C}|x_Tu_Ty_Tz_Ta_ts_tb_tc_Tdf}$ and $\sigma_{Y_{\oT}\tE^{\A\C}|x_Tu_Ty_Tz_Tdf}$we get,

\begin{align*}
2\delta & \geq \bbE_{i\in\oT}\sfI(Y_i:\tE^{\A\C}|D_iF_iR_i)_\vph \\
& = \bbE_{i\in\oT}\bbE_{\sfP_{Y_iF_iR_i|\clE}}\sfD\Big(\vph_{\tE^{\A\C}|y_id_if_ir_i} \Big\Vert\vph_{\tE^{\A\C}|d_if_ir_i}\Big) \\
& \geq \frac{1}{2}\bbE_{i\in\oT}\bbE_{\sfP_{X_iU_iY_iR_i|\clE}}\sfD\left(\vph_{\tE^{\A\C}|x_iu_iy_ir_i}\middle\Vert\vph_{\tE^{\A\C}|x_iu_ir_i}\right) \\
& \geq \frac{1}{2}\bbE_{i\in\oT}\bbE_{\sfP_{X_iU_iY_iR_i|\clE}}\sfB\left(\vph_{\tE^{\A\C}|x_iu_iy_ir_i}, \vph_{\tE^{\A\C}|x_iu_ir_i}\right)^2.
\end{align*}
Applying Jensen's inequality on the above, we get
\begin{align}
\bbE_{i\in\oT}\bbE_{\sfP_{X_iU_iY_iR_i|\clE}}\sfB\left(\vph_{\tE^{\A\C}|x_iu_iy_ir_i},\vph_{\tE^{\A\C}|x_iu_ir_i}\right) & \leq 2\sqrt{\delta}. \label{eq:XY-margin}
\end{align}
Moreover, shifting the expectation from $\sfP_{X_iU_iY_iR_i|\clE}$ to $\sfP_{Y_i}\sfP_{X_iU_iR_i|\clE,Y_i}$ and conditioning on $Y_i=\perp$ (which happens with probability $\alpha$ under $\sfP_{Y_i}$) like before, we get,
\begin{align}
 \bbE_{i\in\oT}\bbE_{\sfP_{X_iU_iR_i|\clE,Y_i=\perp }}\sfB\left(\vph_{\tE^{\A\C}|x_iu_i,Y_i=\perp r_i},\vph_{\tE^{\A\C}|x_iu_ir_i}\right) & \leq \frac{4\sqrt{\delta}}{\alpha}. \label{eq:XY-perp}
\end{align}
Using the triangle inequality on \eqref{eq:XY-margin} and \eqref{eq:XY-perp}, we get,
\begin{align*}
& \bbE_{i\in\oT}\bbE_{\sfP_{X_iU_iY_iR_i|\clE}}\sfB\left(\vph_{\tE^{\A\C}|x_iu_iy_ir_i},\vph_{\tE^{\A\C}|x_iu_i,Y_i=\perp, r_i}\right) \nonumber \\
\leq & \bbE_{i\in\oT}\bbE_{\sfP_{X_iU_iY_iR_i|\clE}}\sfB\left(\vph_{\tE^{\A\C}|x_iu_iy_ir_i},\vph_{\tE^{\A\C}|x_iu_ir_i}\right) \nonumber \\
& \quad + \bbE_{i\in\oT}\bbE_{\sfP_{X_iU_iR_i|\clE,Y_i=\perp }}\sfB\left(\vph_{\tE^{\A\C}|x_iu_i,Y_i=\perp, r_i},\vph_{\tE^{\A\C}|x_iu_ir_i}\right) \nonumber \\
& \quad + \bbE_{i\in\oT}\norm{(\sfP_{X_iU_iR_i|\clE}-\sfP_{X_iU_iR_i|\clE,Y_i=\perp})\sfP_{Y_i|\clE,X_iU_iR_i}}_1 \nonumber \\
\leq & 2\sqrt{\delta} + \frac{4\sqrt{\delta}}{\alpha} + \frac{2}{\alpha}\bbE_{i\in\oT}\norm{\sfP_{X_iU_iY_iR_i|\clE} - \sfP_{X_iU_iR_i|\clE}\sfP_{Y_i|X_iU_i}}_1 + \frac{5}{\alpha}\bbE_{i\in\oT}\norm{\sfP_{X_iU_iY_i|\clE}-\sfP_{X_iU_iY_i}}_1 \nonumber \\
\leq & 2\sqrt{\delta} + \frac{4\sqrt{\delta}}{\alpha} + \frac{4\sqrt{2\delta}}{\alpha} + \frac{5\sqrt{2\delta}}{\alpha} \\
\leq & O\left(\frac{\sqrt{\delta}}{\alpha}\right).
\end{align*}
In the third line of the above calculation, we have noted that we can apply item (i) of Fact \ref{fc:anchor-t*} to bound the distance $\norm{\sfP_{X_iU_iR_i|\clE}-\sfP_{X_iU_iR_i|\clE,Y_i=\perp}}_1$ with $T=Y_i$, since we have $\sfP_{X_iU_iY_i}(x_i,u_i,\perp)$ $=\alpha\cdot\sfP_{X_iU_i}(x_i,u_i)$ for all $x_i,u_i$. In the fourth line, we have used \eqref{eq:XYZR-1} and \eqref{eq:XYZ} to bound the trace distances. Finally, using Uhlmann's theorem and the Fuchs-van de Graaf inequality on the above, we get that there exist unitaries $V^\B_{i,x_iu_iy_ir_i}$ acting on the registers outside $\tE^{\A\C}$, i.e., on $Y_{\oT}\tY_{\oT}U_{\oT}S_{\oT}BE^\B$, such that
\begin{equation}\label{eq:Y-unitary-1}
\bbE_{i\in\oT}\bbE_{\sfP_{X_iU_iY_iR_i|\clE}}\norm{V^\B_{i,x_iu_iy_ir_i}\otimes\Id\state{\vph}_{x_iu_i,Y_i=\perp,r_i}(V^\B_{i,x_iu_iy_ir_i})^\dagger\otimes\Id - \state{\vph}_{x_iu_iy_ir_i}}_1 \leq O\left(\frac{\sqrt{\delta}}{\alpha}\right).
\end{equation}
Since $y_i$ is either $\perp$, in which case the unitary $V^\B_{i,x_iu_iy_ir_i}$ is just the identity, or $y_i$ is equal to $x_i$, $V^\B_{i,x_iu_iy_ir_i}$ is in fact just $V^\B_{i,u_iy_ir_i}$.

Let $\clO_{Z_iA_iS_i}$ be the channel which measures the $Z_iA_i\tS_i$ registers and records the outcome, which commutes with $V^\B_{i,u_iy_ir_i}$. We have,
\begin{align*}
& \clO_{Z_iA_iS_i}\left(V^\B_{i,u_iy_ir_i}\otimes\Id\state{\vph}_{x_iu_i,Y_i=\perp,r_i}(V^\B_{i,u_iy_ir_i})^\dagger\otimes\Id\right) \\
= & \bbE_{\sfP_{Z_iA_iS_i|\clE,x_iu_i,Y_i=\perp,r_i}}\state{z_ia_is_i}\otimes\left(V^\B_{i,u_iy_ir_i}\otimes\Id\state{\vph}_{x_iu_i,Y_i=\perp,z_ia_is_ir_i}(V^\B_{i,u_iy_ir_i})^\dagger\otimes\Id\right) \\
& \clO_{Z_iA_iS_i}\left(\state{\vph_{x_iu_iy_ir_i}}\right) = \bbE_{\sfP_{Z_iA_iS_i|\clE,x_iu_iy_ir_i}}\state{z_ia_is_i}\otimes\state{\vph}_{x_iu_iy_iz_ia_is_ir_i}.
\end{align*}
Applying Fact \ref{fc:chan-l1} on \eqref{eq:Y-unitary-1} we thus get,
\begin{align}
& \bbE_{i\in\oT}\bbE_{\sfP_{X_iU_iY_iZ_iA_iS_iR_i|\clE}}\norm{V^\B_{i,u_iy_ir_i}\otimes\Id\state{\vph}_{x_iu_i,Y_i=\perp,z_ia_is_ir_i}(V^\B_{i,u_iy_ir_i})^\dagger\otimes\Id - \state{\vph}_{x_iu_iy_iz_ia_is_ir_i}}_1 \nonumber \\
\leq & \bbE_{i\in\oT}\bbE_{\sfP_{X_iU_iY_iR_i|\clE}}\norm{V^\B_{i,u_iy_ir_i}\otimes\Id\state{\vph}_{x_iu_i,Y_i=\perp,r_i}(V^\B_{i,u_iy_ir_i})^\dagger\otimes\Id - \state{\vph}_{x_iu_iy_ir_i}}_1 \nonumber \\
& \quad + 2\bbE_{i\in\oT}\norm{\sfP_{X_iU_iY_iR_i|\clE}(\sfP_{Z_iA_iS_i|\clE,X_iU_iY_iR_i} - \sfP_{Z_iA_iS_i|\clE,X_iU_i,Y_i=\perp,R_i})}_1 \nonumber \\
\leq & O\left(\frac{\sqrt{\delta}}{\alpha}\right), \label{eq:Y-unitary-2}
\end{align}
where to bound the second term of the second line, we have used \eqref{eq:ASZ-perp}.

Repeating the same steps as above on Charlie's side and using \eqref{eq:ASY-perp} this time,  we get that there exist unitaries $V^\C_{i,u_iz_ir_i}$ acting on $Z_{\oT}\tZ_{\oT}\tU_{\oT}\tS_{\oT}CE^\C$ such that
\[ \bbE_{i\in\oT}\bbE_{\sfP_{X_iU_iY_iZ_iA_iS_iR_i|\clE}}\norm{\Id\otimes V^\C_{i,u_iz_ir_i}\state{\vph}_{x_iu_iy_i,Z_i=\perp,a_is_ir_i}\Id\otimes(V^\C_{i,u_iz_ir_i})^\dagger - \state{\vph}_{x_iu_iy_iz_ia_is_ir_i}}_1 \leq O\left(\frac{\sqrt{\delta}}{\alpha}\right).\]
We can use \eqref{eq:XYZ} again to shift the expectation in the above expression from $\sfP_{X_iU_iY_iZ_iR_i|\clE}$ to $\sfP_{Y_i}\sfP_{X_iU_iZ_iR_i|\clE,Y_i}$, with only a $\sqrt{2\delta}$ loss. We can then condition on $Y_i=\perp$ (which happens with probability $\alpha$ under $\sfP_{Y_i}$) to get,
\begin{align}
& \bbE_{i\in\oT}\bbE_{\sfP_{X_iU_iZ_iA_iS_iR_i|\clE,Y_i=\perp}}\norm{\Id\otimes V^\C_{i,u_iz_ir_i}\state{\vph}_{x_iu_i,\perp\perp,a_is_ir_i}\Id\otimes(V^\C_{i,u_iz_ir_i})^\dagger - \state{\vph}_{x_iu_i,Y_i=\perp,z_ia_is_ir_i}}_1 \nonumber \\
\leq & O\left(\frac{\sqrt{\delta}}{\alpha^2}\right). \label{eq:Z-unitary-1}
\end{align}
Now, tracing out the $Y_i$ register from \eqref{eq:ASZ-perp} we get,
\[ \bbE_{i\in\oT}\norm{\sfP_{X_iU_iR_i|\clE}(\sfP_{Z_iA_iS_i|\clE,X_iU_iR_i} - \sfP_{Z_iA_iS_i|\clE,X_iU_i,Y_i=\perp,R_i})}_1 \leq O\left(\frac{\sqrt{\delta}}{\alpha}\right),\]
and moreover, as we have seen before, we can use Fact \ref{fc:anchor-t*} to bound $\bbE_{i\in\oT}\norm{\sfP_{X_iU_iR_i|\clE} - \sfP_{X_iU_iR_i|\clE,Y_i=\perp}}_1$. Therefore,
\begin{align*}
\bbE_{i\in\oT}\norm{\sfP_{X_iU_iZ_iA_iS_iR_i|\clE,Y_i=\perp} - \sfP_{X_iU_iZ_iA_iS_iR_i|\clE}}_1 & \leq \bbE_{i\in\oT}\norm{(\sfP_{X_iU_iR_i|\clE} - \sfP_{X_iU_iR_i|\clE,Y_i=\perp})\sfP_{Z_iA_iS_i|\clE,X_iU_i,Y_i=\perp,R_i}}_1 \\
& \quad + \bbE_{i\in\oT}\norm{\sfP_{X_iU_iR_i|\clE}(\sfP_{Z_iA_iS_i|\clE,X_iU_iR_i} - \sfP_{Z_iA_iS_i|\clE,X_iU_i,Y_i=\perp,R_i})}_1 \\
& \leq O\left(\frac{\sqrt{\delta}}{\alpha}\right).
\end{align*}
Using this along with \eqref{eq:Z-unitary-1} we get,
\begin{align}
& \bbE_{i\in\oT}\bbE_{\sfP_{X_iU_iZ_iA_iS_iR_i|\clE}}\norm{\Id\otimes V^\C_{i,u_iz_ir_i}\state{\vph}_{x_iu_i,\perp\perp,a_is_ir_i}\Id\otimes(V^\C_{i,u_iz_ir_i})^\dagger - \state{\vph}_{x_iu_i,Y_i=\perp,z_ia_is_ir_i}}_1 \nonumber \\
\leq & O\left(\frac{\sqrt{\delta}}{\alpha^2}\right). \label{eq:Z-unitary-2}
\end{align}
Finally, combining \eqref{eq:Y-unitary-2} and \eqref{eq:Z-unitary-2} we get,
\begin{align*}
& \bbE_{i\in\oT}\bbE_{\sfP_{X_iU_iY_iZ_iA_iS_iR_i|\clE}}\norm{V^\B_{i,u_iy_ir_i}\otimes V^\C_{i,u_iz_ir_i}\state{\vph}_{x_iu_i,\perp\perp,a_is_ir_i}(V^\B_{i,u_iy_ir_i})^\dagger\otimes(V^\C_{i,u_iz_ir_i})^\dagger - \state{\vph}_{x_iu_iy_iz_ia_is_ir_i}}_1 \\
\leq & \bbE_{i\in\oT}\bbE_{\sfP_{X_iU_iZ_iA_iS_iR_i|\clE}}\norm{V^\B_{i,u_iy_ir_i}\left(\Id\otimes V^\C_{i,u_iz_ir_i}\state{\vph}_{x_iu_i,\perp\perp,a_is_ir_i}\Id\otimes(V^\C_{i,u_iz_ir_i})^\dagger - \state{\vph}_{x_iu_i,Y_i=\perp,z_ia_is_ir_i}\right)V^\B_{i,u_iy_ir_i}}_1 \\
& \quad + \bbE_{i\in\oT}\bbE_{\sfP_{X_iU_iY_iZ_iA_iS_iR_i|\clE}}\norm{V^\B_{i,u_iy_ir_i}\otimes\Id\state{\vph}_{x_iu_i,Y_i=\perp,z_ia_is_ir_i}(V^\B_{i,u_iy_ir_i})^\dagger\otimes\Id - \state{\vph}_{x_iu_iy_iz_ia_is_ir_i}}_1 \\
\leq & O\left(\frac{\sqrt{\delta}}{\alpha^2}\right).
\end{align*}
This completes the proof of item (iv) of the lemma.

\section*{Acknowledgements}
We thank Eric Culf, Mads Friis Frand-Madsen, Rahul Jain, S{\'e}bastien Lord, and Joseph Renes for helpful discussions. 
We also thank a reviewer for pointing out the issue described in Remark~\ref{remark:termsplit} regarding previous security definitions.
S.~K.~is funded by the Natural Sciences and Engineering Research Council of Canada (NSERC) Discovery Grants Program and Fujitsu Labs America. E.~Y.-Z.~T.~is funded by NSERC Alliance and Huawei Technologies Canada Co., Ltd. Research at the Institute for Quantum Computing (IQC) is supported by Innovation, Science and Economic Development (ISED) Canada.

\bibliographystyle{alpha}
\bibliography{bib-DI-uncloneable} 

\end{document}